\documentclass[10pt,a4paper]{article}

\usepackage{hyperref}

\usepackage{doi}
\usepackage{cite}

\usepackage[fleqn]{amsmath}
\usepackage{amsfonts,amsthm,amssymb}

\usepackage{calc}

\usepackage{graphicx}

\usepackage{mathrsfs}
\usepackage{makeidx}
\usepackage{xypic}
\usepackage[nottoc]{tocbibind}
\usepackage{pstricks}
\usepackage{bbm}
\usepackage[arrow, matrix, curve]{xy}
\usepackage{epic}
\usepackage{eepic}
\usepackage{epsfig}
\usepackage{pdflscape} % pour changer en landscape en cours de document

\numberwithin{equation}{section}

\usepackage{enumitem}

\xymatrixcolsep{0.5cm}
\xymatrixrowsep{0.5cm}

\newcommand{\nobarfrac}{\genfrac{}{}{0pt}{}}

\newtheorem{theorem}{Theorem}[section]
\newtheorem{proposition}{Proposition}[section]

\newtheorem{corollary}{Corollary}[section]

\newtheorem{identity}{Identity}
\newtheorem{prop}{Proposition}[section]
\newtheorem{cor}{Corollary}[section]

{\theoremstyle{remark}
\newtheorem{rem}{\sl Remark}

\newcommand{\tr}{\operatorname{tr}}
\newcommand{\Res}{\operatorname{Res}}
\newcommand{\End}{\operatorname{End}}

\newcommand{\bra}[1]{\langle\,#1\,|}
\newcommand{\ket}[1]{|\,#1\,\rangle}

\newcommand{\moy}[1]{\langle\,#1\,\rangle}

\def\eps{\epsilon}

\newcommand{\mathsc}[1]{{\normalfont\textsc{#1}}}

\makeatletter
\newcommand{\vast}{\bBigg@{3}}
\newcommand{\Vast}{\bBigg@{4}}
\makeatother

\DeclareMathSymbol{\Alpha}{\mathalpha}{operators}{"41}
\DeclareMathSymbol{\Beta}{\mathalpha}{operators}{"42}
\DeclareMathSymbol{\Epsilon}{\mathalpha}{operators}{"45}
\DeclareMathSymbol{\Zeta}{\mathalpha}{operators}{"5A}
\DeclareMathSymbol{\Eta}{\mathalpha}{operators}{"48}
\DeclareMathSymbol{\Iota}{\mathalpha}{operators}{"49}
\DeclareMathSymbol{\Kappa}{\mathalpha}{operators}{"4B}
\DeclareMathSymbol{\Mu}{\mathalpha}{operators}{"4D}
\DeclareMathSymbol{\Nu}{\mathalpha}{operators}{"4E}
\DeclareMathSymbol{\Omicron}{\mathalpha}{operators}{"4F}
\DeclareMathSymbol{\Rho}{\mathalpha}{operators}{"50}
\DeclareMathSymbol{\Tau}{\mathalpha}{operators}{"54}
\DeclareMathSymbol{\Chi}{\mathalpha}{operators}{"58}
\DeclareMathSymbol{\omicron}{\mathord}{letters}{"6F}

\begin{document}

LPENSL-TH-08/26

%\vspace{45pt} \today

\bigskip

\bigskip

\begin{center}
%{\Large Title 1}

%{\Large Matrix elements of local operators for the most general \textquotedblleft homogeneous\textquotedblright\ XXZ spin 1/2 open chains}

%{\Large Title 2}

\textbf{\Large On the gauge invariant boundary condition in open XXZ  chains: SoV bases, elementary blocks and overlaps}
% with unparallel boundary fields} \vspace{45pt}
\end{center}

\begin{center}
{\large \textbf{G. Niccoli}\footnote{{Univ Lyon, Ens de Lyon, Univ
Claude Bernard, CNRS, Laboratoire de Physique, F-69342 Lyon, France;
giuliano.niccoli@ens-lyon.fr}} }, {\large \textbf{V. Terras}\footnote{{Université Paris-Saclay, CNRS,  LPTMS, 91405, Orsay, France; veronique.terras@universite-paris-saclay.fr}} }
\end{center}

\begin{center}
\vspace{45pt}

\today
\vspace{45pt}
\end{center}

\begin{abstract}
We consider the open XXZ spin chain with unparallel boundary fields, in the special case in which the Vertex-IRF transform of the K-matrix labelling the boundary field at site 1 is diagonal and invariant under the value of the gauge parameters, the other K-matrix being arbitrary. We discuss here several properties related to this special gauge invariant boundary condition. 
Such a model can be solved in the Vertex-IRF framework with arbitrary gauge parameters: we can diagonalise the transfer matrix by means of generalised Sklyanin's Separation of Variable (SoV) approach in a family of SoV bases that depend on two arbitrary gauge parameters $\alpha$ and $\beta$, $\beta$ being the so-called dynamical parameter and $\alpha$ an arbitrary shift of the spectral parameter.
We show here that such SoV bases depend actually on $\alpha$ and $\beta$ through the difference $\alpha-\beta$ only. Considering the ungauged limit in which $\alpha-\beta$ tends to infinity, we show how to compute all elementary building  blocks for the correlation functions of this model.   
We finally show how to compute overlaps %and matrix elements of local operator overlaps 
between eigenstates of this model and eigenstates of a spin chain sharing the same boundary condition at site $N$ and in which the gauge invariant boundary condition at site 1 has been changed to a more general one.

\end{abstract}

\begin{center}
\vspace{45pt}
\end{center}

\newpage

\tableofcontents

\newpage

\section{Introduction}

A large literature \cite
{AlcBBBQ87,Skl88,GhoZ94,JimKKKM95,JimKKMW95,SkoS95,KapS96,FanHSY96,Nep02,CaoLSW03,Nep04,NepR03,NepR04,deGE05,Bas06,BasK07,YanZ07,FilK11,KitKMNST07,KitKMNST08,FraSW08,CraRS10,CraRS11,FraGSW11,Pro11,Nic12,CaoYSW13b,FalN14,FalKN14,KitMN14,BelC13,Bel15,BelP15,BelP15b,AvaBGP15,BelP16,KitMNT17,MaiNP17,KitMNT18,GriDT19,BelPS21}
has been devoted to the analysis of the spectrum of integrable quantum chains with open boundary conditions. In particular, open XXZ spin-$1/2$ chains with non-parallel boundary magnetic fields provide a rich setting for exploring quantum integrability beyond the standard Bethe Ansatz within the Quantum Inverse Scattering Method (QISM) \cite{FadS78,FadST79,FadT79} and the reflection algebra representation theory \cite{Che84}. The integrable structure of these models—encoded in a commuting family of transfer matrices built from the six-vertex $R$-matrix and boundary $K$-matrices—was established for arbitrary boundary fields by Sklyanin \cite{Skl88}. However, the Algebraic Bethe Ansatz (ABA) can be applied in a straightforward way only when both boundary fields are parallel and oriented along the $z$-direction, here on referred as longitudinal boundary fields. In that case, the existence of a simple reference state and a conserved $S_z$ charge enables an ABA description \cite{AlcBBBQ87,Skl88} of the spectrum, and correlation functions were also derived in this ABA framework in \cite{KitKMNST07,KitKMNST08}. 

For genuinely unparallel boundary fields, the ABA approach usually breaks down. Nevertheless, under a particular constraint relating the two boundary fields (usually referred to as Nepomechie’s constraint), part of the spectrum can still be described by usual Bethe equations \cite{Nep02,Nep04,NepR03,NepR04}, and some of the eigenstates can still be constructed by means of a generalisation of the ABA relying on a trigonometric version of Baxter’s Vertex-IRF transformation \cite{Bax73a}, see the work \cite{CaoLSW03}, which specialises the XYZ results of \cite{FanHSY96} to the XXZ case, see also \cite{FilK11}. Further methods designed to mimic or generalise the ABA, such as the off-diagonal Bethe Ansatz \cite{CaoYSW13a,CaoYSW13b} and the modified Bethe Ansatz \cite{BelC13,Bel15,BelP15,BelP15b,AvaBGP15,BelP16}, have also been developed to tackle more general boundary fields beyond Nepomechie’s constraint. However, until very recently, the computation of correlation functions remained essentially restricted to longitudinal boundary-field configurations\footnote{A large body of exact results on correlation functions
\cite{JimM95L,JimMMN92,JimM96,KitMT99,KitMT00,KitMST02a,KitMST05a,KitMST05b,KitKMST07,GohKS04,GohKS05,BooGKS07,GohS10,DugGKS15,GohKKKS17,KitKMST09a,KitKMST09b,KitKMST09c,KozMS11a,KozMS11b,KozT11,KitKMST11a,KitKMST11b,KitKMST12,DugGK13,KitKMT14,CauHM05,CauM05,PerSCHMWA06,BooJMST05,BooJMST06,BooJMST06a,BooJMST06b,BooJMST07,BooJMST09,JimMS09,JimMS11,MesP14,Poz17,BabGKS21,KozT23}
was obtained mainly for periodic chains or for open chains with longitudinal boundary fields. In particular, in their q-vertex operator approach, the Kyoto group computed correlation functions for semi-infinite chains with one longitudinal boundary field \cite{JimKKKM95,JimKKMW95}, while the Lyon group derived them starting from the finite chains with both longitudinal fields in the ABA framework \cite{KitMT99,KitMT00,KitKMNST07,KitKMNST08}. These approaches do not extend naturally to non-longitudinal boundary conditions.}.

The quantum Separation of Variables (SoV) framework offers a structurally different approach. Initiated by Sklyanin \cite{Skl85,Skl85a,Skl90,Skl92,Skl95,Skl96} and later reformulated in a purely algebraic form in \cite{MaiN18,MaiN19c}, SoV is shown to rely only on the integrable structure of the model. It does not require a pseudo-vacuum and makes no ansatz on the eigenstates. Instead, it constructs a basis in which the transfer matrix acts in a separated form, reducing its action to simple shifts of the separated variables with coefficients depending only on the corresponding SoV coordinate. In this basis, the eigenstates admit a factorised representation of their wave-functions over the spectrum of the separated variables, and their characterisation reduces to solving finite-difference equations of Baxter type.

The SoV approach has proved powerful for a wide class of integrable models \cite{Skl85,Skl85a,Skl90,Skl92,Skl95,Skl96,BabBS96,Smi98a,Smi01,DerKM01,DerKM03,DerKM03b,BytT06,vonGIPS06,FraSW08,AmiFOW10,NicT10,Nic10a,Nic11,FraGSW11,GroMN12,GroN12,Nic12,Nic13,Nic13a,Nic13b,GroMN14,FalN14,FalKN14,KitMN14,NicT15,LevNT16,NicT16,KitMNT16,JiaKKS16,KitMNT17,MaiNP17,KitMNT18,MaiN18,MaiN19}, including open XXZ chains with generic boundary fields \cite{Nic12,FalKN14,KitMN14,KitMNT18,MaiN19}. A central SoV feature is the complete description of the spectrum and the universal determinant representations  
\cite{GroMN12,Nic12,Nic13,Nic13a,Nic13b,GroMN14,FalN14,FalKN14,LevNT16,KitMNT16,KitMNT17,MaiNP17,KitMNT18,NicT24}  
for scalar products of separate states, see also \cite{MaiNV20b} and \cite{CavGL19,GroLRV20} for higher-rank generalisations. 
In recent works \cite{Nic21,NicT22,NicT23,NicT25}, we have extended the SoV approach to compute correlation functions under the most general boundary conditions compatible with a description of the ground state in terms of a homogeneous $TQ$-equation, i.e. under Nepomechie's constraint. Our results were obtained within a generalisation of Sklyanin's SoV approach based on the Vertex-IRF transformation of the boundary Yang-Baxter algebra. However,  those results were obtained only for a specific subset of a particular basis of local operators, chosen so as to preserve a key structural property of separate states under their action. This limitation was due to technical difficulties related to the use of the Vertex-IRF transformation, and more specifically to the dependence of the SoV basis on two gauge parameters $\alpha$ and $\beta$, {\em a priori} fixed by the boundary conditions, but which may be shifted under the action of local operators.
%, and to some possible shifts of the latter gauge parameter $\beta$ under the action of local operators which induce

In the present paper, we focus on a special configuration of unparallel boundary fields which plays a particularly distinguished role in the aforementioned approach: whereas the boundary field at site $N$ is arbitrary, the boundary field at site 1 is chosen so  that the Vertex-IRF transformation of its associated K-matrix is proportional to the identity and is {\em gauge invariant}, in the sense that it is invariant with respect to a change of the  gauge parameters $\alpha$ and $\beta$ \cite{NicT22}.  As a consequence, the model can be solved within Sklyanin's SoV approach in a whole family of SoV bases, parametrised by the two arbitrary gauge parameters $\alpha$ and $\beta$, or more precisely, as we show here, by their difference $\alpha-\beta$. All these bases provide a separated representation of the eigenstates of the associated gauge invariant transfer matrix, and this flexibility makes the gauge invariant model a natural reference point for the analysis of elementary blocks and overlaps.

We here provide a complete study of the model with such particular boundary conditions. We construct the aforementioned family of SoV bases and show their dependence on $\alpha-\beta$ only. In each of these bases, the transfer matrix spectrum and eigenstates can be characterised in terms of $Q$-polynomial solutions of an  homogeneous $TQ$-equation. In the limit $\eta(\alpha-\beta)\to +\infty$ ($\eta$ being the crossing parameters of the model), we notably recover the ungauged SoV basis constructed in \cite{Nic12}.  We use this ungauged SoV basis to compute for this model  the complete set of elementary blocks for the correlation functions, both for the finite chain and in the thermodynamic limit, thereby overcoming—within this special configuration—the limitations present in our recent papers \cite{NicT22,NicT23,NicT25}.

A very interesting consequence of the aforementioned flexibility in the choice of the SoV basis % aspect of this model  
is that it is possible to compute overlaps between the eigenstates of this model, with this particular gauge invariant boundary condition at site 1, and the eigenstates of a spin chain in which the boundary field at site 1 has been changed to a more general non-longitudinal one, the (non-longitudinal) boundary field at site $N$ being unchanged. Computation of overlaps constitutes a crucial step in the study of non-equilibrium evolution of a model after a quantum quench \cite{CauE13}. This is usually a difficult problem,  even for simple integrable models for which other physical quantities  (such as e.g. correlation functions at equilibrium) can be exactly computed and for quenches which do not break integrability. The reason for that is that a quantum quench often involves a change of a parameter of the system (e.g. a coupling constant) which modifies the resolution algebra, so that the eigenstates of the model before and after the quench are constructed using different Yang-Baxter algebras. 

In the case of an open XXZ spin chain with boundary fields, a possible interesting quench protocol consists in changing one of the boundary magnetic fields. Although only local, such a perturbation of the model is worth studying since it may induce  macroscopic changes in the ground state, and even phase transitions \cite{GriDT19,PasLPAA23}. In the longitudinal case, %overlaps between eigenstates before and after such a boundary quench can be computed since 
Bethe states are constructed from a boundary monodromy matrix which involves only one of the boundary fields, so  that a change of the other boundary field does not modify their algebraic form. As a consequence, overlaps between eigenstates before and after such a boundary quench can be expressed through generalisations \cite{KitKMNST07,Wan02} of Slavnov determinants \cite{Sla89}, see \cite{AbeKT25}. The situation is however more complicated for non-longitudinal boundary fields. In the latter case, the construction of the eigenstates is performed either by means of the Vertex-IRF transformation \cite{CaoLSW03,FalKN14,KitMNT18}, or from SoV bases constructed from the whole transfer matrix \cite{MaiN19}. In both cases this construction depends on both boundary fields. In particular, in the Vertex-IRF framework, the transformed gauge boundary monodromy matrix constructed from the boundary K-matrix $K_-$ involves two gauge parameters $\alpha$ and $\beta$ which are themselves fixed in terms of the boundary parameters of the other boundary K-matrix $K_+$ (or conversely). 
%Hence, the SoV basis in which eigenstates are expressed, or equivalently the generalised Bethe states, depend also on the $K_+$ boundary parameters via these gauge parameters $\alpha$ and $\beta$. 
%depends on the Vertex–IRF transformation and on gauge parameters that themselves depend on the boundary fields. 
Hence, a boundary quench changing $K_+$ also induces a modification of these gauge parameters $\alpha$ and $\beta$, and therefore of the SoV bases in which the eigenstates are expressed, so that  the computation of the corresponding boundary overlaps is, in the generic non-longitudinal case, still an open problem.

As already mentioned, a central feature of the spin chain with the aforementioned gauge invariant boundary condition at site 1 is that it can be solved in a whole family of SoV bases parametrised by almost any value of $\alpha-\beta$. The latter correspond actually to the SoV bases of a family of spin chains with various non-diagonal matrices $K_+$, i.e. for various values of the non-longitudinal boundary field $h_+$ at site 1, in which $\alpha-\beta$ is fixed in terms of $K_+$ and $h_+$. Hence, when considering a quench protocol modifying the boundary field $h_+$ at site 1 from the gauge invariant one to a more general non-longitudinal one, one can choose for the computation of the corresponding overlaps a SoV basis of the gauge invariant model which also yields a separated representation for the eigenstates of the post quench Hamiltonian with a general boundary field at site~1. In this way, both sets of eigenstates—before and after the quench—can be expressed in the same SoV basis, and the computation of overlaps reduces to scalar products of separate states with identical separated variables. This mechanism generalises the longitudinal boundary result of \cite{AbeKT25} to a case in which one of the boundary configurations is completely general (the other one involving the gauge invariant boundary field at site 1), and it does so through a completely different and independent set of techniques, relying on the flexibility of the SoV bases in the gauge invariant model. This structural role of the gauge invariant boundary condition—as a ``bridge'' enabling the computation of overlaps with fully generic boundary matrices—is one of the central messages of the present work. 
%It allows us to compute overlaps between eigenstates of the gauge invariant model and eigenstates of a model with completely general boundary fields on finite chains, and it sets the basis for analysing their thermodynamic limit. The resulting framework provides a unified and efficient description of boundary quenches in open XXZ chains with unparallel boundary fields.

The paper is organised as it follows. In Section~\ref{sec-XXZ}, we introduce standard notations and definitions concerning the model. In Section~\ref{sec-SoVbases}, we give the complete characterisation of the spectrum and eigenstates of the XXZ open spin chain in the gauge invariant case. This is done by constructing, in the Vertex-IRF framework, the family of SoV bases  which provide a separated representation for these eigenstates. In  Section~\ref{sec-corr}, we  compute all elementary blocks for the correlation functions of this model, both for the finite chain and in the thermodynamic limit.  This is done by choosing, among all possible SoV bases, the one which corresponds to the ungauged limit. In Section~\ref{sec-overlap}, we explain how one can compute overlaps in the case of a boundary quench which modifies the boundary field at site 1 from the gauge-invariant one to a general, non-longitudinal one, the (non-longitudinal) boundary field at site $N$ remaining unchanged. The resulting overlaps are expressed in finite volume as determinants of generalised Slavnov's type. Finally, we present our conclusion in Section 6.  

%%%%%%%%%%%%%
\section{The open spin-1/2 XXZ quantum chain with boundary fields}
\label{sec-XXZ}

We recall in this section the general algebraic framework  \cite{Skl88} leading to the solution of the open spin-1/2 XXZ quantum chain with general boundary fields.
The Hamiltonian of this model is an operator acting on the $2^{N}$-dimensional quantum space $\mathcal{H}=\otimes _{n=1}^{N}\mathcal{H}_{n}$ (with $\mathcal{H}_{n}\simeq \mathbb{C}^{2}$) given by
\begin{equation}\label{Ham}
   H=\sum_{n=1}^{N-1} \Big[ \sigma_n^x\sigma_{n+1}^x+ \sigma_n^y\sigma_{n+1}^y+\Delta\, \sigma_n^z\sigma_{n+1}^z\Big]
   +\sum_{a\in\{x,y,z\}}\Big[ h_+^a\sigma_1^a+h_-^a\sigma_N^a\Big].
\end{equation}
Here, the local operators at site $n$ are given by the corresponding usual Pauli matrices, $\sigma _{n}^{\alpha }\in\End(\mathcal{H}_{n})$  ($\alpha \in \{x,y,z\}$), the anisotropy coupling constant is parametrised as $\Delta =\cosh \eta $, and the components $h_\pm^a$  of the  two boundary magnetic fields are parametrised in terms of six boundary parameters $\varsigma_{\pm }$, $\kappa_{\pm }$, $\tau_{\pm }$, or alternatively in terms of $\varphi_\pm$, $\psi_\pm$, $\tau_\pm$, as
\begin{align}
  &h_\pm^x=2\kappa_\pm \, \sinh\eta\,\frac{ \cosh\tau_\pm}{\sinh\varsigma_\pm}
                  =\sinh\eta\, \frac{\cosh\tau_\pm}{\sinh\varphi_\pm\,\cosh\psi_\pm},\\
  &h_\pm^y=2i\kappa_\pm \, \sinh\eta\,\frac{ \sinh\tau_\pm}{\sinh\varsigma_\pm}
                  =i\sinh\eta\, \frac{\sinh\tau_\pm}{\sinh\varphi_\pm\,\cosh\psi_\pm},\\
  &h_\pm^z=\sinh\eta\,\coth\varsigma_\pm %\frac{ \cosh\zeta_\pm}{\sinh\zeta_\pm}
                  =\sinh\eta\, \coth\varphi_\pm\,\tanh\psi_\pm.
                  %\frac{\cosh\alpha_\pm\, \sinh\beta_\pm}{\sinh\alpha_\pm\, \cosh\beta_\pm},
\end{align}
These  two sets of boundary parameters are related by
\begin{equation}
   \sinh\varphi_\pm\, \cosh \psi_\pm =\frac{\sinh\varsigma_\pm}{2\kappa_\pm},
   \qquad
   \cosh \varphi_\pm\, \sinh\psi_\pm =\frac{\cosh\varsigma_\pm}{2\kappa_\pm}.
\label{reparam-bords}
\end{equation}
%

%%%%%%%%
%\subsection{General algebraic framework}

%Let us briefly recall the general algebraic framework \cite{Skl88} ensuring 
The integrable structure of the model relies on the construction of a  boundary monodromy matrices $\mathcal{U}_{-,a}(\lambda )$, acting on $\mathcal{H}_{a}\otimes \mathcal{H}$, where $\mathcal{H}_{a}=\mathbb{C}^{2}$ is the so-called auxiliary space, which satisfies the following reflection equation:
\begin{equation}
R_{ba}(\lambda -\mu )\,\mathcal{U}_{-,a}(\lambda )\,R_{ab}(\lambda +\mu-\eta )\,\mathcal{U}_{-,b}(\mu )
=\mathcal{U}_{-,b}(\mu )\,R_{ba}(\lambda+\mu -\eta )\,\mathcal{U}_{-,a}(\lambda )\,R_{ab}(\lambda -\mu ).
\label{bYB}
\end{equation}
The latter has to be understood on $\mathcal{H}_{a}\otimes \mathcal{H}_{b}\otimes \mathcal{H}$, where $\mathcal{H}_{a,b}$ are two copies of the auxiliary space.
In this equation, the R-matrix
\begin{equation}
R_{ab}(\lambda )=%
\begin{pmatrix}
\sinh (\lambda +\eta ) & 0 & 0 & 0 \\ 
0 & \sinh \lambda  & \sinh \eta  & 0 \\ 
0 & \sinh \eta  & \sinh \lambda  & 0 \\ 
0 & 0 & 0 & \sinh (\lambda +\eta )%
\end{pmatrix}%
\in \text{End}(\mathcal{H}_{a}\otimes \mathcal{H}_{b}),  \label{R-6V}
\end{equation}
is the 6-vertex trigonometric solution of the Yang-Baxter equation.
The boundary monodromy matrix $\mathcal{U}_{-,a}(\lambda )$ solution of \eqref{bYB} can be explicitly constructed in terms of the bulk monodromy matrix $T_a(\lambda )\in \End(\mathcal{H}_a\otimes \mathcal{H})$ and of the boundary K-matrix $K_{-,a}(\lambda )\in \End(\mathcal{H}_a)$ as
\begin{equation}
\mathcal{U}_{-,a}(\lambda )
 =T_a(\lambda )\,K_{-,a}(\lambda )\,\hat{T}_a(\lambda )
 =\begin{pmatrix}
\mathcal{A}_{-}(\lambda ) & \mathcal{B}_{-}(\lambda ) \\ 
\mathcal{C}_{-}(\lambda ) & \mathcal{D}_{-}(\lambda )%
\end{pmatrix}.  \label{def-U-}
\end{equation}
The bulk monodromy matrix $T_a(\lambda )$ is itself constructed as
\begin{equation}
T_a(\lambda )=R_{a1}(\lambda -\xi _1-\eta /2)\dots R_{aN}(\lambda -\xi_N-\eta /2),  \label{mon-T}
\end{equation}
whereas
\begin{align}
\hat{T}_a(\lambda )& =(-1)^{N}\,\sigma _a^{y}\,T_a^{t_a}(-\lambda)\,\sigma _a^{y}  \notag \\
& =R_{aN}(\lambda +\xi _{N}-\eta /2)\dots R_{a1}(\lambda +\xi _{1}-\eta /2).
\label{That}
\end{align}
The boundary K-matrix is given by the most general scalar solution of the reflection equation \cite{deVG93,deVG94,GhoZ94},
\begin{equation}
K(\lambda ;\varsigma ,\kappa ,\tau )=\frac{1}{\sinh \varsigma }\,%
\begin{pmatrix}
\sinh (\lambda -\eta /2+\varsigma ) & \kappa e^{\tau }\sinh (2\lambda -\eta )
\\ 
\kappa e^{-\tau }\sinh (2\lambda -\eta ) & \sinh (\varsigma -\lambda +\eta
/2)%
\end{pmatrix}%
,  \label{mat-K}
\end{equation}
with
\begin{equation}
K_-(\lambda )=K(\lambda ;\varsigma _-,\kappa _-,\tau _-),
\qquad
K_+(\lambda )=K(\lambda +\eta ;\varsigma _+,\kappa _+,\tau _+),
\label{def-Kpm}
\end{equation}
so that here $K_-(\lambda)$ encodes the boundary field $h_-$ at site $N$, whereas $K_+(\lambda)$ encodes the boundary field $h_+$ at site 1.
Then, still following Sklyanin \cite{Skl88}, we introduce the transfer matrix as
\begin{equation}
 \mathcal{T}(\lambda )=\tr_a \big[K_{+,a}(\lambda )\,T_a(\lambda)\,K_{-,a}(\lambda )\,\hat{T}_a(\lambda ) \big]. \label{transfer}
\end{equation}
The latter forms a one-parameter family of commuting operators on $\mathcal{H}$, 
which also commute with the Hamiltonian \eqref{Ham} in the homogeneous limit $\xi _{m}=0$, $m=1,\ldots,N$. More precisely, in this homogeneous limit, the Hamiltonian \eqref{Ham} is given by
\begin{equation}
H=\frac{2\,(\sinh \eta )^{1-2N}}{\tr[K_{+}(\eta /2)]\,\tr [K_{-}(\eta /2)]}\
\frac{d}{d\lambda }\mathcal{T}(\lambda )_{\,\vrule height13ptdepth1pt\>{\lambda =\frac\eta 2},\ \xi _{m}=0}
+ \text{\ constant.}
\label{Ht}
\end{equation}

The boundary monodromy matrix $\mathcal{U}_{-}(\lambda )$ moreover satisfies the following inversion relation:
\begin{equation}
\mathcal{U}_{-}(\lambda +\eta /2)\,\mathcal{U}_{-}(-\lambda +\eta /2)
=\frac{\det_{q}\mathcal{U}_{-}(\lambda )}{\sinh (2\lambda -2\eta )},
\label{inv-U-}
\end{equation}
in terms of the quantum determinant $\det_q\mathcal{U}_{-}(\lambda )$, 
\begin{align}
\frac{\det_q\mathcal{U}_{-}(\lambda )}{\sinh (2\lambda -2\eta )}
& =\mathcal{A}_{-}(\eta /2\pm \lambda )\,\mathcal{A}_{-}(\eta /2\mp \lambda )+\mathcal{B}_{-}(\eta /2\pm \lambda )\,\mathcal{C}_{-}(\eta /2\mp \lambda ) 
\notag \\
& =\mathcal{D}_{-}(\eta /2\pm \lambda )\,\mathcal{D}_{-}(\eta /2\mp \lambda)+\mathcal{C}_{-}(\eta /2\pm \lambda )\,\mathcal{B}_{-}(\eta /2\mp \lambda ).
\label{det-U-}
\end{align}
The latter is a central element of the reflection algebra,
\begin{equation}
   \big[{\det}_q\, \mathcal{U}_{-}(\lambda )\, ,\, \mathcal{U}_{-}(\mu ) \big]=0,
\end{equation}
which takes the following form:
\begin{equation}
{\det}_{q}\,\mathcal{U}_{-}(\lambda )={\det}_{q} T(\lambda )\,{\det}_{q} T(-\lambda )\,{\det}_{q} K_{-}(\lambda ),
\label{det-prod}
\end{equation}
in terms of the bulk quantum determinant 
\begin{equation}
   {\det}_{q}T(\lambda )=a(\lambda +\eta /2)\,d(\lambda -\eta /2),
\label{det-T}
\end{equation}
and of the quantum determinant of the scalar boundary matrix $K_{-}(\lambda )$:
\begin{equation}
\frac{ {\det}_{q}K_{\pm }(\lambda )}{\sinh (2\lambda \pm 2\eta )}
=\mp \frac{\big(\sinh ^{2}\lambda -\sinh ^{2}\varphi _{\pm }\big)\big(\sinh^{2}\lambda +\cosh ^{2}\psi _{\pm }\big)}{\sinh^{2}\varphi _{\pm }\,\cosh^{2}\psi _{\pm }}.
\label{det-K-}
\end{equation}
Here we have defined
\begin{equation}
a(\lambda )=\prod_{n=1}^{N}\sinh (\lambda -\xi _{n}+\eta /2),\qquad
d(\lambda )=\prod_{n=1}^{N}\sinh (\lambda -\xi _{n}-\eta /2),  \label{a-d}
\end{equation}
and we have used in \eqref{det-K-} the rewriting of the boundary parameters  $\varsigma _{\pm },\kappa _{\pm }$ in terms of $\varphi _{\pm },\psi _{\pm }$ given by \eqref{reparam-bords}.

In this paper, we shall more particularly consider a special boundary configuration given by the following K-matrix on site 1,
\begin{equation}\label{gauge-inv-K}
K^\text{(Inv)}_+(\lambda )=\begin{pmatrix}
   e^{\lambda+\eta/2 } & 0 \\ 
   0 & e^{-(\lambda +\eta /2)} \end{pmatrix}
   =e^{(\lambda+\eta/2 )\sigma^{z}},
\end{equation}
while leaving the boundary field (and parameters) on the last site of the chain arbitrary. This means that the boundary field on site 1 is given by
\begin{equation}\label{h+inv}
   h_+^\mathrm{(Inv)}=\sinh\eta \, \sigma_1^z.
\end{equation}
The quantum determinant of the K-matrix \eqref{gauge-inv-K} is simply given by
\begin{equation}
   {\det}_{q}K^\text{(Inv)}_+(\lambda )=\sinh(2\lambda+2\eta).
\end{equation}
Note that the special K-matrix \eqref{gauge-inv-K} can be obtained from the general one by taking the limit $\varsigma_+\to +\infty$ while keeping $\kappa_+$ and $\tau_+$ finite (or alternatively, in terms of the other set of boundary parameters, by taking the limit $\varphi_+,\psi_+ \to +\infty$ while keeping $\tau_+$ finite):
\begin{equation}
   K^\text{(Inv)}_+(\lambda )=\lim_{\varsigma_+\to +\infty} K(\lambda +\eta ;\varsigma _+,\kappa _+,\tau _+).
\end{equation}
For this particular boundary condition at site 1 and generic boundary condition at site $N$, the corresponding transfer matrix is
\begin{align}\label{Tinv}
   \mathcal{T}^\text{(Inv)}(\lambda) 
   &= \tr \left[ {K}_+^\text{(Inv)}(\lambda)\ \mathcal{U}_{-}(\lambda )\right] \nonumber\\
   &= e^{\lambda-\frac\eta 2}\,\frac{\sinh(2\lambda+\eta)}{\sinh2\lambda}\,\mathcal{A}_-(\lambda)
   + e^{-\lambda-\frac\eta 2}\,\frac{\sinh(2\lambda-\eta)}{\sinh2\lambda}\,\mathcal{A}_-(-\lambda)
   \nonumber\\
   &= e^{-\lambda+\frac\eta 2}\,\frac{\sinh(2\lambda+\eta)}{\sinh2\lambda}\,\mathcal{D}_-(\lambda)
   + e^{\lambda+\frac\eta 2}\,\frac{\sinh(2\lambda-\eta)}{\sinh2\lambda}\,\mathcal{D}_-(-\lambda),
 \end{align}
 in which $\mathcal{U}_{-}$ is here constructed from a generic, non-diagonal boundary matrix $K_-$. 
As we shall see, this particular K-matrix presents the interesting property that, when gauge-transformed through the vertex-IRF transformation, it produces a matrix which is still diagonal and moreover does not depend on the gauge parameters, see \eqref{hatK+inv}, so that the transfer matrix \eqref{Tinv} can be diagonalised through Sklyanin's SoV approach in a family of SoV bases with arbitrary gauge parameters, see Section~\ref{sec-SoVbases}.
%\footnote{Note that we have already consider this special kind of boundary K-matrix and boundary condition in our previous article \cite{NicT22}, but on the last site of the chain. The gauge invariance property of the transfer matrix was used there to write a simple boundary-bulk decomposition of the boundary Bethe states in terms of bulk ones. Here instead, we shall play with this invariance to derive other kinds of properties such as e.g. the possibility to compute some overlaps.}.

%%%%%%%%%%%%%%%%%%%%%%%%%%%%%%%%%%%%%%%%%
%%%%%%%%%%%%%%%%%%%%%%%%%%%%%%%%%%%%%%%%%
\section{The gauge invariant case: transfer matrix eigenstates in different SoV bases}
\label{sec-SoVbases}

The transfer matrix $ \mathcal{T}^\text{(Inv)}(\lambda)$ \eqref{Tinv} is, by construction, a polynomial in $\sinh^2\lambda$ of degree $N+1$, which satisfies the following centrality conditions:
\begin{align}
   &\mathcal{T}^\text{(Inv)}(\tfrac\eta 2)=(-1)^N 2\cosh\eta\, {\det}_qT(0), \label{T-value1} \\
   &\mathcal{T}^\text{(Inv)}(\tfrac\eta 2+i\tfrac\pi 2)= -2\cosh\eta\,\coth\varsigma_+\,{\det}_q T(i\tfrac\pi 2), \label{T-value2}
\end{align}
as well as the quantum determinant identity, 
\begin{align}
   \mathcal{T}^\text{(Inv)}(\xi_n+\tfrac\eta 2)\, \mathcal{T}^\text{(Inv)}(\xi_n-\tfrac\eta 2)
   &= -\frac{\det_q K_+(\xi_n)\, \det_q\mathcal{U}_-(\xi_n)}{\sinh(2\xi_n+\eta)\,\sinh(2\xi_n-\eta)}
   \nonumber\\
   &= -\frac{(\sinh^2\lambda-\sinh^2\varphi_-)(\sinh^2\lambda+\cosh^2\psi_-)}{\sinh^2\varphi_-\,\cosh^2\psi_-}
   \nonumber\\
   &\quad\times
   \frac{\sinh(2\xi_n+2\eta)\,\sinh(2\xi_n-2\eta)}{\sinh(2\xi_n+\eta)\,\sinh(2\xi_n-\eta)}\,{\det}_qT(\xi_n)\,{\det}_qT(-\xi_n).
\end{align}
In this section, we diagonalise this transfer matrix by Sklyanin's SoV approach. The K-matrix \eqref{gauge-inv-K} being diagonal, this can of course be done as in \cite{Nic12} by using a basis which diagonalises the operator $\mathcal{B}_-(\lambda)$. This is however not the only possible basis that can be used: we show here that one can actually use a whole family of bases pseudo-diagonalising  gauge transform operators $\mathcal{B}_-(\lambda|\alpha,\beta)$ for arbitrary gauge parameter $\alpha$ and $\beta$. We moreover discuss the dependence of such bases on the gauge parameters $\alpha$ and $\beta$, and show that they only depend on their difference $\alpha-\beta$.

%%%%%%%
\subsection{Gauge transformation of the model}
\label{sec-gauge}

%Let us first briefly recall  the gauge transformation of the model through Baxter's Vertex-IRF transformation. This transformation has been used in \cite{FalKN14,KitMNT18,NicT22,NicT24} to diagonalise the transfer matrix of the XXZ spin chain with general open boundary conditions within Sklyanin's Separation of Variable (SoV) approach.

Following \cite{CaoLSW03,FalKN14,KitMNT18,NicT22,NicT24}, we consider  the gauge transformation of the model via Baxter's Vertex-IRF transformation \cite{Bax73a}. More precisely, we define the transformed gauge boundary monodromy matrix as
\begin{align}
\mathcal{U}_{-}(\lambda |\alpha ,\beta )& =S^{-1}(\eta /2-\lambda |\alpha ,\beta )\ \mathcal{U}_{-}(\lambda )\ S(\lambda -\eta /2|\alpha,\beta )  \notag \\
& =\begin{pmatrix}
\mathcal{A}_{-}(\lambda |\alpha ,\beta ) & \mathcal{B}_{-}(\lambda |\alpha,\beta ) \\ 
\mathcal{C}_{-}(\lambda |\alpha ,\beta ) & \mathcal{D}_{-}(\lambda |\alpha,\beta )
\end{pmatrix}, \label{gauged-U}
\end{align}
in which $S(\lambda|\alpha,\beta)$ is the trigonometric Vertex-IRF transformation matrix, which reads 
\begin{equation}
S(\lambda |\alpha ,\beta )=\begin{pmatrix}
e^{\lambda -\eta (\alpha+\beta)} & e^{\lambda -\eta (\alpha-\beta)} \\ 
1 & 1 \end{pmatrix}.  \label{mat-S}
\end{equation}
Here $\alpha $ is an arbitrary shift of the spectral parameter, whereas $\beta$ is the so-called dynamical parameter.
The inverse of \eqref{mat-S} explicitly reads:
\begin{equation}\label{S-1}
    S^{-1}(\lambda|\alpha,\beta)=\frac{1}{\det S(\lambda|\alpha,\beta)}\, S^\mathrm{Adj}(\lambda|\alpha,\beta),
\end{equation}
with
\begin{align}\label{Sadj-detS}
    S^\mathrm{Adj}(\lambda|\alpha,\beta) = \begin{pmatrix} 1 & -e^{\lambda-(\alpha-\beta)\eta} \\ -1 & e^{\lambda-(\alpha+\beta)\eta} \end{pmatrix},
    \qquad
    \det S(\lambda|\alpha,\beta) = -2\, e^{\lambda-\alpha\eta}\,\sinh(\beta\eta).
\end{align}
Then, the transformed gauge monodromy matrix \eqref{gauged-U} satisfies a dynamical version of the reflection equation, 
\begin{align}
& R_{21}^{\mathrm{SOS}}(\lambda -\mu |\beta )\,\mathcal{U}_{-,1}(\lambda |\alpha ,\beta +\sigma _{2}^{z})\,R_{12}^{\mathrm{SOS}}(\lambda +\mu -\eta |\beta )\,\mathcal{U}_{-,2}(\mu |\alpha ,\beta +\sigma _{1}^{z}) 
 \notag \\
&\hspace{2cm}
 =\mathcal{U}_{-,2}(\mu |\alpha ,\beta +\sigma _{1}^{z})\,R_{21}^{\mathrm{SOS}}(\lambda +\mu -\eta |\beta )\,\mathcal{U}_{-,1}(\lambda |\alpha ,\beta +\sigma _{2}^{z})\,R_{12}^{\mathrm{SOS}}(\lambda -\mu |\beta ),
\end{align}
involving the R-matrix $R^\mathrm{SOS}$ of the trigonometric solid-on-solid (SOS) model:
\begin{equation}
R^{\mathrm{SOS}}(\lambda |\beta )=%
\begin{pmatrix}
\sinh (\lambda +\eta ) & 0 & 0 & 0 \\ 
0 & \frac{\sinh (\eta (\beta +1))}{\sinh (\eta \beta )}\,\sinh \lambda  & 
\frac{\sinh (\lambda +\eta \beta )}{\sinh (\eta \beta )}\,\sinh \eta  & 0 \\ 
0 & \frac{\sinh (\eta \beta -\lambda )}{\sinh (\eta \beta )}\,\sinh \eta  & 
\frac{\sinh (\eta (\beta -1))}{\sinh (\eta \beta )}\,\sinh \lambda  & 0 \\ 
0 & 0 & 0 & \sinh (\lambda +\eta )%
\end{pmatrix}%
.  \label{R-SOS}
\end{equation}

It is worth remarking that, by definition, the operators
%the off-diagonal elements of the gauge monodromy matrix \eqref{gauged-U}, 
%
\begin{equation*}
     \det S(\tfrac\eta 2-\lambda|\alpha,\beta)\, \mathcal{B}_{-}(\lambda |\alpha,\beta ),
     \quad \text{respectively}
     \quad
     \det S(\tfrac\eta 2-\lambda|\alpha,\beta)\, \mathcal{C}_{-}(\lambda |\alpha,\beta ), 
\end{equation*}
in \eqref{gauged-U} depend on the gauge parameters only through the difference $\alpha -\beta $, respectively through the sum $\alpha +\beta $. In the following, we shall intensively use this property for the rescaled version  of the former operator. We therefore define 
\begin{align}
  \widehat{\mathcal{B}}_{-}(\lambda |\alpha -\beta ) 
  &=\det S(\tfrac\eta 2-\lambda|\alpha,\beta)\, \mathcal{B}_{-}(\lambda |\alpha,\beta )=-2\, e^{-\lambda+\frac\eta 2-\alpha\eta}\sinh(\beta\eta)\, \mathcal{B}_{-}(\lambda |\alpha,\beta )
  \nonumber\\
  &=\mathcal{B}_-(\lambda)+e^{\lambda-\frac\eta 2-\eta(\alpha-\beta)}\mathcal{A}_-(\lambda)-e^{-\lambda+\frac\eta 2-\eta(\alpha-\beta)}\mathcal{D}_-(\lambda)-e^{-2\eta(\alpha-\beta)}\mathcal{C}_-(\lambda).
   \label{Bhat}
\end{align}
%
%and
%
%\begin{align} \label{Chat}
% &\widehat{\mathcal{C}}_{-}(\lambda |\alpha +\beta )   =\det S(\tfrac\eta 2-\lambda|\alpha,\beta)\, \mathcal{C}_{-}(\lambda |\alpha,\beta )=-2\, e^{-\lambda+\frac\eta 2-\alpha\eta}\sinh(\beta\eta)\, \mathcal{C}_{-}(\lambda |\alpha,\beta ) .
 %\end{align}
%
%which depend only on the difference $\alpha -\beta $, respectively on the sum $\alpha +\beta $, of the gauge parameters $\alpha$ and $\beta$. 

\begin{rem}\label{rem-limit}
It will be interesting, in the following, to consider the limit $\eta(\alpha-\beta)\to +\infty$ with $\eta(\alpha+\beta)$ remaining fixed and finite. In this limit the Vertex-IRF matrices $S(\lambda |\alpha ,\beta )$ \eqref{mat-S} and $S^\mathrm{Adj}(\lambda|\alpha,\beta) $ become lower triangular, and 
\begin{equation}\label{Bhat-lim}
   \widehat{\mathcal{B}}_{-}(\lambda |\alpha -\beta ) \to \mathcal{B}_- (\lambda ),
\end{equation}
whereas %the limit of $\mathcal{A}_{-}(\lambda |\alpha ,\beta ) $ and $\mathcal{D}_{-}(\lambda |\alpha ,\beta )$ are expressed in terms of 
\begin{align}
  &\mathcal{A}_{-}(\lambda |\alpha ,\beta ) \to e^{2\lambda-\eta} \,  \mathcal{A}_{-}(\lambda  ) +e^{\lambda-\frac\eta 2+\eta(\alpha+\beta)}\,\mathcal{B}_{-}(\lambda), \label{A-lim}\\
  &\mathcal{D}_{-}(\lambda |\alpha ,\beta ) \to  \mathcal{D}_{-}(\lambda  ) - e^{\lambda-\frac\eta 2+\eta(\alpha+\beta)}\, \mathcal{B}_{-}(\lambda  ) .
\end{align}
\end{rem}

\begin{rem}\label{rem-norm-hatB}
The definition \eqref{Bhat} of the renormalised operator $\widehat{\mathcal B}_-$ %, respectively \eqref{Chat}, 
differs from the one used in \cite{NicT22,NicT23} by a factor $-2 e^{-(\alpha-\beta)\eta}$. %, respectively $-2 e^{-(\alpha+\beta)\eta}$. 
The purpose of this new choice of normalisation is to simplify the consideration of the aforementioned limit, see \eqref{Bhat-lim}. It also formally matches the choice of normalisation used in the XYZ case in \cite{NicT25}. 
\end{rem}

%%%%%%%%%%%%%%%
\subsection{Expression of the transfer matrix $\mathcal{T}^\text{(Inv)}(\lambda) $ in terms of the gauge operators}

For any value of the gauge parameters $( \alpha ,\beta ) $, the gauge transformation of the K-matrix \eqref{gauge-inv-K} explicitly reads :
\begin{align}\label{hatK+inv}
    \widehat{K}_+^\text{(Inv)}(\lambda |\alpha,\beta)
    &\equiv S^{-1}(\lambda +\tfrac\eta 2|\alpha ,\beta )\ K_{+}^\text{(Inv)}(\lambda )\ S(-\lambda-\tfrac\eta 2 |\alpha ,\beta )  \nonumber \\
%    &= S^{-1}(\lambda -\eta /2|\alpha -1,\beta )\ K_{+}^\text{(Inv)}(\lambda )\ S(\eta /2-\lambda |\alpha +1,\beta )  \nonumber \\
    &= e^{-(\lambda +\eta /2)}
\begin{pmatrix}
1 & 0 \\ 
0 & 1
\end{pmatrix}.
\end{align}
Hence, the transformed gauge matrix $\widehat{K}_+^\text{(Inv)}(\lambda |\alpha,\beta) \equiv \widehat{K}_+^\text{(Inv)}(\lambda)$ is diagonal (it is even proportional to the identity matrix) and moreover does not depend on the value of the gauge parameters $\alpha$ and $\beta$. This is a very interesting property which enables us to express the corresponding transfer matrix \eqref{Tinv} with the special boundary condition \eqref{gauge-inv-K}-\eqref{h+inv} at site 1 and generic boundary conditions at site $N$ in terms of the transformed gauge operators $\mathcal{A}_{-}(\pm\lambda |\alpha ,\beta-1 )$ only, or alternatively in terms of $\mathcal{D}_{-}(\pm\lambda |\alpha ,\beta+1 )$, {\em for any arbitrary values of $\alpha$ and $\beta$}.

More precisely, 
\begin{align}\label{Tinv-tr}
   \mathcal{T}^\text{(Inv)}(\lambda) 
%   &= \tr \left[ {K}_+^\text{(Inv)}(\lambda)\ \mathcal{U}_{-}(\lambda )\right] \nonumber\\
   &= \tr \left[ \widehat{K}_+^\text{(Inv)}(\lambda |\alpha,\beta)\ \widehat{\mathcal{U}}_{-}(\lambda |\alpha,\beta)\right] ,
\end{align}
with $\widehat{\mathcal{U}}_{-}(\lambda |\alpha ,\beta ) $ being of the form:
\begin{align}\label{hat-U-}
\widehat{\mathcal{U}}_{-}(\lambda |\alpha ,\beta ) 
&=S^{-1}(-\lambda-\tfrac\eta 2 |\alpha,\beta )\ \mathcal{U}_{-}(\lambda )\ S(\lambda+\tfrac\eta 2|\alpha ,\beta ) \nonumber\\
&=S^{-1}(\tfrac\eta 2-\lambda |\alpha +1,\beta )\ \mathcal{U}_{-}(\lambda )\ S(\lambda-\tfrac\eta 2|\alpha -1,\beta ) \nonumber\\
&=
\begin{pmatrix}
e^{\eta }\,\frac{\sinh (\eta (\beta -1))}{\sinh (\eta \beta )\,}\mathcal{A}_{-}(\lambda |\alpha ,\beta -1) & \star  \\ 
\star  & e^{\eta }\,\frac{\sinh (\eta (\beta +1))}{\sinh (\eta \beta )}\,\mathcal{D}_{-}(\lambda |\alpha ,\beta +1)
\end{pmatrix},
\end{align}
in which we do not specify the off-diagonal elements of \eqref{hat-U-} since they do not contribute to the trace \eqref{Tinv-tr}.
%with the special boundary condition \eqref{gauge-inv-K}-\eqref{h+inv} at site 1 and generic boundary conditions at site $N$, can be explicitly written in terms of gauge-transformed operators $\mathcal{A}_{-}(\pm\lambda |\alpha ,\beta-1 )$ or $\mathcal{D}_{-}(\pm\lambda |\alpha ,\beta+1 )$ as
This leads to the following expressions for the transfer matrix $\mathcal{T}^\text{(Inv)}(\lambda )$ in terms of $\mathcal{A}_{-}(\pm\lambda |\alpha ,\beta-1 )$, or alternatively in terms of  $\mathcal{D}_{-}(\pm\lambda |\alpha ,\beta+1 )$\footnote{We recall that the operators $\mathcal{A}_{-}(\pm\lambda |\alpha ,\beta-1 )$ and $\mathcal{D}_{-}(\pm\lambda |\alpha ,\beta+1 )$ are related by the parity relations:
\begin{align*}
  &\mathcal{D}_-(\lambda|\alpha,\beta+1)
  =\frac{\sinh\eta\,\sinh(2\lambda+\eta\beta)}{\sinh(\eta(\beta+1))\,\sinh(2\lambda)}\,\mathcal{A}_-(\lambda|\alpha,\beta-1)
%  \nonumber\\
%  &\hspace{6cm}
  +e^{2\lambda}\frac{\sinh(\eta\beta)\,\sinh(2\lambda-\eta)}{\sinh(\eta(\beta+1))\,\sinh(2\lambda)}\, \mathcal{A}_-(-\lambda|\alpha,\beta-1),
  \\
   &\mathcal{A}_-(\lambda|\alpha,\beta-1)
  =-\frac{\sinh\eta\,\sinh(2\lambda-\eta\beta)}{\sinh(\eta(\beta-1))\,\sinh(2\lambda)}\,\mathcal{D}_-(\lambda|\alpha,\beta+1)
%  \nonumber\\
%  &\hspace{6cm}
+e^{2\lambda}\frac{\sinh(\eta\beta)\,\sinh(2\lambda-\eta)}{\sinh(\eta(\beta-1))\,\sinh(2\lambda)}\, \mathcal{D}_-(-\lambda|\alpha,\beta+1).
%  \\
%  &\mathcal{B}_-(-\lambda|\alpha,\beta)=-e^{-2\lambda}\frac{\sinh(2\lambda+\eta)}{\sinh(2\lambda-\eta)}\,\mathcal{B}_-(\lambda|\alpha,\beta),  \\
%  &\mathcal{C}_-(-\lambda|\alpha,\beta)=-e^{-2\lambda}\frac{\sinh(2\lambda+\eta)}{\sinh(2\lambda-\eta)}\,\mathcal{C}_-(\lambda|\alpha,\beta),
\end{align*}
}:
\begin{align}
\mathcal{T}^\text{(Inv)}(\lambda )
& =e^{-\lambda +\eta /2}\,\frac{\sinh (2\lambda+\eta )}{\sinh 2\lambda }\,\mathcal{A}_{-}(\lambda |\alpha ,\beta-1)
   +e^{\lambda +\eta /2}\,\frac{\sinh (2\lambda -\eta )}{\sinh 2\lambda }\,\mathcal{A}_{-}(-\lambda |\alpha ,\beta -1) 
   \label{Tinv-gaugeA}\\
& =e^{-\lambda +\eta /2}\,\frac{\sinh (2\lambda +\eta )}{\sinh 2\lambda }\,\mathcal{D}_{-}(\lambda |\alpha ,\beta +1)
   +e^{\lambda +\eta /2}\,\frac{\sinh(2\lambda -\eta )}{\sinh 2\lambda }\,\mathcal{D}_{-}(-\lambda |\alpha,\beta +1).
    \label{Tinv-gaugeD}
\end{align}
We stress again that, unlike for more general boundary conditions, see \eqref{Gauge-T-A}-\eqref{Gauge-T-D} in which the parameters $\alpha$ and $\beta$ are fixed by \eqref{cond-diff-gauge}-\eqref{cond-sum-gauge}, the expressions \eqref{Tinv-gaugeA} and \eqref{Tinv-gaugeD} are valid  for any value of the gauge parameters $\alpha$ and $\beta$.  Moreover, the dependence on these parameters in \eqref{Tinv-gaugeA} and \eqref{Tinv-gaugeD} is only contained in the gauged boundary operators $\mathcal{A}_{-}(\lambda |\alpha ,\beta -1)$ and $\mathcal{D}_{-}(\lambda |\alpha ,\beta +1)$.

%%%%%%%%%%%
\subsection{Different SoV bases and gauge dependence}
\label{sec-diff-bases}

%Following Sklyanin's SoV approach \cite{Skl85,Skl90},  and due to the very simple form \eqref{Tinv-gaugeA} or \eqref{Tinv-gaugeD} of the transfer matrix $\mathcal{T}^\text{(Inv)}$ in terms of the gauge operators $\mathcal{A}_{-}(\pm\lambda |\alpha ,\beta-1 )$ or  $\mathcal{D}_{-}(\pm\lambda |\alpha ,\beta+1 )$, one can define 

Following Sklyanin's SoV approach \cite{Skl85,Skl90}, a SoV basis pseudo-diagonalising the gauge operator $\mathcal{B}_-(\lambda|\alpha,\beta)$, or equivalently its renormalised counterpart \eqref{Bhat}, was explicitly constructed in \cite{FalKN14,KitMNT18}. We recall here this construction and, in addition, discuss the dependence of such a basis in terms of the gauge parameters $\alpha$ and $\beta$.

%More precisely, we have shown that the gauged boundary $\widehat{\mathcal{B}}_{-}$-operators depend on the gauge parameters $\left( \alpha ,\beta \right) $ only through their difference $\alpha -\beta $. Then, the SoV pseudo-eigenbasis of these gauged boundary $\widehat{\mathcal{B}}_{-}$-operators, when properly normalised, should {\em a priori} also depend on the gauge parameters $\left( \alpha ,\beta \right) $ only through their difference $\alpha -\beta $. Here, we use the new SoV basis \cite{MaiN18,MaiN19}, constructed by the use of the transfer matrice $\mathcal{T}^\text{(inv)}$, to prove that this is the case and that indeed all the dependence in $\alpha -\beta $ of Sklyanin's generalised SoV basis is contained in the definition of the pseudo-reference state.

Let us suppose that the inhomogeneity parameters are generic, or at least that they satisfy the condition
\begin{equation} \label{cond-inh}
   \xi_j^{(h_j)}\not=\xi_k^{(h_k)}\!\!\!\!\mod i\pi,\quad \forall\, (j,h_j),(k,h_k)\in \{1,\ldots ,N\}\times\{0,1\} \text{ such that } (j,h_j)\not= (k,h_k),
\end{equation}
in which we have used the following shortcut notation for shifted inhomogeneity parameters:
\begin{equation}\label{xi-shifted}
    \xi_j^{(h)}=\xi_j+\tfrac\eta 2 -h\eta,\qquad j\in\{1,\ldots,N\},\quad h\in\{0,1\}.
\end{equation}
%
%\begin{equation}
%\xi _{j},\xi _{j}\pm \xi _{k}\notin \{0,-\eta ,\eta \}\text{ mod}(i\pi),\quad \forall j,k\in \{1,\ldots ,N\},\ j\neq k, \label{cond-inh}
%\end{equation}
%
Then, let us define, for each  $N$-tuple $\mathbf{h}\equiv (h_{1},\ldots ,h_{N})\in \{0,1\}^{N}$, and  arbitrary values of $\alpha$ and $\beta$, the following states:
%
%\noindent\fbox{\parbox{\linewidth-2\fboxrule-2\fboxsep}{
%
\begin{align}
  & \ket{\mathbf{h},\alpha,\beta+1}_\text{Sk} =\prod_{n=1}^N\left(\frac{\mathcal{D}_-(\xi_n+\frac\eta 2|\alpha,\beta+1)}{\mathsf{k}_n\,\mathsf{A}_-(\frac\eta 2-\xi_n)}\right)^{\! h_n} 
   S_{1\ldots N}(\boldsymbol{\xi} |\alpha,\beta)\,\ket{\underline{0}},
   \label{ket-SoV} \\
%\end{align}
%
%and\footnote{Note that the chosen normalisation is here slightly different from the one used in our precedent works \cite{KitMNT18,NicT22}.}
%
%\begin{align}
   &{}_\text{Sk}\bra{\alpha,\beta-1,\mathbf{h}} =\bra{0}\, S^\mathrm{Adj}_{1\ldots N}(\boldsymbol{\xi} |\alpha,\beta)
   \prod_{n=1}^N\left(\frac{\mathcal{A}_-(\frac\eta 2-\xi_n |\alpha,\beta-1)}{\mathsf{A}_-(\frac\eta 2-\xi_n)}\right)^{\! 1-h_n} ,
   \label{bra-SoV}
\end{align}
in which $h_n$, $n\in\{1,\ldots,N\}$, denotes the $n$-th component of the $N$-tuple $\mathbf{h}$. 
Here $\ket{0}=\otimes_{n=1}^N \left( \nobarfrac{1}{0} \right)_n$ and $\ket{\underline{0}}=\otimes_{n=1}^N \left( \nobarfrac{0}{1} \right)_n$ denote the usual reference states, whereas their respective dual states are $\bra{0}$ and $\bra{\underline{0}}$. $S_{1\ldots N}(\boldsymbol{\xi} |\alpha,\beta)$ and $S^\mathrm{Adj}_{1\ldots N}(\boldsymbol{\xi} |\alpha,\beta)$ respectively stand for the following products of Vertex-IRF matrices all along the chain:
\begin{align}
    &S_{1\ldots N}(\boldsymbol{\xi} |\alpha,\beta)
    =S_1(-\xi_1|\alpha,\beta)\, S_2(-\xi_2|\alpha,\beta+\sigma_1^z)\ldots S_N(-\xi_N|\alpha,\beta+\sigma_1^z+\ldots+\sigma_{N-1}^z), \label{Sgauge-N}\\
    &S^\mathrm{Adj}_{1\ldots N}(\boldsymbol{\xi} |\alpha,\beta)    
    =S^\mathrm{Adj}_N(-\xi_N|\alpha,\beta+\sigma_1^z+\ldots+\sigma_{N-1}^z)\ldots S^\mathrm{Adj}_2(-\xi_2|\alpha,\beta+\sigma_1^z)\, S^\mathrm{Adj}_1(-\xi_1|\alpha,\beta). \label{Sgauge-Adj-N}
\end{align}
The normalisation coefficients $\mathsf{k}_n$ and $\mathsf{A}_{-}(\lambda )$ in \eqref{ket-SoV} and \eqref{bra-SoV} are chosen as
\begin{equation}\label{def-kn-A-}
  \mathsf{k}_n=\frac{\sinh(2\xi_n+\eta)}{\sinh(2\xi_n-\eta)}\qquad
  \text{and}
  \qquad
   \mathsf{A}_{-}(\lambda )=\mathsf{g}_{-}(\lambda )\,a(\lambda )\,d(-\lambda ),
\end{equation}
with $\mathsf{g}_-$ such that 
\begin{equation}\label{fct-g}
   \mathsf{g}_-(\lambda+\tfrac\eta 2)\, \mathsf{g}_-(-\lambda+\tfrac\eta 2)=\frac{\det_q K_-(\lambda)}{\sinh(2\lambda-2\eta)},
\end{equation}
so that we have
\begin{equation}\label{prod-A-}
   \mathsf{A}_-(\lambda+\tfrac\eta 2)\, \mathsf{A}_-(-\lambda+\tfrac\eta 2)=\frac{\det_{q}\mathcal{U}_{-}(\lambda )}{\sinh(2\lambda -2\eta )}.
\end{equation}
Note that we have chosen here a different normalisation for the state \eqref{bra-SoV} than in our previous articles \cite{KitMNT18,NicT22}. As we shall see, this new choice of normalisation is useful for simplifying the dependence of this state on the gauge parameters $\alpha$ and $\beta$.

As shown in \cite{FalKN14,KitMNT18,NicT22,NicT24}, such states pseudo-diagonalise the operator $\mathcal{B}_-(\lambda|\alpha,\beta)$, or  equivalently its renormalised counterpart  $\widehat{\mathcal{B}}_{-}(\lambda|\alpha-\beta)$. More precisely,  the action of \eqref{Bhat} on \eqref{ket-SoV} and \eqref{bra-SoV} is respectively given by  
%{\bf (to check)}
%
\begin{align}
  &\widehat{\mathcal{B}}_-(\lambda|\alpha-\beta)\, \ket{\mathbf{h},\alpha,\beta}_\text{Sk}
  =(-1)^N a_\mathbf{h}(\lambda)\,a_\mathbf{h}(-\lambda)\, %\sinh(2\lambda-\eta)\, \mathsf{b}
  \widehat{\mathsf b}_-(\lambda|\alpha-\beta+N)\,
  \ket{\mathbf{h},\alpha,\beta+2}_\text{Sk}, \label{eigenr-Bhat}
  \\
  & {}_\text{Sk}\bra{\alpha,\beta,\mathbf{h}}\, \widehat{\mathcal{B}}_-(\lambda|\alpha-\beta)
  =(-1)^N a_\mathbf{h}(\lambda)\,a_\mathbf{h}(-\lambda)\, %\sinh(2\lambda-\eta)\, \mathsf{b}
  \widehat{\mathsf b}_-(\lambda|\alpha-\beta-N)\,
   {}_\text{Sk}\bra{\alpha,\beta-2,\mathbf{h}},  \label{eigenl-Bhat}
\end{align}
where
\begin{align}
    &a_{\mathbf{h}}(\lambda )=\prod_{n=1}^{N}\sinh (\lambda -\xi _{n}-\tfrac \eta 2+h_n\eta ),  \label{a-h}
\end{align}
and\footnote{The coefficient $\widehat{\mathsf b}_-(\lambda|\alpha-\beta)$ is related to the upper off-diagonal coefficient of the transformed gauge K-matrix $K_-$ as
\begin{align}
   \widehat{\mathsf b}_-(\lambda|\alpha-\beta)
   &=\det S(-\lambda+\frac\eta 2 |\alpha,\beta)\left[ S^{-1}(-\lambda+\tfrac\eta 2 |\alpha,\beta)\, K_-(\lambda)\, S(\lambda-\tfrac \eta 2 |\alpha,\beta) \right]_{12} \nonumber\\
   &= \left[ S^\mathrm{Adj}(-\lambda+\tfrac\eta 2 |\alpha,\beta)\, K_-(\lambda)\, S(\lambda-\tfrac \eta 2 |\alpha,\beta) \right]_{12} .
\end{align}
} 
\begin{align}\label{hat-b-}
    &\widehat{\mathsf b}_-(\lambda|x)=\sinh(2\lambda-\eta)\, %\mathsf{b}_-(x) \\
    %&\mathsf{b}_-(x) = 
    \frac{e^{-x\eta}}{\sinh\varphi_-\,\cosh\psi_-}\,\big[ \sinh(\varphi_-+\psi_-)+\sinh(\tau_-+x\eta)\big]. %\label{b-}
\end{align}
We also recall that these states satisfy some orthogonality condition:
\begin{equation}\label{orth-SoV-Sk}
  \moy{ \alpha ,\beta -1,\mathbf{h}\,|\,\mathbf{k},\alpha ,\beta +1 }
   \propto %\mathsf{N}(\boldsymbol{\xi}|\alpha,\beta)\,
   \delta _{\mathbf{h},\mathbf{k}}\, 
  \frac{ e^{2\sum_{j=1}^{N}h_{j}\xi _{j}}}{\widehat{V}(\xi _{1}^{(h_{1})},\ldots,\xi _{N}^{(h_{N})})},
\end{equation}
with some normalisation coefficient that will be specified later, see \eqref{norm-coeff}.
%with normalisation coefficient 
%
%\begin{align}\label{norm-coeff}
%  \frac{\mathsf{N}(\boldsymbol{\xi}|\alpha,\beta)}{ \widehat{V}(\xi_{1}^{(0)},\ldots ,\xi _{N}^{(0)}) }
%     &=  \moy{ \alpha ,\beta -1,\boldsymbol{0}\,|\,\boldsymbol{0},\alpha ,\beta +1 } \nonumber\\
 %    &= \bra{0}\, S^\mathrm{Adj}_{1\ldots N}(\boldsymbol{\xi} |\alpha,\beta)   \left(\prod_{n=1}^{N} \frac{\mathcal{A}_-(\frac\eta 2-\xi_n |\alpha,\beta-1)}{\mathsf{A}_-(\frac\eta 2-\xi_n)} \right)   S_{1\ldots N}(\boldsymbol{\xi} |\alpha,\beta)\,\ket{\underline{0}}.
%\end{align}
%
Here 
\begin{equation}\label{VDM}
\widehat{V}(x_{1},\ldots ,x_{N})=\det_{1\leq i,j\leq N}\left[ \sinh^{2(j-1)}x_{i}\right] =\prod_{j<k}(\sinh ^{2}x_{k}-\sinh ^{2}x_{j}),
\end{equation}
denotes a generalised Vandermonde determinant for the $N$-tuple of variables $(x_1,\ldots,x_N)$.
Hence, under the condition \eqref{cond-inh}, the states  \eqref{ket-SoV} for $\mathbf{h}\in\{0,1\}^N$, respectively the states \eqref{bra-SoV}  for $\mathbf{h}\in\{0,1\}^N$, provide a basis of the space of states $\mathcal{H}$, respectively of the dual space of states $\mathcal{H}^*$, for almost any values\footnote{In particular, one should have $\mathsf{N}(\boldsymbol{\xi}|\alpha-\beta)\not= 0$, where $\mathsf{N}(\boldsymbol{\xi}|\alpha-\beta)$ is the normalisation coefficient defined in \eqref{norm-coeff}.} of the gauge parameters $\alpha$ and $\beta$.
Moreover, due to the very simple rewriting \eqref{Tinv-gaugeA} or \eqref{Tinv-gaugeD} of the transfer matrix $\mathcal{T}^\text{(Inv)}(\lambda)$ in terms of the gauge operators $\mathcal{A}_{-}(\pm\lambda |\alpha ,\beta-1 )$ or  $\mathcal{D}_{-}(\pm\lambda |\alpha ,\beta+1 )$, all these different bases separate the variables for the transfer matrix spectral problem in $\mathcal{H}$, respectively in $\mathcal{H}^*$, and this for almost any values of $\alpha$ and $\beta$ (the details of this approach then works as in the general case studied in \cite{KitMNT18}). Hence we have defined here a whole family of SoV bases for the transfer matrix $\mathcal{T}^\text{(Inv)}$.

Let us now discuss more precisely the dependence of the states \eqref{ket-SoV} and \eqref{bra-SoV} on $\alpha$ and $\beta$. We have shown that the gauge boundary operators $\widehat{\mathcal{B}}_{-}$ \eqref{Bhat} depend on the gauge parameters $\left( \alpha ,\beta \right) $ only through their difference $\alpha -\beta $. Thus, the SoV pseudo-eigenbasis of these gauge boundary operators $\widehat{\mathcal{B}}_{-}$, once properly normalised, should {\em a priori} also depend on the gauge parameters $\left( \alpha ,\beta \right) $ only via their difference $\alpha -\beta $. 
Here we use the new SoV basis \cite{MaiN18,MaiN19}, constructed using the transfer matrix $\mathcal{T}^\text{(Inv)}$, to demonstrate that this is indeed the case and that, in fact, all dependence on $\alpha -\beta$ in Sklyanin’s generalised SoV basis is contained within the definition of the pseudo-reference state.
%As announced above, let us now show that the states \eqref{ket-SoV} and \eqref{bra-SoV} only depend on the difference $\alpha-\beta$ of the two gauge parameters $\alpha$ and $\beta$, and not on their sum.

Let us first remark that the transformed gauge pseudo-reference states appearing in the definition of \eqref{ket-SoV} and \eqref{bra-SoV} only depend on $\alpha$ and $\beta$ through the difference $\alpha-\beta$:
\begin{align}
  &S_{1\ldots N}(\boldsymbol{\xi} |\alpha,\beta)\,\ket{\underline{0}} 
  =\begin{pmatrix} \,e^{-\xi _{n}-\eta (\alpha -\beta -1+n )}\, \\  1 \end{pmatrix}_{\!n},  
  \label{Pseudo-R} \\
  &\bra{0}\, S^\mathrm{Adj}_{1\ldots N}(\boldsymbol{\xi} |\alpha,\beta)
  = \otimes_{n=1}^N   \begin{pmatrix}  1 & -e^{-\xi_n-\eta(\alpha-\beta+1-n)}  \end{pmatrix}_n .
%  =\otimes_{n=1}^{\mathsf{N}}\left( -e^{\xi _{n}},e^{-\eta (\alpha +1-\beta-n)}\,\right) _{n}.
  \label{Pseudo-L}
\end{align}
We now introduce the following left and right SoV basis in the new SoV approach of \cite{MaiN19}:
\begin{align}
& \ket{\mathbf{h},\alpha -\beta -1 }
    \equiv  \prod_{j=1}^{N}\left( \frac{\mathcal{T}^\text{(Inv)}(\xi_j+\frac\eta 2)}{\mathsf{n}(\xi_j+\frac\eta 2)\, \mathsf{k}_j\, \mathsf{A}_-(\frac\eta 2-\xi_j)}\right) ^{h_j} S_{1\ldots N}(\boldsymbol{\xi} |\alpha,\beta)\,\ket{\underline{0}}, 
     \label{ket-Tinv}\\
& \bra{\alpha -\beta +1,\mathbf{h} }
   \equiv \bra{0}\, S^\mathrm{Adj}_{1\ldots N}(\boldsymbol{\xi} |\alpha,\beta)
   \prod_{j=1}^{N}\left( \frac{\mathcal{T}^\text{(Inv)}(\frac\eta 2-\xi_j)}{\mathsf{n}(\frac\eta 2-\xi_j)\, \mathsf{A}_-(\frac\eta 2-\xi_j)}\right) ^{1-h_j},
   \label{bra-Tinv}
\end{align}
where
\begin{equation}
  \mathsf{n}(\lambda )=e^{-\lambda +\eta /2}\,\frac{\sinh (2\lambda +\eta )}{\sinh(2\lambda )}.
\end{equation}
%
%and $\mathcal{T}^\text{(Inv)}$ is the transfer matrix constructed from $\mathcal{U}_-$ and $K_+^\text{(Inv)}$ \eqref{gauge-inv-K} as in  \eqref{Tinv-tr}.
Let us point out that, by definition, these SoV bases depend on $\alpha$ and $\beta$ only through their difference $\alpha -\beta $, and that this dependence is exclusively contained in the transformed gauge pseudo-reference states $S_{1\ldots N}(\boldsymbol{\xi} |\alpha,\beta)\,\ket{\underline{0}} $ and $\bra{0}\, S^\mathrm{Adj}_{1\ldots N}(\boldsymbol{\xi} |\alpha,\beta)$, see \eqref{Pseudo-R}-\eqref{Pseudo-L}.

Then, we have the following result:

\begin{proposition}\label{prop-id-gauge}
For any given choice of the gauge parameters $(\alpha ,\beta )$,  the state \eqref{ket-SoV} (respectively \eqref{bra-SoV}) coincides with the state \eqref{ket-Tinv} (respectively \eqref{bra-Tinv}):
\begin{equation} \label{gauged-SoV-identity}
   \ket{\mathbf{h},\alpha ,\beta +1}_{\mathrm{Sk}}=\ket{\mathbf{h},\alpha -\beta-1} ,\quad
   {}_{\mathrm{Sk}}\bra{\alpha ,\beta -1,\mathbf{h}}=\bra{ \alpha -\beta +1,\mathbf{h}},
   \quad \forall \mathbf{h}\in\{0,1\}^N.
\end{equation}
For almost any value of $\alpha-\beta$, these states form a basis of $\mathcal{H}$, respectively of $\mathcal{H}^*$, and satisfy the following orthogonality condition:
\begin{equation}
\langle \,\alpha -\beta +1, \mathbf{h}\,|\,\mathbf{h}^{\prime },\alpha -\beta-1\,\rangle 
=\delta _{\mathbf{h},\mathbf{h}^{\prime }}\,
\frac{\mathsf{N}(\boldsymbol{\xi}|\alpha-\beta)\,e^{2\sum_{j=1}^{N}h_{j}\xi_{j}}}{\widehat{V}(\xi _{1}^{(h_{1})},\ldots,\xi _{N}^{(h_{N})})},  \label{Ortho-norm}
\end{equation}
in which 
\begin{align}\label{norm-coeff}
  \frac{\mathsf{N}(\boldsymbol{\xi}|\alpha-\beta)}{ \widehat{V}(\xi_{1}^{(0)},\ldots ,\xi _{N}^{(0)}) }
%  &=\frac{\mathsf{N}(\boldsymbol{\xi}|\alpha,\beta)}{ \widehat{V}(\xi_{1}^{(0)},\ldots ,\xi _{N}^{(0)}) }  \nonumber\\
  &=
  \bra{0}\, S^\mathrm{Adj}_{1\ldots N}(\boldsymbol{\xi} |\alpha,\beta)
   \left(\prod_{j=1}^{N} \frac{\mathcal{T}^\mathrm{(Inv)}(\frac\eta 2-\xi_j)}{\mathsf{n}(\frac\eta 2-\xi_j)\, \mathsf{A}_-(\frac\eta 2-\xi_j)}\right) S_{1\ldots N}(\boldsymbol{\xi} |\alpha,\beta)\,\ket{\underline{0}}.
\end{align}
\end{proposition}

\begin{rem}
The normalisation coefficient \eqref{norm-coeff} can explicitly be computed, see Appendix~C of \cite{KitMNT18}. With the present choice of normalisation for the states \eqref{bra-SoV}, it reads %{\bf (to check)}
\begin{align}\label{expr=norm}
    \frac{\mathsf{N}(\boldsymbol{\xi}|\alpha-\beta)}{ \widehat{V}(\xi_{1}^{(0)},\ldots ,\xi _{N}^{(0)}) }
    &= \bra{0}\, S^\mathrm{Adj}_{1\ldots N}(\boldsymbol{\xi} |\alpha,\beta)   \left(\prod_{n=1}^{N} \frac{\mathcal{A}_-(\frac\eta 2-\xi_n |\alpha,\beta-1)}{\mathsf{A}_-(\frac\eta 2-\xi_n)} \right)   S_{1\ldots N}(\boldsymbol{\xi} |\alpha,\beta)\,\ket{\underline{0}} \nonumber\\
%  &=\frac{\mathsf{N}(\boldsymbol{\xi}|\alpha,\beta)}{ \widehat{V}(\xi_{1}^{(0)},\ldots ,\xi _{N}^{(0)}) }
  %\mathsf{N}(\boldsymbol{\xi}|\alpha-\beta) &= \mathsf{N}(\boldsymbol{\xi}|\alpha,\beta)   \nonumber\\
%   &= (-1)^N \frac{\widehat{V}(\xi_{1},\ldots ,\xi_{N}) }{\widehat{V}(\xi_{1}^{(1)},\ldots ,\xi _{N}^{(1)}) }   \prod_{j=1}^N \left[\frac{\widehat{\mathsf b}(\frac\eta 2-\xi_j|\alpha-\beta-1-N+2j)}{\mathsf{g}_-(\frac\eta 2-\xi_j)}\frac{\det S(-\xi_j|\alpha,\beta+j-1)\,\sinh(\eta(\beta+1+N-2j))}{\det S(\xi_j|\alpha,\beta+1+N-2j)\,\sinh(\eta(\beta+N-j))} \right] \nonumber\\
   &= (-1)^N \frac{\widehat{V}(\xi_{1},\ldots ,\xi_{N}) }{\widehat{V}(\xi_{1}^{(1)},\ldots ,\xi _{N}^{(1)}) }
   \prod_{j=1}^N \frac{\widehat{\mathsf b}_-(\frac\eta 2-\xi_j|\alpha-\beta-1-N+2j)}{e^{2\xi_j}\, \mathsf{g}_-(\frac\eta 2-\xi_j)},
\end{align}
which indeed depends on $\alpha$ and $\beta$ through the difference $\alpha-\beta$ only.
\end{rem}

\begin{proof}
The equality \eqref{gauged-SoV-identity} follows from the explicit form \eqref{Tinv-gaugeA}-\eqref{Tinv-gaugeD} of the transfer matrix $\mathcal{T}^\text{(Inv)}$ for any choice of the gauge parameters $( \alpha ,\beta) $,  in which the dependence on these parameters is only contained in the gauged boundary operators $\mathcal{A}_{-}(\lambda |\alpha ,\beta -1)$ and $\mathcal{D}_{-}(\lambda |\alpha ,\beta +1)$. So, as a corollary of Theorem~2.1 and Theorem~3.1 of \cite{MaiN19}, and following the same line of the proof as in Proposition~3.1 of \cite{Nic21}, one can derive the identification of the SoV bases in \eqref{gauged-SoV-identity}.
\end{proof}

In particular, one can consider the following  ungauged bases,
\begin{align}
& \ket{\mathbf{h} } = \prod_{j=1}^{N}\left( \frac{\mathcal{T}^\text{(Inv)}(\xi_j+\frac\eta 2)}{\mathsf{n}(\xi_j+\frac\eta 2)\, \mathsf{k}_j\, \mathsf{A}_-(\frac\eta 2-\xi_j)}\right) ^{h_j} \,\ket{\underline{0}}
 = \prod_{j=1}^{N}\left( \frac{\mathcal{D}(\xi_j+\frac\eta 2)}{\mathsf{k}_j\, \mathsf{A}_-(\frac\eta 2-\xi_j)}\right) ^{h_j} \,\ket{\underline{0}}, 
  \label{ungauge-ket}\\
&\bra{\mathbf{h} }
   = \bra{0}\, 
   \prod_{j=1}^{N}\left( \frac{\mathcal{T}^\text{(Inv)}(\frac\eta 2-\xi_j)}{\mathsf{n}(\frac\eta 2-\xi_j)\, \mathsf{A}_-(\frac\eta 2-\xi_j)}\right) ^{1-h_j}
   = \bra{0}\, 
   \prod_{j=1}^{N}\left( \frac{e^{-2\xi_j}\, \mathcal{A}(\frac\eta 2-\xi_j)}{ \mathsf{A}_-(\frac\eta 2-\xi_j)}\right) ^{1-h_j}  , \label{ungauge-bra}
\end{align}
in which we have used the expression \eqref{Tinv} of the transfer matrix to identify the two members of \eqref{ungauge-ket}, respectively of \eqref{ungauge-bra}, by using the same reasoning as above. These ungauged bases can  easily be obtained as a limit of the gauged bases:

\begin{corollary}\label{cor-limit}
The following limit holds:
\begin{align}
   &\ket{\mathbf{h}} 
   =\lim_{\substack{ \eta (\alpha -\beta )\rightarrow+\infty  \\ \text{with }\alpha +\beta \text{ finite}}} \ket{\mathbf{h},\alpha,\beta +1}_{\mathrm{Sk}}
   =\lim_{\substack{ \eta (\alpha -\beta )\rightarrow+\infty  \\ \text{with }\alpha +\beta \text{ finite}}}
   \ket{\mathbf{h},\alpha -\beta -1} ,
    \label{lim-rstate}\\
  &\bra{\mathbf{h} }
   =\lim_{\substack{ \eta (\alpha -\beta )\rightarrow+\infty  \\ \text{with }\alpha +\beta \text{ finite}}}   { }_{\mathrm{Sk}}\bra{\alpha ,\beta -1,\mathbf{h}}
   =\lim_{\substack{ \eta (\alpha -\beta )\rightarrow+\infty  \\ \text{with }\alpha +\beta \text{ finite}}}  \bra{ \alpha -\beta +1,\mathbf{h}}.
   \label{lim-lstate}
\end{align}
\end{corollary}

Note that, in this limit (which is also the limit considered in Remark~\ref{rem-limit}), we obtain a SoV basis which diagonalises the operator $\mathcal{B}_-(\lambda)$:
\begin{align}
  &\mathcal{B}_-(\lambda)\, \ket{\mathbf{h}}
  =(-1)^N a_\mathbf{h}(\lambda)\,a_\mathbf{h}(-\lambda)\,  b_-(\lambda)\,
  \ket{\mathbf{h}}, \label{eigenr-Blim}
  \\
  & \bra{\mathbf{h}}\, \mathcal{B}_-(\lambda)
  =(-1)^N a_\mathbf{h}(\lambda)\,a_\mathbf{h}(-\lambda)\, b_-(\lambda)\,
   \bra{\mathbf{h}},  \label{eigenl-Blim}
\end{align}
in which 
\begin{equation}
    b_-(\lambda)=\lim_{\eta(\alpha-\beta)\to +\infty} \widehat{\mathsf b}_-(\lambda|\alpha-\beta)
    =\frac{e^{\tau_-}\,\sinh(2\lambda-\eta)}{2\sinh\varphi_-\cosh\psi_-}
    =\frac{\kappa_- e^{\tau_-}\,\sinh(2\lambda-\eta)}{\sinh\varsigma_-}
\end{equation}
is the upper-right off-diagonal element of the boundary matrix $K_-(\lambda)$.
The states \eqref{ungauge-ket}-\eqref{ungauge-bra} moreover satisfy the orthogonality condition
\begin{equation}
\langle \,\mathbf{h}\,|\,\mathbf{h}^{\prime }\,\rangle 
=\delta _{\mathbf{h},\mathbf{h}^{\prime }}\,
\frac{\mathsf{N}(\boldsymbol{\xi} )\,e^{2\sum_{j=1}^{N}h_{j}\xi_{j}}}{\widehat{V}(\xi _{1}^{(h_{1})},\ldots,\xi _{N}^{(h_{N})})},  \label{Ortho-ungauge}
\end{equation}
in which 
\begin{align}\label{norm-ungauge}
  \frac{\mathsf{N}(\boldsymbol{\xi})}{ \widehat{V}(\xi_{1}^{(0)},\ldots ,\xi _{N}^{(0)}) }
  &=\lim_{\eta(\alpha-\beta)\to +\infty} 
  \frac{\mathsf{N}(\boldsymbol{\xi}|\alpha-\beta)}{ \widehat{V}(\xi_{1}^{(0)},\ldots ,\xi _{N}^{(0)}) }
  \nonumber\\
%  &=  \bra{0}\,  \left(\prod_{j=1}^{N} \frac{\mathcal{T}^\mathrm{(Inv)}(\frac\eta 2-\xi_j)}{\mathsf{n}(\frac\eta 2-\xi_j)\, \mathsf{A}_-(\frac\eta 2-\xi_j)}\right) \ket{\underline{0}}\\
  &=  (-1)^N \frac{\widehat{V}(\xi_{1},\ldots ,\xi_{N}) }{\widehat{V}(\xi_{1}^{(1)},\ldots ,\xi _{N}^{(1)}) }
   \prod_{j=1}^N \frac{b_-(\frac\eta 2-\xi_j)}{e^{2\xi_j}\, \mathsf{g}_-(\frac\eta 2-\xi_j)}.
\end{align}
%

%%%%%%%%
\subsection{Transfer matrix spectrum and eigenstates}

From now on, we consider the following special choice of the function $\mathsf{g}_-$ \eqref{fct-g}, given in terms of the boundary parameters $\varphi_-,\psi_-$ and of a sign $\varepsilon\in\{+,-\}$ as
\begin{equation}\label{geps}
\mathsf{g}_-^{(\varepsilon )}(\lambda +\eta /2)=(-1)^{N}\frac{\sinh(\lambda +\varepsilon \varphi_-)\,\cosh (\lambda +\varepsilon \psi_-)}{\sinh (\varepsilon \varphi_-)\,\cosh (\psi_-)},
\end{equation}
and we define, for each value of $\varepsilon\in\{+,-\}$, the coefficient
\begin{equation}\label{Aeps}
\mathbf{A}^\mathrm{(Inv)}_\varepsilon(\lambda )=\,\frac{\sinh (2\lambda +\eta )}{\sinh (2\lambda )}\, \mathsf{A}_-^{(\varepsilon )}(\lambda ),
\qquad  \text{with}\quad
\mathsf{A}_-^{(\varepsilon )}(\lambda )=\mathsf{g}_-^{(\varepsilon )}(\lambda)\,a(\lambda )\,d(-\lambda ).
\end{equation}
Let $\Sigma _{Q}^{M}$ be the set of all the polynomial in $\sinh ^{2}\lambda $ of the form
\begin{equation}
Q(\lambda )=\prod_{j=1}^{M}\frac{\cosh (2\lambda )-\cosh (2\lambda _{j})}{2}
=\prod_{j=1}^{M}\left( \sinh ^{2}\lambda -\sinh ^{2}\lambda _{j}\right) ,
\label{Polyform}
\end{equation}
for some set of roots $\lambda_1,\ldots,\lambda_M\in\mathbb{C}$ such that
\begin{equation}
\left( Q(\xi _{n}^{(0)}),Q(\xi _{n}^{(1)})\right) \neq (0,0)\qquad
\forall n\in \{1,\ldots,N\},
\end{equation}
and we also denote
\begin{equation}
\Sigma _{Q}=\cup _{n=0}^{N}\Sigma _{Q}^{n}.
\end{equation}

Then, we can formulate the following complete characterisation of the spectrum and eigenstates of the transfer matrix $\mathcal{T}^\mathrm{(Inv)}(\lambda )$:

\begin{theorem}
\label{Theo-Sp-Inv}
Under the condition \eqref{cond-inh}, and for almost any choice of the boundary parameters, the transfer matrix $\mathcal{T}^\mathrm{(Inv)}(\lambda )$ is diagonalizable with simple spectrum.
Moreover, a function $\tau(\lambda )$ is an eigenvalue of the transfer matrix $\mathcal{T}^\mathrm{(Inv)}(\lambda )$ if and only if it is an entire function of $\lambda $ such that, for any $\varepsilon\in\{+,-\}$, there exists a unique $Q^{(\varepsilon)}(\lambda )\in \Sigma _{Q}$ satisfying the homogeneous $TQ$-equation
\begin{equation}
    \tau(\lambda )\, Q^{(\varepsilon )}(\lambda )=\mathbf{A}^\mathrm{(Inv)}_\varepsilon(\lambda )\, Q^{(\varepsilon )}(\lambda -\eta )+\mathbf{A}^\mathrm{(Inv)}_\varepsilon(-\lambda )\, Q^{(\varepsilon )}(\lambda + \eta ).
\label{TQ-hom}
\end{equation}
For almost any value of $\alpha-\beta$, the corresponding unique right and left eigenstates (up to a global arbitrary normalisation factor) % $\mathsf{R}_{Q^{(\varepsilon)}}^{(\alpha-\beta)}$, respectively $\mathsf{L}_{Q^{(\varepsilon)}}^{(\alpha-\beta)}$) 
can be written as
\begin{align}
&  \ket{Q^{(\varepsilon)},\alpha-\beta-1 }_\varepsilon 
= \sum_{\mathbf{h}\in\{0,1\}^{N}}\prod_{n=1}^{N} Q^{(\varepsilon )}(\xi _{n}^{(h_{n})})\
e^{-\sum_{j}h_{j}\xi _{j}}\,\widehat{V}(\xi _{1}^{(h_{1})},\ldots ,\xi_{N}^{(h_{N})})\ \ket{\mathbf{h},\alpha-\beta-1}_\varepsilon  ,  
      \label{eigenvect-SoV} \\
& {}_\varepsilon\bra{ \alpha-\beta+1, Q^{(\varepsilon )} }
= \sum_{\mathbf{h}\in\{0,1\}^{N}}\prod_{n=1}^{N}\left[ \left(\frac{\sinh (2\xi _{n}-2\eta )}{\sinh(2\xi _{n}+2\eta )}\,\frac{\mathbf{A}^\mathrm{(Inv)}_\varepsilon(\xi _{n}+\frac{\eta }{2})}{\mathbf{A}^\mathrm{(Inv)}_\varepsilon(-\xi _{n}+\frac{\eta }{2})} \right)^{\! h_{n}}\ Q^{(\varepsilon )}(\xi _{n}^{(h_{n})})\right]  \notag \\
& \hspace{6cm} \times e^{-\sum_{j}h_{j}\xi _{j}}\,\widehat{V}(\xi _{1}^{(h_{1})},\ldots,\xi _{N}^{(h_{N})})\ {}_\varepsilon\bra{\alpha-\beta+1, \mathbf{h}}, 
\label{eigencovect-SoV}
\end{align}
on the SoV bases \eqref{ket-SoV} and \eqref{bra-SoV}, or equivalently \eqref{ket-Tinv}-\eqref{bra-Tinv}, with the choice \eqref{geps} for $\mathsf{g}_- = \mathsf{g}_-^{(\varepsilon )}$.

Finally, the transfer matrix $\mathcal{T}^\mathrm{(Inv)}(\lambda )$ is isospectral to an open spin chain with both longitudinal boundary fields with boundary parameters $\varsigma _{\pm }^{(D)}$ given as 
\begin{equation}\label{iso-diag}
\varsigma _{\epsilon_D }^{(D)}=\varphi_-,\qquad \varsigma _{-\epsilon_D}^{(D)}=\psi_-,
\end{equation}
for $\epsilon_D =1$ or $-1.$
\end{theorem}

\begin{proof}
This is a standard result, see \cite{Nic12,Nic13,FalKN14,KitMN14}. In the SoV framework, it is first proven that all the eigenstates of the transfer matrix have the above form \eqref{eigenvect-SoV} and \eqref{eigencovect-SoV} with coefficients satisfying the identities
\begin{equation}
\frac{Q^{(\varepsilon )}(\xi _{n}^{(1)})}{Q^{(\varepsilon )}(\xi_{n}^{(0)})}
=\frac{\tau(\xi _{n}^{(0)})}{\mathbf{A}^\mathrm{(Inv)}_\varepsilon(\xi _{n}^{(0)})}=\frac{\mathbf{A}^\mathrm{(Inv)}_\varepsilon(-\xi _{n}^{(1)})}{\tau(\xi _{n}^{(1)})},
\qquad n\in\{1,\ldots,N\}.
\label{eq-Q-dis}
\end{equation}
Then, the characterisation of the transfer matrix eigenvalues given by the system of equations \eqref{eq-Q-dis} is proven to be equivalent to one given by the homogeneous functional $TQ$-equation \eqref{TQ-hom}  \cite{KitMN14}, see also \cite{NicT22}. 
\end{proof}

In particular, taking the limit of the gauge parameters as in Corollary~\ref{cor-limit}, one can write the right and left eigenstates \eqref{eigenvect-SoV} and \eqref{eigencovect-SoV} on the ungauged bases \eqref{ungauge-ket} and \eqref{ungauge-bra} as
\begin{align}
&  \ket{Q^{(\varepsilon)} }_\varepsilon  = %\ket{Q^{(\varepsilon)} }= \frac{1}{\mathsf{R}_{Q^{(\varepsilon)}} }
\sum_{\mathbf{h}\in\{0,1\}^{N}}\prod_{n=1}^{N} Q^{(\varepsilon )}(\xi _{n}^{(h_{n})})\
e^{-\sum_{j}h_{j}\xi _{j}}\,\widehat{V}(\xi _{1}^{(h_{1})},\ldots ,\xi_{N}^{(h_{N})})\ \ket{\mathbf{h}}_\varepsilon ,  
      \label{r-ungauge} \\
&  {}_\varepsilon\bra{  Q^{(\varepsilon )} } = %\bra{  Q^{(\varepsilon )} }= \frac{1}{\mathsf{L}_{Q^{(\varepsilon)}} }
\sum_{\mathbf{h}\in\{0,1\}^{N}}\prod_{n=1}^{N}\left[ \left(\frac{\sinh (2\xi _{n}-2\eta )}{\sinh(2\xi _{n}+2\eta )}\,\frac{\mathbf{A}^\mathrm{(Inv)}_\varepsilon(\xi _{n}+\frac{\eta }{2})}{\mathbf{A}^\mathrm{(Inv)}_\varepsilon(-\xi _{n}+\frac{\eta }{2})} \right)^{\! h_{n}}\ Q^{(\varepsilon )}(\xi _{n}^{(h_{n})})\right]  \notag \\
& \hspace{7cm} \times e^{-\sum_{j}h_{j}\xi _{j}}\,\widehat{V}(\xi _{1}^{(h_{1})},\ldots,\xi _{N}^{(h_{N})})\ {}_\varepsilon\bra{\mathbf{h}},  
\label{l-ungauge}
\end{align}
up to a global arbitrary normalisation factor.

%\begin{rem}
For $Q^{(\varepsilon)}(\lambda )\in \Sigma _{Q}$ solution of the $TQ$-equation \eqref{TQ-hom} with a given eigenvalue $\tau(\lambda)$, the simplicity of the transfer matrix spectrum implies that all the states \eqref{eigenvect-SoV} for different values of $\alpha-\beta$ are proportional, and that they are proportional to the state \eqref{r-ungauge}. Similarly, under the same hypothesis, all the states \eqref{eigencovect-SoV} for different values of $\alpha-\beta$ are proportional, and are proportional to the state \eqref{l-ungauge}.
%\end{rem}

\begin{rem}\label{rem-eps}
In \eqref{eigenvect-SoV}  and \eqref{eigencovect-SoV}, we have explicitly specified, using the subscript $\varepsilon$,  the specific choice of $\varepsilon$ entering the definition of the SoV states \eqref{ket-SoV} and \eqref{bra-SoV}, or equivalently \eqref{ket-Tinv}-\eqref{bra-Tinv}, via the definition of $\mathsf{g}_-=\mathsf{g}_-^{(\varepsilon )}$ \eqref{geps}.  Note that the states $\ket{\mathbf{h},\alpha-\beta-1}_\varepsilon$ and $\ket{\mathbf{h},\alpha-\beta-1}_{-\varepsilon}$, respectively the states ${}_\varepsilon\bra{\alpha-\beta+1, \mathbf{h}}$ and ${}_{-\varepsilon}\bra{\alpha-\beta+1, \mathbf{h}}$, are proportional:
\begin{align}
   &\ket{\mathbf{h},\alpha-\beta-1}_{-\varepsilon}=\prod_{n=1}^N\left(\frac{\mathsf{g}_-^{(\varepsilon)}(\frac\eta 2-\xi_n)}{\mathsf{g}_-^{(-\varepsilon)}(\frac\eta 2-\xi_n)}\right)^{\! h_n}\, \ket{\mathbf{h},\alpha-\beta-1}_\varepsilon, \\
   &{}_{-\varepsilon}\bra{\alpha-\beta+1, \mathbf{h}}=\prod_{n=1}^N\left(\frac{\mathsf{g}_-^{(\varepsilon)}(\frac\eta 2-\xi_n)}{\mathsf{g}_-^{(-\varepsilon)}(\frac\eta 2-\xi_n)}\right)^{\! 1-h_n}\, {}_\varepsilon\bra{\alpha-\beta+1, \mathbf{h}}.
   \label{bra-eps-eps'}
\end{align}
Moreover, if $Q^{(\varepsilon)}$ and $Q^{(-\varepsilon)}$ are the respective solutions of the two $TQ$-equations \eqref{TQ-hom} corresponding to the same transfer matrix $\tau(\lambda)$ but with with coefficient $\mathbf{A}^\mathrm{(Inv)}_\varepsilon$ and $\mathbf{A}^\mathrm{(Inv)}_{-\varepsilon}$, the states $\ket{Q^{(\varepsilon)},\alpha-\beta-1 }_\varepsilon$ and $ \ket{Q^{(-\varepsilon)},\alpha-\beta-1 }_{-\varepsilon }$ are related by the following simple proportionality relation: %(see \cite{KitMNT18}),
\begin{equation}\label{prop-ket-Q}
    \ket{Q^{(-\varepsilon)},\alpha-\beta-1 }_{-\varepsilon }
    = \prod_{n=1}^N\frac{Q^{(-\varepsilon)}(\xi_n+\frac\eta 2)}{Q^{(\varepsilon)}(\xi_n+\frac\eta 2)}\ \ket{Q^{(\varepsilon)},\alpha-\beta-1 }_\varepsilon ,
\end{equation}
and a similar explicit %(although slightly more complicated) 
proportionality relation holds between ${}_{-\varepsilon}\bra{ \alpha-\beta+1, Q^{(-\varepsilon )} }$ and ${}_\varepsilon\bra{ \alpha-\beta+1, Q^{(\varepsilon )} }$: %see \cite{KitMNT18}.
\begin{equation}\label{prop-bra-Q}
    {}_{-\varepsilon }\bra{Q^{(-\varepsilon)},\alpha-\beta-1 } 
    = \prod_{n=1}^N \left[\frac{\mathsf{g}_-^{(\varepsilon)}(\frac\eta 2-\xi_n)}{\mathsf{g}_-^{(-\varepsilon)}(\frac\eta 2-\xi_n)}
    \frac{Q^{(-\varepsilon)}(\xi_n+\frac\eta 2)}{Q^{(\varepsilon)}(\xi_n+\frac\eta 2)}\right]\ {}_{\varepsilon }\bra{Q^{(\varepsilon)},\alpha-\beta-1 } .
\end{equation}
%
%In what follows, we shall in general assume that $\varepsilon$ is fixed and identical in the definition of the $TQ$-equation \eqref{TQ-hom} and in the SoV states \eqref{ket-SoV}-\eqref{ket-Tinv} and \eqref{bra-SoV}-\eqref{bra-Tinv}. Therefore, except if needed, we shall no longer specify its value using such a subscript on the corresponding states, so as not to unnecessarily complicate the notations.
\end{rem}

\begin{rem} 
By means of identity (4.47) of \cite{KitMNT18}, the state ${}_\varepsilon\bra{ \alpha-\beta+1, Q^{(\varepsilon )} }$ \eqref{eigencovect-SoV} can be rewritten as
\begin{multline}
  {}_\varepsilon\bra{ \alpha-\beta+1, Q^{(\varepsilon )} }
  = \frac{\widehat{V}(\xi _{1}^{(0)},\ldots,\xi _{N}^{(0)})}{\widehat{V}(\xi _{1}^{(1)},\ldots,\xi _{N}^{(1)})}
  \sum_{\mathbf{h}\in\{0,1\}^{N}}\prod_{n=1}^{N}\left[ \left(-\frac{\mathsf{g}_-^{(\varepsilon )}(\xi _{n}+\frac{\eta }{2})}{\mathsf{g}_-^{(\varepsilon )}(-\xi _{n}+\frac{\eta }{2})} \right)^{\! h_{n}}\, Q^{(\varepsilon )}(\xi _{n}^{(h_{n})})\right] 
 \\
         \times e^{-\sum_{j}h_{j}\xi _{j}}\,\widehat{V}(\xi _{1}^{(1-h_{1})},\ldots,\xi _{N}^{(1-h_{N})})\ {}_\varepsilon\bra{\alpha-\beta+1, \mathbf{h}}, 
\label{eigencovect-SoV-bis}
\end{multline}
and a similar representation also holds in the ungauged limit \eqref{l-ungauge}.
\end{rem}

Let us finally mention that the two Q-functions $Q^{(+)}$ and $Q^{(-)}$ solution of the $TQ$-equation \eqref{TQ-hom} for the same transfer matrix eigenvalue $\tau(\lambda)$ and with coefficient $\mathbf{A}^\mathrm{(Inv)}_+$ and $\mathbf{A}^\mathrm{(Inv)}_-$ respectively, satisfy some generalised quantum Wronskian identity which fixes the sum of their respective degree $M^{(+)}$ and $M^{(-)}$. More precisely, we have the following result, which is the direct XXZ analog to Proposition~3.1 written in \cite{KitMNT17} in the open XXX case:

\begin{proposition}\label{prop-Wr}
Under the same hypothesis as in the previous theorem, let us consider, for a given eigenvalue $\tau(\lambda)$ of the transfer matrix $\mathcal{T}^\mathrm{(Inv)}(\lambda )$, the two Q-functions $Q^{( \varepsilon) }\in\Sigma_Q^{M^{(\varepsilon)}}$ ($M^{(\varepsilon)}\le N$) solution of the $TQ$-equation \eqref{TQ-hom} corresponding to the two special choices of the functions  $\mathbf{A}^{(\varepsilon )}$ \eqref{Aeps}  and $\mathsf{g}^{(\varepsilon )}$ \eqref{geps} for $\varepsilon\in\{+,-\}$, and let us define their corresponding generalised quantum Wronskian as
\begin{equation}
W_{\tau}(\lambda )=\mathsf{g}_-^{(+)}(\lambda )\, Q^{(+)}(\lambda -\eta )\,Q^{(-)}(\lambda )-\mathsf{g}_-^{(-)}(\lambda )\, Q^{(-)}(\lambda -\eta )\,Q^{(+)}(\lambda ). \label{Wronsk}
\end{equation}
%
%a given transfer matrix eigenvalue $\tau(\lambda )$, satisfy the following Wronskian equation 
Then, we have the following Wronskian identity: %{\bf (to check again)}
\begin{equation}
W_{\tau}(\lambda )=-\frac{\cosh (\varphi_-+\psi_- -(M^{(+)}-M^{(-)})\eta )}{\sinh\varphi_-\cosh\psi_-}\,\sinh (2\lambda -\eta)\,a(-\lambda )\,d(\lambda ),
\end{equation}
and moreover
\begin{equation}
   M^{(-)}+M^{(+)}=N.  \label{QbarQ-degree}
\end{equation}
\end{proposition}

\begin{proof}
The proof follows closely the proof of Proposition~3.1 of \cite{KitMNT17} for the open XXX case, see also  \cite{KitMNT16}. By using the functional $TQ$-equation \eqref{TQ-hom}, one can write, for each $\varepsilon\in\{+,-\}$,
\begin{multline}
  \sinh (2\lambda )\, \tau(\lambda )\, Q^{(\varepsilon )}(\lambda)\, Q^{( -\varepsilon) }(\lambda )
  =\mathsc{a}(\lambda)\,\mathsf{g}_-^{(\varepsilon) }(\lambda )\, Q^{(-\varepsilon) }(\lambda )\, Q^{(\varepsilon) }(\lambda-\eta )
  \\
  -\mathsc{a}(-\lambda)\,\mathsf{g}_-^{( \varepsilon) }(-\lambda )\, Q^{( -\varepsilon) }(\lambda)\, Q^{( \varepsilon) }(\lambda +\eta ),
\end{multline}
in which 
\begin{equation}
    \mathsc{a}(\lambda )=\sinh (2\lambda +\eta )\,a(\lambda)\,d(-\lambda ).
\end{equation}
Hence, taking the difference of the two possible cases $\varepsilon =+$ or $-$, one obtains
\begin{equation}
\mathsc{a}(\lambda )\,W_\tau(\lambda )= \mathsc{a}(-\lambda )\,W_\tau (-\lambda ).
\end{equation}
Note that both sides of this equation are trigonometric polynomials in $\lambda $. Moreover, the roots of $\mathsc{a}(\lambda )$ are all distinct from the roots of $\mathsc{a}(-\lambda )$. It therefore follows that
\begin{equation}
W_\tau(\lambda )=w(\lambda )\, \mathsc{a}(-\lambda )
\end{equation}
with $w(\lambda )$ being an even trigonometric polynomial in $\lambda $. Using moreover the fact that
\begin{equation}
   \mathsc{a}(\lambda -\eta /2)=-\mathsc{a}(-\lambda-\eta /2),
   \qquad
   W_{\tau }(\lambda +\eta /2)=-W_{\tau }(-\lambda +\eta/2),
\end{equation}
one concludes that $w(\lambda +\eta/2)=w(-\lambda+\eta/2)$, which means that it is a constant. The value of this constant can be obtained from the leading asymptotic behaviour of $W_\tau(\lambda )$ and $\mathsc{a}(-\lambda )$ when $\lambda \rightarrow \pm \infty $, which also gives \eqref{QbarQ-degree}.
\end{proof}

%%%%%%%%%%%%%%%%%%%%%
\subsection{Separate states as generalised Bethe states}

As shown in \cite{KitMNT18,NicT22}, for any $Q\in\Sigma_Q$ (not necessarily solution of a $TQ$-equation), separate states of the form\footnote{Note that we underline explicitly, in the notations $\ket{Q,\alpha-\beta-1}$ of \eqref{r-separate} and $\bra{\alpha-\beta+1,  Q}$ of \eqref{l-separate}, the corresponding SoV basis on which these states are expressed as separate states. It is important to underline here that, for a generic $Q\in\Sigma_Q$, a separate state of the form \eqref{r-separate} or \eqref{l-separate} may {\em a priori} no longer be a separate state when expressed on a different SoV basis with $\alpha'-\beta'\not=\alpha-\beta$. }
\begin{align}
  \ket{Q,\alpha-\beta-1} &= %\frac{1}{\mathsf{R}_{Q}^{(\alpha-\beta)} }
\sum_{\mathbf{h}\in\{0,1\}^{N}}\prod_{n=1}^{N} Q(\xi _{n}^{(h_{n})})\
e^{-\sum_{j}h_{j}\xi _{j}}\,\widehat{V}(\xi _{1}^{(h_{1})},\ldots ,\xi_{N}^{(h_{N})})\ \ket{\mathbf{h},\alpha-\beta-1} ,  
      \label{r-separate} \\
 \bra{\alpha-\beta+1,  Q} &= %\frac{1}{\mathsf{L}_{Q}^{(\alpha-\beta)} }
\sum_{\mathbf{h}\in\{0,1\}^{N}}\prod_{n=1}^{N}\left[ \left(\frac{\sinh (2\xi _{n}-2\eta )}{\sinh(2\xi _{n}+2\eta )}\,\frac{\mathbf{A}^\mathrm{(Inv)}_\varepsilon(\xi _{n}+\frac{\eta }{2})}{\mathbf{A}^\mathrm{(Inv)}_\varepsilon(-\xi _{n}+\frac{\eta }{2})}
%\frac{\mathsf{A}_-^{(\varepsilon )}(\xi _{n}+\frac{\eta }{2})}{\mathsf{A}_-^{(\varepsilon )}(-\xi _{n}+\frac{\eta }{2})} 
\right)^{\! h_{n}}\ Q(\xi _{n}^{(h_{n})})\right]  \notag \\
& \hspace{3.3cm} \times e^{-\sum_{j}h_{j}\xi _{j}}\,\widehat{V}(\xi _{1}^{(h_{1})},\ldots,\xi _{N}^{(h_{N})})\ \bra{\alpha-\beta+1, \mathbf{h}},
\nonumber\\
&= \frac{\widehat{V}(\xi _{1}^{(0)},\ldots,\xi _{N}^{(0)})}{\widehat{V}(\xi _{1}^{(1)},\ldots,\xi _{N}^{(1)})}
  \sum_{\mathbf{h}\in\{0,1\}^{N}}\prod_{n=1}^{N}\left[ \left(-\frac{\mathsf{g}_-^{(\varepsilon )}(\xi _{n}+\frac{\eta }{2})}{\mathsf{g}_-^{(\varepsilon )}(-\xi _{n}+\frac{\eta }{2})} \right)^{\! h_{n}}\, Q(\xi _{n}^{(h_{n})})\right] 
 \nonumber\\
  & \hspace{3.3cm}       \times e^{-\sum_{j}h_{j}\xi _{j}}\,\widehat{V}(\xi _{1}^{(1-h_{1})},\ldots,\xi _{N}^{(1-h_{N})})\ \bra{\alpha-\beta+1, \mathbf{h}},
\label{l-separate}
\end{align}
written on the basis \eqref{ket-SoV}, respectively \eqref{bra-SoV}, for a given value of $\alpha-\beta$, can be expressed in the form of  generalised Bethe states.

More precisely, it follows from the action \eqref{eigenr-Bhat}-\eqref{eigenl-Bhat} of \eqref{Bhat} on \eqref{ket-SoV} and \eqref{bra-SoV} that the separate states of the form \eqref{r-separate}-\eqref{l-separate} for any $Q_{\boldsymbol \lambda}\in \Sigma _{Q}^{M}$ with roots given by $\boldsymbol{\lambda}=\{\lambda_1,\ldots,\lambda_M\}$ can be written as the following generalised Bethe states:
\begin{align}
& \ket{Q_{\boldsymbol \lambda},\alpha-\beta-1} =\mathsf{c}_{Q_{\boldsymbol \lambda},\alpha-\beta-1}^{(R)}\ \underline{\widehat{\mathcal{B}}}_{-,M}(\{\lambda_{i}\}_{i=1}^{M}|\alpha -\beta +1)\,\ket{\Omega _{\alpha -\beta -1+2M} } ,
\label{separate-ABA-R} \\
& \bra{\alpha-\beta+1, Q_{\boldsymbol \lambda}}=\mathsf{c}_{Q_{\boldsymbol \lambda},\alpha-\beta+1}^{(L)} \ \bra{ \Omega _{\alpha -\beta +1-2M} }\,\underline{\widehat{\mathcal{B}}}_{-,M}(\{\lambda _{i}\}_{i=1}^{M}|\alpha -\beta +1-2M),
\label{separate-ABA-L}
\end{align}
with {%\bf (to check)}
\begin{align}
  & \mathsf{c}_{Q_{\boldsymbol \lambda},\alpha-\beta-1}^{(R)} %\mathsf{R}_Q^{( \alpha -\beta ) }
  =\frac{(-1)^{NM}\, \mathsf{N}(\boldsymbol{\xi}| \alpha-\beta+2M )}{\prod_{j=1}^{M} \widehat{\mathsf b}_-(\lambda_j|\alpha -\beta +2j+N-1)}, 
    \label{norm-ABA-R} \\
  & \mathsf{c}_{Q_{\boldsymbol \lambda},\alpha-\beta+1}^{(L)} %\mathsf{L}_Q^{( \alpha -\beta ) }
  =\frac{(-1)^{NM}\,\mathsf{N}(\boldsymbol{\xi}| \alpha-\beta-2M )}{\prod_{j=1}^{M}\widehat{\mathsf b}_-(\lambda_j|\alpha -\beta -2j-N+1)}.  \label{norm-ABA-L}
\end{align}
Here we have used the shortcut notation
\begin{align}
\underline{\widehat{\mathcal{B}}}_{-,M}(\{\lambda _{i}\}_{i=1}^{M}|\alpha-\beta +1)& =\widehat{\mathcal{B}}_{-}(\lambda _{1}|\alpha -\beta +1)\cdots \widehat{\mathcal{B}}_{-}(\lambda _{M}|\alpha -\beta +2M-1)  \notag \\
& =\prod_{j=1\rightarrow M}\widehat{\mathcal{B}}_{-}(\lambda _{j}|\alpha-\beta +2j-1),
\end{align}
and the pseudo reference states in \eqref{separate-ABA-R}-\eqref{separate-ABA-L} are defined as the following particular separate states:
\begin{align}
   \ket{\Omega_{x} } &=\frac{1}{\mathsf{N}(\boldsymbol{\xi}| x+1 )}
   \sum_{\mathbf{h}\in\{0,1\}^N}e^{-\sum_j h_j \xi_j }\,\widehat{V}(\xi_{1}^{(h_{1})},\ldots ,\xi _{N}^{(h_{N})})\  \ket{\mathbf{h},x } , 
   \label{R-ref-state}\\
   \bra{\Omega_{x}} &=\frac{1}{\mathsf{N}(\boldsymbol{\xi}| x -1)}
   \sum_{\mathbf{h}\in\{0,1\}^N}\prod_{n=1}^{N}\left[ \frac{\sinh (2\xi _{n}-2\eta )}{\sinh(2\xi _{n}+2\eta )}\,\frac{\mathbf{A}^\mathrm{(Inv)}_\varepsilon(\xi _{n}+\frac{\eta }{2})}{\mathbf{A}^\mathrm{(Inv)}_\varepsilon(-\xi _{n}+\frac{\eta }{2})}\right]^{h_{n}}\,
   \nonumber\\
   &\hspace{6cm}\times
   e^{-\sum_j h_j \xi_j }\, \widehat{V}(\xi_{1}^{(h_{1})},\ldots ,\xi _{N}^{(h_{N})})\ \bra{x,\mathbf{h}},
   \nonumber\\
   &=\frac{1}{\mathsf{N}(\boldsymbol{\xi}| x -1)} \frac{\widehat{V}(\xi _{1}^{(0)},\ldots,\xi _{N}^{(0)})}{\widehat{V}(\xi _{1}^{(1)},\ldots,\xi _{N}^{(1)})}
  \sum_{\mathbf{h}\in\{0,1\}^{N}}\prod_{n=1}^{N} \left[-\frac{\mathsf{g}_-^{(\varepsilon )}(\xi _{n}+\frac{\eta }{2})}{\mathsf{g}_-^{(\varepsilon )}(-\xi _{n}+\frac{\eta }{2})} \right]^{h_{n}}
   \nonumber\\
   &\hspace{6cm}\times
   e^{-\sum_j h_j \xi_j }\, \widehat{V}(\xi_{1}^{(1-h_{1})},\ldots ,\xi _{N}^{(1-h_{N})})\ \bra{x,\mathbf{h}}.
    \label{L-ref-state}
\end{align}
Note that it follows from the previous results on the $(\alpha,\beta)$-dependance of the basis elements \eqref{ket-SoV}-\eqref{bra-SoV} that the pseudo reference states appearing in \eqref{separate-ABA-R}-\eqref{separate-ABA-L} depend on $\alpha-\beta$ only (and not on $\alpha+\beta$). Note also that the relations \eqref{separate-ABA-R} and \eqref{separate-ABA-L} hold whatever the chosen form for $\mathsf{g}_-$ satisfying \eqref{fct-g}. 
%Let us now be more precise about the $\alpha-\beta$ dependence of the generalised Bethe states \eqref{separate-ABA-R}-\eqref{separate-ABA-R}.

%%%%%%%
%\subsubsection{The pseudo-reference states}

Let us be more precise about the pseudo-reference states \eqref{R-ref-state} and \eqref{L-ref-state} in the case in which $\mathsf{g}_-=\mathsf{g}_-^{(\varepsilon)}$ is chosen as in \eqref{geps}.
To this aim, let us introduce, as in \cite{NicT22}, the following notations: 
\begin{align}\label{ref-x}
   \ket{\eta, y}&= \otimes_{n=1}^N 
   \begin{pmatrix} e^{-\xi_n-\eta(y+n-N)} \\ 1 \end{pmatrix}_n,
   \qquad
   \bra{y,\eta}=  \otimes_{n=1}^N \begin{pmatrix} -1 & e^{-\xi_n-\eta(y-n+N)} \end{pmatrix}_n.
 \end{align}
The states \eqref{ref-x} result from the action of the Vertex-IRF transformation \eqref{Sgauge-N} or \eqref{Sgauge-Adj-N} on the reference states $\ket{0}$ or $\bra{\underline 0}$. More precisely:
\begin{align}\label{ref-state-actS}
   &S_{1\ldots N}(\boldsymbol{\xi} |\alpha ,\beta)\, \ket{ 0 }=\ket{\eta,\alpha+\beta+N-1} , %\label{Right-C-ref}
   \qquad \bra{\underline 0}\,S_{1\ldots N}^\mathrm{Adj}( \boldsymbol{\xi} |\alpha ,\beta ) = \bra{\alpha+\beta-N+1,\eta}. %\label{Left-B-ref} .  
\end{align}
In \cite{NicT22}, we have shown that, under some particular constraint relating $\alpha+\beta$ and the boundary parameters entering the parametrisation of the K-matrix $K_-$, the SoV pseudo reference state $\ket{\Omega_{\alpha-\beta+2M-1}}$\footnote{The SoV pseudo reference state $\ket{\Omega_{\alpha-\beta+2M-1}}$ (respectively $\bra{\Omega_{\alpha-\beta-2M+1}}$) differs in fact from the SoV pseudo reference states $\ket{\Omega_{\alpha,\beta+1-2M}}$ (respectively $\bra{\Omega_{\alpha,\beta-1+2M}}$) considered in \cite{NicT22} by some normalisation factor, due to the different normalisation conventions that we have chosen here for the basis \eqref{bra-SoV}.}. and the state $\ket{\eta,\alpha+\beta+N-1-2M}$ were proportional, see Proposition~3.4 of \cite{NicT22}. A similar relation, a priori subject to a different constraint between $\alpha+\beta$ and the $-$ boundary parameters, was also obtained there  between $\bra{\Omega_{\alpha-\beta-2M+1}}$ and $\bra{\alpha+\beta-N+2M+1,\eta}$. Since we now have the additional information (with respect to \cite{NicT22}) that the SoV pseudo reference states defined as above do not depend on $\alpha+\beta$, we can reformulate these results in a way which no longer involves the aforementioned constraints.  We moreover here explicitly compute  the proportionality coefficients:

\begin{proposition}\label{Norm-references}
For $\mathsf{g}_-=\mathsf{g}_-^{(\varepsilon)}$ defined as in \eqref{geps}, the right and left pseudo reference states $\ket{\Omega_x}\equiv  \ket{\Omega_x}_\varepsilon$ \eqref{R-ref-state} and $\bra{\Omega_x}\equiv {}_\varepsilon\bra{\Omega_x}$ \eqref{L-ref-state} depend on $x$ only through an overall factor. They admit the following explicit tensor product representation involving the states \eqref{ref-x} defined in terms of specific combinations %$y_R^{(\varepsilon)}$ and $y_L^{(\varepsilon)}$ 
of the $K_{-}$ boundary parameters:
\begin{alignat}{2}
   &\ket{\Omega_x}_\varepsilon =  %mathsf{\widetilde c}_{x,y_R}^{(R)}  
   \mathsf{d}_{x}^{(R,\varepsilon)}\, \ket{\eta , y_R^{(\varepsilon)} }, \qquad
   \quad & &\text{with}\quad  \eta \, y_R^{(\varepsilon)} =-\tau_--\varepsilon(\varphi_-+\psi_-)+\frac{1-\varepsilon}2 i\pi,
  % y_R=\left( -\tau_--\varepsilon (\varphi_-+\psi_-)+\frac{1-\varepsilon }{2}i\pi \mod2\pi i\right) /\eta ,  
   \label{Ref-SoV-R}\\
   &{}_\varepsilon\bra{\Omega_x} = %\mathsf{\widetilde c}_{x,y_L}^{(L)} 
   \mathsf{d}_{x}^{(L,\varepsilon)} \,\bra{ y_L^{(\varepsilon)} ,\eta },
   \quad & &\text{with}\quad \eta\, y_L^{(\varepsilon)}=-\tau_-+\varepsilon (\varphi_-+\psi_-)+\frac{1-\varepsilon }{2}i\pi ,  
\label{Ref-SoV-L}
\end{alignat}
for any value of $x$ such that the coefficients $\mathsf{d}_{x}^{(R,\varepsilon)}$ and $\mathsf{d}_{x}^{(L,\varepsilon)} $ are non-zero and finite.
These coefficients are respectively given by %{\bf (to check)}
\begin{align}
   &\mathsf{d}_{x}^{(R,\varepsilon)} %\mathsf{\widetilde c}_{x, y_R}^{(R)}
   =\prod_{n=1}^N\left[ e^{-\xi_n}\, \det S\big(-\xi_n \big|\tfrac{y_R^{(\varepsilon)}+x-N+2}2,\tfrac{y_R^{(\varepsilon)} -x-N}2+n-1\big)\right]^{-1} ,\\
   &\mathsf{d}_{x}^{(L,\varepsilon)}  %\mathsf{\widetilde c}_{x, y_L}^{(L)}
   =\prod_{n=1}^N \left[\det S\big(-\xi_n \big|\tfrac{y_L^{(\varepsilon)}+x+N-2}2,\tfrac{y_L^{(\varepsilon)}-x+N}2+n-1\big)\right]^{-1}.
\end{align}
\end{proposition}

\begin{proof}
In Proposition~3.4 of \cite{NicT22}, we have already proven the proportionality relations between the SoV pseudo reference state $\ket{\Omega_{\alpha-\beta+2M-1}}_\varepsilon$ and the state $\ket{\eta,\alpha+\beta+N-1-2M}$ \eqref{ref-x} under the constraint
\begin{equation}
   \eta(\alpha+\beta+N-1-2M)+\tau_-=-\varepsilon(\varphi_-+\psi_-)+\frac{1-\varepsilon}2 i\pi  \mod 2\pi i.
\end{equation}
Since the state $\ket{\Omega_{\alpha-\beta+2M-1}}_\varepsilon$ does not depend on $\alpha+\beta$, it actually means that the state $\ket{\Omega_x}_\varepsilon$, for any value of $x$, is proportional to the state $\ket{\eta,y_R^{(\varepsilon)}}$ with $y_R^{(\varepsilon)}$ given by \eqref{Ref-SoV-R}.
Let us now compute the proportionality factor by computing the following two scalar products:
\begin{equation}
   {}_\varepsilon\moy{x+2,\mathbf{h=1}\, |\, \Omega_x }_\varepsilon ,\quad \text{and} \quad
    {}_\varepsilon\moy{x+2,\mathbf{h=1}\, |\, \eta, y_R^{(\varepsilon)} }.
\end{equation}
On the one hand, from the orthogonality relation \eqref{Ortho-norm}, we get
\begin{equation}
   {}_\varepsilon\moy{x+2,\mathbf{h=1}\, |\, \Omega_x }_\varepsilon = e^{\sum_{j=1}^{N}\xi_{j}}.
\end{equation}
On the other hand, from \eqref{bra-Tinv} and \eqref{ref-state-actS}, we obtain
\begin{align}
    {}_\varepsilon\moy{x+2,\mathbf{h=1}\, |\, \eta, y_R^{(\varepsilon)} }
    &=  \bra{0}\, S^\mathrm{Adj}_{1\ldots N}(\boldsymbol{\xi} |\gamma,\delta)\, S_{1\ldots N}(\boldsymbol{\xi} |\gamma , \delta)\, \ket{ 0 } \nonumber\\
    &=\prod_{n=1}^N \det S(-\xi_n |\gamma,\delta+n-1)
\end{align}
for $(\gamma,\delta)$ such that
\begin{equation}
     \gamma+\delta =y_R^{(\varepsilon)} -N+1,\quad \text{and}\quad \gamma-\delta =x+1,
\end{equation}
i.e. $\gamma=(y_R^{(\varepsilon)} +x -N+2)/2$ and $\delta=(y_R^{(\varepsilon)} -x -N)/2$.

A similar reasoning can be made to obtain the relation between ${}_\varepsilon\bra{\Omega_x}$ and $\bra{y_L^{(\varepsilon)},\eta}$.
\end{proof}

%\begin{rem}
%The definition \eqref{Ref-SoV-R} and \eqref{Ref-SoV-L} of $y_R$ and $y_L$ implies the same value of $\varepsilon\in\{+,-\}$ as in the definition of the states \eqref{R-ref-state} and \eqref{L-ref-state}.
%\end{rem}

As a consequence, the separate states  \eqref{r-separate}  and \eqref{l-separate}  can be rewritten as
\begin{align}
& \ket{Q_{\boldsymbol \lambda},x}_\varepsilon =\mathsf{c}_{Q_{\boldsymbol \lambda},x}^{(R)} \, \mathsf{d}_{x+2M}^{(R,\varepsilon)}\  \underline{\widehat{\mathcal{B}}}_{-,M}(\{\lambda_{i}\}_{i=1}^{M}|x+2)\,\ket{\eta, y_R^{(\varepsilon)} } ,
\label{sep-ABA-R-bis} \\
&{}_\varepsilon \bra{x, Q_{\boldsymbol \lambda}}=\mathsf{c}_{Q_{\boldsymbol \lambda},x}^{(L)} \, \mathsf{d}_{x-2M}^{(L,\varepsilon)}\ \bra{ y_L^{(\varepsilon)},\eta }\,\underline{\widehat{\mathcal{B}}}_{-,M}(\{\lambda _{i}\}_{i=1}^{M}| x-2M),
\label{sep-ABA-L-bis}
\end{align}
for any value of $x=\alpha-\beta-1$ in \eqref{sep-ABA-R-bis} and $x=\alpha-\beta+1$ in \eqref{sep-ABA-L-bis} such that the states are well-defined and the corresponding normalisation coefficients are non-zero and finite.

Let us finally stress that usual ungauged Bethe states can simply be obtained by taking the limit $\eta(\alpha-\beta)\to +\infty$, see Remark~\ref{rem-limit}: %which follows directly from the expression \eqref{Bhat}: %{\bf (to check)}
\begin{equation}\label{lim-rstate-ABA}
    \prod_{i=1}^M\mathcal{B}_-(\lambda_i)\, \ket{\eta ,y_R^{(\varepsilon)}}
     =\lim_{\eta (\alpha -\beta )\rightarrow +\infty }
       \underline{\widehat{\mathcal{B}}}_{-,M}(\{\lambda_i \}_{i=1}^M |\alpha -\beta +1)\, \ket{\eta ,y_R^{(\varepsilon)}} ,
\end{equation}
and
\begin{equation}\label{lim-lstate-ABA}
     \bra{ y_L^{(\varepsilon)} ,\eta }\,\prod_{i=1}^M \mathcal{B}_-(\lambda_i)
     =\lim_{\eta (\alpha -\beta )\rightarrow +\infty }
     \bra{ y_L^{(\varepsilon)} ,\eta }\, \underline{\widehat{\mathcal{B}}}_{-,M}(\{\lambda_i\}_{i=1}^{M}|\alpha -\beta +1-2M).
\end{equation}
Hence, considering also the limit of the separate states \eqref{r-separate} and \eqref{l-separate},
\begin{align}
  \ket{Q} &= \lim_{\eta (\alpha -\beta )\rightarrow +\infty } \ket{Q,\alpha-\beta-1}
  \nonumber\\
  &=\sum_{\mathbf{h}\in\{0,1\}^{N}}\prod_{n=1}^{N} Q(\xi _{n}^{(h_{n})})\
e^{-\sum_{j}h_{j}\xi _{j}}\,\widehat{V}(\xi _{1}^{(h_{1})},\ldots ,\xi_{N}^{(h_{N})})\ \ket{\mathbf{h}} ,  
      \label{r-sep-ungauge} 
\end{align}
\begin{align}
  \bra{Q} &= \lim_{\eta (\alpha -\beta )\rightarrow +\infty } \bra{\alpha-\beta+1,  Q}
  \nonumber\\
  &=\sum_{\mathbf{h}\in\{0,1\}^{N}}\prod_{n=1}^{N}\left[ \left(\frac{\sinh (2\xi _{n}-\eta )}{\sinh(2\xi _{n}+\eta )}\,\frac{\mathsf{A}_-^{(\varepsilon )}(\xi _{n}+\frac{\eta }{2})}{\mathsf{A}_-^{(\varepsilon )}(-\xi _{n}+\frac{\eta }{2})} \right)^{\! h_{n}}\ Q(\xi _{n}^{(h_{n})})\right]  \notag \\
& \hspace{6cm} \times e^{-\sum_{j}h_{j}\xi _{j}}\,\widehat{V}(\xi _{1}^{(h_{1})},\ldots,\xi _{N}^{(h_{N})})\ \bra{ \mathbf{h}},
\label{lsep-ungauge}
\end{align}
we obtain the following relation between separate states and Bethe states in the ungauged limit:
\begin{align}
& \ket{Q_{\boldsymbol \lambda}}_\varepsilon =\mathsf{c}_{Q_{\boldsymbol \lambda}} \, \mathsf{d}^{(R,\varepsilon)}\   \prod_{i=1}^M\mathcal{B}_-(\lambda_i)\, \ket{\eta, y_R^{(\varepsilon)} } ,
\label{sep-ABA-R-ungauge} \\
& {}_\varepsilon\bra{Q_{\boldsymbol \lambda}}=\mathsf{c}_{Q_{\boldsymbol \lambda}} \, \mathsf{d}^{(L,\varepsilon)}\ \bra{ y_L^{(\varepsilon)},\eta }\, \prod_{i=1}^M\mathcal{B}_-(\lambda_i).
\label{sep-ABA-L-ungauge}
\end{align}
Here the proportionality coefficients are given by
\begin{align}
   \mathsf{c}_{Q_{\boldsymbol \lambda}} 
   &= \lim_{\eta x\rightarrow +\infty }  \mathsf{c}_{Q_{\boldsymbol \lambda},x}^{(R)} 
     = \lim_{\eta x\rightarrow +\infty }  \mathsf{c}_{Q_{\boldsymbol \lambda},x}^{(L)} 
     = \frac{(-1)^{NM}\,\mathsf{N}(\boldsymbol{\xi} )}{\prod_{j=1}^{M} b_-(\lambda_j)},
     \label{c-ABA-ungauge}\\
\end{align}
and
\begin{align}
   &\mathsf{d}^{(R,\varepsilon)} = \lim_{\eta x\rightarrow +\infty } \mathsf{d}_{x}^{(R,\varepsilon)}
   = \prod_{n=1}^N e^{2\xi_n+\eta(y_R^{(\varepsilon)}-N+n)}, \\
   &\mathsf{d}^{(L,\varepsilon)} = \lim_{\eta x\rightarrow +\infty } \mathsf{d}_{x}^{(L,\varepsilon)}
   = \prod_{n=1}^N e^{\xi_n+\eta(y_L^{(\varepsilon)}+N+n-2)}.
\end{align}

%%%%%%%%%%%%%%%%%%%%%%%%%%%%
%%%%%%%%%%%%%%%%%%%%%%%%%%%%
\section{The gauge invariant case: elementary blocks for correlation functions}
\label{sec-corr}

In \cite{NicT22,NicT23} (see also \cite{NicT24} for the XYZ generalisation), we computed elementary blocks for correlation functions for the XXZ spin chain with unparallel boundary fields, in the case in which the boundary parameters satisfy Nepomechie's constraint \cite{Nep04,NepR03} so that the ground state is described by usual Bethe equations around half-filling. More precisely, we have computed mean values, in the ground state, of some elements $\underline{\mathbf E}_m^{\boldsymbol{\epsilon'},\boldsymbol{\epsilon}}(\alpha,\beta)$, for $\boldsymbol{\epsilon}=(\epsilon_1,\ldots,\epsilon_m)$ and $\boldsymbol{\epsilon'}=(\epsilon'_1,\ldots,\epsilon'_m)\in\{1,2\}^m$,  of a particular basis of the set of local operators on the first $m$ sites of the chain. This basis of local operators was adequately defined in terms of the gauge parameters $\alpha$ and $\beta$ so that its action on the generalised Bethe states takes a simple form, see the next subsection.

Let us briefly recall  the main steps of the computation performed in \cite{NicT22,NicT23}. We use here the notations introduced in the previous section, which clarify the difficulties encountered there. The transfer matrix $\mathcal{T}(\lambda)$ of a spin chain with general boundary conditions can still be diagonalised in the SoV basis \eqref{ket-SoV}-\eqref{bra-SoV}, or equivalently \eqref{ket-Tinv}-\eqref{bra-Tinv}, with eigenstates having a similar form as \eqref{r-separate} and\footnote{There is actually a slight difference in the definition of the left separate state in the case with more general boundary conditions, since the latter involves instead of $\mathbf{A}^\mathrm{(Inv)}_\varepsilon$ a coefficient of the $TQ$-equation which depends on all boundary parameters, see \eqref{l-sep-gen}.}\eqref{l-separate}, but in that case the gauge parameters $\alpha$ and $\beta$ are directly related to the boundary parameters of the $K_+$ matrix,  see \eqref{cond-diff-gauge}-\eqref{cond-sum-gauge} and section~\ref{sec-spectrum-gen}.
%Hence the transfer matrix eigenstates can be expressed as separate states Under Nepomechie's constraint, some of the eigenstates of the transfer matrix $\mathcal{T}(\lambda)$ can be expressed as  
In this context, we considered  in \cite{NicT22,NicT23} the matrix elements of a quasi-local operator $\underline{\mathbf E}_m^{\boldsymbol{\epsilon'},\boldsymbol{\epsilon}}(\alpha,\beta)$ in the eigenstates $\ket{Q,\alpha-\beta-1}$ and $\bra{\alpha-\beta+1,Q}$ of the transfer matrix $\mathcal{T}(\lambda)$ for a given solution $Q\in\Sigma_Q^M$ of the (homogeneous) $TQ$-equation with roots $\{\lambda_1,\ldots,\lambda_M\}$\footnote{In the following, except if explicitly mentioned, we suppose that the sign $\varepsilon$ entering in the definition of the states via $\mathsf{g}_-\equiv \mathsf{g}_-^{(\varepsilon)}$ \eqref{geps} or via $y_R\equiv y_R^{(\varepsilon)}$ \eqref{Ref-SoV-R} is fixed. Therefore, except if needed, we shall no longer specify its value in the expression of the corresponding states, so as not to unnecessarily complicate the notations.}:
\begin{align}\label{matrix-el-gauge}
   \moy{\underline{\mathbf E}_m^{\boldsymbol{\epsilon'},\boldsymbol{\epsilon}}(\alpha,\beta)} 
   &\equiv \frac{\bra{\alpha-\beta+1,Q}\, \underline{\mathbf E}_m^{\boldsymbol{\epsilon'},\boldsymbol{\epsilon}}(\alpha,\beta) \, \ket{ Q,\alpha-\beta-1} }{\moy{\alpha-\beta+1,Q\, |\,Q,\alpha-\beta-1} },
   \nonumber\\
   &= \frac{\bra{\alpha-\beta+1,Q}\, 
        \underline{\mathbf E}_m^{\boldsymbol{\epsilon'},\boldsymbol{\epsilon}}(\alpha,\beta) \, 
        \underline{\widehat{\mathcal{B}}}_{-,M}(\{\lambda_{i}\}_{i=1}^{M}|\alpha -\beta +1)\,\ket{\eta, y_R } }
   {\bra{\alpha-\beta+1,Q}\, \underline{\widehat{\mathcal{B}}}_{-,M}(\{\lambda_{i}\}_{i=1}^{M}|\alpha -\beta +1)\,\ket{\eta, y_R } },
\end{align}
in which we used the proportionality relation between the separate state $\ket{Q,\alpha-\beta-1}$ and the generalised boundary Bethe state $\underline{\widehat{\mathcal{B}}}_{-,M}(\{\lambda_{i}\}_{i=1}^{M}|\alpha -\beta +1)\,\ket{\eta, y_R } $. 
%We recall that, in \cite{NicT22,NicT23}, the gauge parameters $\alpha$ and $\beta$ were fixed in terms of the boundary parameters of the $K_+$ matrix,  see \eqref{cond-diff-gauge}-\eqref{cond-sum-gauge}. 
The strategy used in \cite{NicT22,NicT23} was to decompose the generalised boundary Bethe state $\underline{\widehat{\mathcal{B}}}_{-,M}(\{\lambda_{i}\}_{i=1}^{M}|\alpha -\beta +1)\,\ket{\eta, y_R }$ 
into a sum of generalised bulk Bethe states, to compute the action of the local operator $\underline{\mathbf E}_m^{\boldsymbol{\epsilon'},\boldsymbol{\epsilon}}(\alpha,\beta) $ on these generalised bulk Bethe states using the bulk inverse problem, and to reconstruct the result in terms of boundary Bethe states. The result of this action (see Theorem~\ref{th-2208}) was then given in terms of generalised boundary Bethe states of the form
\begin{equation}\label{Bethe-state-shifted}
   \underline{\widehat{\mathcal{B}}}_{-,M+\tilde{m}_{\boldsymbol{\epsilon,\epsilon'}}}(\{\mu_j\}_{j=1}^{M+\tilde{m}_{\boldsymbol{\epsilon,\epsilon'}}}|\alpha -\beta -2\tilde{m}_{\boldsymbol{\epsilon,\epsilon'}} +1)\,\ket{\eta, y_R } ,
\end{equation}
involving possibly a different number $M+\tilde{m}_{\boldsymbol{\epsilon,\epsilon'}}$ of $\mathcal{B}$-operators and a shift of the dynamical parameter $\beta$ by an integer number $2\tilde{m}_{\boldsymbol{\epsilon,\epsilon'}}$ which depends on the local operator $\underline{\mathbf E}_m^{\boldsymbol{\epsilon'},\boldsymbol{\epsilon}}(\alpha,\beta)$ under consideration:
\begin{equation}\label{def-tildem}
\tilde{m}_{\boldsymbol{\epsilon,\epsilon'}}=\sum_{r=1}^{m}(\epsilon'_{r}-\epsilon_{r})
=\sum_{r=n}^{m}\frac{(-1)^{\epsilon'_r}-(-1)^{\epsilon_r}}{2}.
\end{equation}
A boundary Bethe state of the form \eqref{Bethe-state-shifted} can still be  identified with a separate state, but {\em a priori} on a SoV basis which differs from the initial one: if $P\in\Sigma_Q^{M+\tilde{m}_{\boldsymbol{\epsilon,\epsilon'}}}$ has roots $\{\mu_1,\ldots\mu_{M+\tilde{m}_{\boldsymbol{\epsilon,\epsilon'}}}\}$, the resulting state \eqref{Bethe-state-shifted} is proportional to the separate state   $\ket{P,\alpha -\beta -2{\tilde{m}_{\boldsymbol{\epsilon,\epsilon'}}}-1}$. Hence, the matrix elements \eqref{matrix-el-gauge} is obtained as a sum over scalar products of the form
\begin{equation}\label{resulting-SP}
   \frac{\moy{\alpha-\beta+1,Q\, | \, P,\alpha-\beta-2\tilde{m}_{\boldsymbol{\epsilon,\epsilon'}}-1} }{\moy{\alpha-\beta+1,Q\, |\,Q,\alpha-\beta-1} }.
\end{equation}
The computation of such scalar products for $\tilde{m}_{\boldsymbol{\epsilon,\epsilon'}}\not=0$, which involves the expression of  left and right separate states in different SoV basis, still remains an open problem. Hence, in \cite{NicT22,NicT23}, we had to restrict our study to matrix elements of the subset of operators $\underline{\mathbf E}_m^{\boldsymbol{\epsilon'},\boldsymbol{\epsilon}}(\alpha,\beta) $ with $\tilde{m}_{\boldsymbol{\epsilon,\epsilon'}}=0$,  for which the resulting scalar products can be computed by using the results of \cite{KitMNT18}.

The chain with the special boundary condition \eqref{gauge-inv-K}-\eqref{h+inv} on site 1 and a generic non-longitudinal boundary field on site $N$ provides a particularly simple example of a case with unparallel boundary fields in which all elementary blocks for correlation functions can be explicitly computed. In that case, the gauge parameters $\alpha$ and $\beta$ used to construct the SoV bases are arbitrary, and we can in particular use, for the computation of correlation functions, the ungauged basis obtained in the limit $\eta(\alpha-\beta)\to +\infty$: in this limit, the discrepancy that we observe in \eqref{resulting-SP} between the SoV bases used to express the left and the right separate state disappears, so that the resulting scalar product \eqref{resulting-SP}  can be computed using the results of \cite{KitMNT18}. This is of course not the only case for which all elementary blocks can be computed in an ungauged basis: it is enough for this that the K-matrix at site 1 is triangular, see \cite{Nic12}. The special case considered here  is however particularly simple in that, even if one of the boundary is non-diagonal and completely generic, it is isospectral to the case in which both boundaries are diagonal, see \eqref{iso-diag}, and is not subject to Nepomechie's constraint\footnote{\label{foot-config-bis}Note that there exists another configuration of the boundaries with the same simple properties, with a diagonal K-matrix on site 1 and a triangular one on site $N$: it is obtained by setting $\kappa_+=0$ while taking the following limit 
\begin{equation}
e^{r t_-}=\lim_{\substack{ \tau_-\longrightarrow r\infty  \\ \kappa_-\longrightarrow 0}} \kappa_-\, e^{r\tau_-},\quad r\in \{-1,1\},
\end{equation}
with $t_-$ finite on the last site. This corresponds to the following boundary matrices:
\begin{equation}
K_{+}(\lambda )=\frac{1}{\sinh \varsigma_+}
\begin{pmatrix}
\sinh (\varsigma_+ +\lambda +\eta /2) & 0 \\ 
0 & \sinh (\varsigma_+ -\lambda -\eta /2)
\end{pmatrix},
\end{equation}
and
\begin{equation}
K_{-}(\lambda )=\frac{1}{\sinh \varsigma_-}%
\begin{pmatrix}
\sinh (\varsigma_-+\lambda -\eta /2) & \delta _{r,1}\, e^{t_-} \sinh(2\lambda -\eta ) \\ 
\delta _{r,-1}\, e^{-t_-} \sinh (2\lambda -\eta ) & \sinh (\varsigma_--\lambda +\eta /2)%
\end{pmatrix},
\end{equation}
and to the following boundary fields in the Hamiltonian:
\begin{equation}
\frac{\sinh \eta }{\sinh \varsigma_+}\cosh \varsigma_+\, \sigma_{1}^{z}
+\frac{\sinh \eta }{\sinh \varsigma_-}\big[\cosh \varsigma_-\, \sigma_{N}^{z}+e^{rt_-}\,(\sigma _{N}^{x}+ir\sigma _{N}^{y})\big].
\end{equation}
The elementary blocks for this configuration can then be computed similarly as for the case we consider here. %The result is identical, with the replacement $(\varphi_-,\psi_-)\to (\varsigma_+,\varsigma_-+i\tfrac\pi 2)$.
}: the complete spectrum and eigenstates are given by the solutions of the homogenous $TQ$-equation \eqref{TQ-hom}, hence by usual Bethe equations, see Theorem~\ref{Theo-Sp-Inv}.
%: in the more general triangular case, the second boundary K-matrix at site $N$ cannot be completely general, it should satisfy Nepomechie's constraint for the ground state to satisfy usual Bethe equations; moreover, the boundary fields being in the latter case described by more than two boundary parameters, the representation for the scalar products becomes more complicated, see \cite{KitMNT18}. 

In this section, we explain how to compute all elementary blocks for correlation functions, and especially the one that we did not compute in~\cite{NicT22,NicT23}, in this simple case.

%%%%%%%
\subsection{\label{Gauged Action}Action of local operators on boundary separate states}

We computed in \cite{NicT22} the action of a basis of local operators on a generic separate state \eqref{r-separate} expressed in generic gauged SoV bases \eqref{ket-SoV}.  We recall here these results and use them to derive the limiting case of the action on the ungauged SoV basis that we shall use in this section.

Let $E^{i,j}$, $i,j\in \{1,2\}$ be the elementary matrices on $\mathbb{C}^{2} $, with elements $(E^{i,j})_{k,\ell }=\delta _{i,k}\,\delta _{j,\ell }$, and let us define, as in \cite{NicT22,NicT23}, the following local operators at site $n$: 
\begin{equation}
E_{n}^{\epsilon _{n}^{\prime },\epsilon _{n}}(u|(a,b),(\bar{a},\bar{b}))
=S_{n}(-u|\bar{a},\bar{b})\,E_{n}^{\epsilon _{n}^{\prime },\epsilon _{n}}\, \left[ S_{n}(-u|a,b)\right] ^{-1}
\in \text{End}\mathcal{H}_{n}\ ,
\quad\epsilon _{n}^{\prime },\epsilon _{n}\in \{1,2\}.  \label{Local-op}
\end{equation}
Then for given arbitrary values of the parameters $(u,a,b,\bar{a},\bar{b})$, the operators \eqref{Local-op} correspond to four different linear combinations of the local elementary operators $E_{n}^{i,j}\in \End\mathcal{H}_{n}$, $1\leq i,j\leq 2$, which read 
\begin{align}
& E_{n}^{1,1}(u|(a,b),(\bar{a},\bar{b}))=\frac{-e^{-\eta (\bar{a}+\bar{b})}\,E_{n}^{1,1}+e^{-u-\eta (a+\bar{a}-b+\bar{b})}\,E_{n}^{1,2}-e^{u}\,E_{n}^{2,1}+e^{-\eta (a-b)}\,E_{n}^{2,2}}{2e^{-\eta a}\sinh (\eta b)}, \\
& E_{n}^{1,2}(u|(a,b),(\bar{a},\bar{b}))=\frac{e^{-\eta (\bar{a}+\bar{b})}\,E_{n}^{1,1}-e^{-u-\eta (a+\bar{a}+b+\bar{b})}\,E_{n}^{1,2}+e^{u}\,E_{n}^{2,1}-e^{-\eta (a+b)}\,E_{n}^{2,2}}{2e^{-\eta a}\sinh (\eta b)}, \\
& E_{n}^{2,1}(u|(a,b),(\bar{a},\bar{b}))=\frac{-e^{-\eta (\bar{a}-\bar{b})}\,E_{n}^{1,1}+e^{-u-\eta (a+\bar{a}-b-\bar{b})}\,E_{n}^{1,2}-e^{u}\,E_{n}^{2,1}+e^{-\eta (a-b)}\,E_{n}^{2,2}}{2e^{-\eta a}\sinh (\eta b)}, \\
& E_{n}^{2,2}(u|(a,b),(\bar{a},\bar{b}))=\frac{e^{-\eta (\bar{a}-\bar{b})}\,E_{n}^{1,1}-e^{-u-\eta (a+\bar{a}+b-\bar{b})}\,E_{n}^{1,2}+e^{u}\,E_{n}^{2,1}-e^{-\eta (a+b)}\,E_{n}^{2,2}}{2e^{-\eta a}\sinh (\eta b)}.
\end{align}
Moreover, still as in \cite{NicT22,NicT23}, we define, for each $m$-tuples $\boldsymbol{\epsilon}=(\epsilon_1,\ldots,\epsilon_m)$ and $\boldsymbol{\epsilon'}=(\epsilon'_1,\ldots,\epsilon'_m)\in\{1,2\}^m$, the following tensor product of local operators on the first $m$ sites of the chain: 
\begin{equation}\label{operator-m}
   \underline{\mathbf E}_m^{\boldsymbol{\epsilon'},\boldsymbol{\epsilon}}(a,b)
   =\prod_{n=1}^m E_n^{\epsilon'_n,\epsilon_n}(\xi_n |(a_n,b_n),(\bar{a}_n,\bar{b}_n))
   \in \End(\otimes_{n=1}^m\mathcal{H}_n),
\end{equation}
in terms of the parameters $a,b$ by setting 
\begin{alignat}{2}
& a_{n}=a+1,\qquad & & b_{n}=b-\sum_{r=1}^{n}(-1)^{\epsilon _{r}}, 
\displaybreak[0]  \label{Gauge.Basis-1} \\
& \bar{a}_{n}=a-1,\qquad & & \bar{b}_{n}=b+\sum_{r=n+1}^{m}(-1)^{\epsilon_{r}^{\prime }}-\sum_{r=1}^{m}(-1)^{\epsilon _{r}}=b_{n}+2\tilde{m}_{n+1},
\label{Gauge.Basis-2}
\end{alignat}
with 
\begin{equation}
\tilde{m}_{n}=\sum_{r=n}^{m}(\epsilon _{r}^{\prime }-\epsilon_{r})=\sum_{r=n}^{m}\frac{(-1)^{\epsilon _{r}^{\prime }}-(-1)^{\epsilon _{r}}}{2}.
\end{equation}
Then, as shown in \cite{NicT22}, except for a finite numbers of values of $b\mod 2\pi /\eta $, the set
\begin{equation}
\mathbb{E}_{m}(a,b)=\left\{  \underline{\mathbf E}_m^{\boldsymbol{\epsilon'},\boldsymbol{\epsilon}}(a,b)\ \mid \ \boldsymbol{\epsilon},\boldsymbol{\epsilon'}\in \{1,2\}^{m}\right\} ,  \label{Local-Basis}
\end{equation}
defines a basis of $\text{End}(\otimes _{n=1}^{m}\mathcal{H}_{n})$.
In \cite{NicT22,NicT23}, we computed the action of the elements of this basis on the separate states (or equivalently the gauge boundary Bethe states) \eqref{separate-ABA-R}. Note that the choice of the basis \eqref{Local-Basis} was motivated there by the fact that the action of its elements on the gauge boundary Bethe states has a relatively simple form.

Let us recall this result here:

\begin{theorem}[\cite{NicT22,NicT23}]
\label{th-2208}
For arbitrary parameters $\{\lambda_1,\ldots,\lambda_M\}$, and under any choice of $(a,b)$ satisfying the condition
\begin{equation}\label{cond-ab-y}
a+b=y_R +2M+1-N,
\end{equation}
with $y_R $ defined as in \eqref{Ref-SoV-R} in terms of the boundary parameters of $K_-$ and of a given $\varepsilon\in\{+,-\}$, the action on the generalised boundary Bethe state
\begin{equation}
   \underline{\widehat{\mathcal{B}}}_{-,M}(\{\lambda_i\}_{i=1}^{M}|a-b+1)\, \ket{\eta ,y_R }
\end{equation}
of a generic element $\underline{\mathbf E}_m^{\boldsymbol{\epsilon'},\boldsymbol{\epsilon}}(a,b)$ \eqref{operator-m} of the basis \eqref{Local-Basis} of local operators on the first $m$ sites of the chain is
\begin{multline}
   \underline{\mathbf E}_m^{\boldsymbol{\epsilon'},\boldsymbol{\epsilon}}(a,b) \
   \underline{\widehat{\mathcal{B}}}_{-,M}(\{\lambda_i\}_{i=1}^{M}|a-b+1)\, \ket{\eta ,y_R }
   \\
 =\sum_{\mathsf{B}_{\boldsymbol{\epsilon,\epsilon'}}}
 \mathcal{ F}_{\mathsf{B}_{\boldsymbol{\epsilon,\epsilon'}}}(\{\lambda_j\}_{j=1}^M,\{\xi_j^{(1)}\}_{j=1}^m|a,b)\   
 \underline{\widehat{\mathcal{B}}}_{-,M+\tilde{m}_{\boldsymbol{\epsilon,\epsilon'}}}(\{\lambda _{i}\}_{\substack{ i=1  \\ i\notin \mathsf{B}_{\boldsymbol{\epsilon,\epsilon'}}}}^{M+m}|a-b+1-2\tilde{m}_{\boldsymbol{\epsilon,\epsilon'}})\,
 \ket{\eta ,y_R} .  \label{act-boundary}
\end{multline}
In this expression, we have used the notation $\lambda _{M+j}:=\xi _{m+1-j}^{(1)}$ for $j\in \{1,\ldots ,m\}$ and $\tilde{m}_{\boldsymbol{\epsilon,\epsilon'}}$ is given in terms of $\boldsymbol{\epsilon}$ and $\boldsymbol{\epsilon'}$ as in \eqref{def-tildem}.
%
%\begin{equation}\label{def-tildem}
%\tilde{m}_{\boldsymbol{\epsilon,\epsilon'}}=\sum_{r=1}^{m}(\epsilon'_{r}-\epsilon_{r})
%=\sum_{r=n}^{m}\frac{(-1)^{\epsilon'_r}-(-1)^{\epsilon_r}}{2}.
%\end{equation}
%
The sum  in \eqref{act-boundary} runs over all the possible sets of integers $\mathsf{B}_{\boldsymbol{\epsilon},\boldsymbol{\epsilon'}}=\{\mathsc{b}_{1},\ldots,\mathsc{b}_{s+s^{\prime }}\}$ such that 
\begin{equation}\label{Def-Bss}
\begin{cases}
  \mathsc{b}_p \in \{1,\ldots ,M\}\setminus \{\mathsc{b}_{1},\ldots ,\mathsc{b}_{p-1}\}\qquad & \text{for}\quad 0<p\leq s, 
  \\
  \mathsc{b}_p \in \{1,\ldots ,M+m+1-i_p\}\setminus \{\mathsc{b}_{1},\ldots ,\mathsc{b}_{p-1}\}\quad & \text{for}\quad s<p\leq s+s^{\prime },
\end{cases}
\end{equation}
in which we have defined $s$ and $s'$ as
\begin{equation}\label{def-s-s'}
  s=\sum_{j=1}^{m}(\epsilon_j-1),\qquad s'=\sum_{j=1}^{m}(2-\epsilon'_j),\qquad \text{so that}\quad
  m=s+s'+\tilde{m}_{\boldsymbol{\epsilon,\epsilon'}},
\end{equation}
and integers $i_p$ such that
\begin{alignat}{2}
   &\{ i_p\}_{p\in\{1,\ldots,s\}}=\{1,\ldots,m\}\cap \{j\mid\epsilon_j=2\}\qquad& &\text{with}\quad i_p<i_q\ \ \text{if}\ \ p<q,\label{i_p-s}\\
   &\{ i_p\}_{p\in\{s+1,\ldots,s+s'\}}=\{1,\ldots,m\}\cap \{j\mid\epsilon'_j=1\}\qquad& &\text{with}\quad i_p>i_q\ \ \text{if}\ \ p<q.\label{i_p-s'}
\end{alignat}
The coefficient in the sum \eqref{act-boundary} is given by the following expression:
\begin{align}\label{act-bound-coeff} 
& \mathcal{ F}_{\mathsf{B}_{\boldsymbol{\epsilon,\epsilon'}}}(\{\lambda_j\}_{j=1}^M,\{\xi_j^{(1)}\}_{j=1}^m|a,b)
  = \left[(-1)^N \frac{e^{\eta a}}2\right]^{\tilde{m}_{\boldsymbol{\epsilon,\epsilon'}}  }
    \prod_{n=1}^m\frac{e^{\eta }}{\sinh (\eta b_n)}\,
    \sum_{\sigma_{\alpha_+}=\pm }\frac{\prod_{j=1}^{s+s'}d(\lambda_{{\mathsc b}_j}^\sigma)}{\prod_{j=1}^m d(\xi _j^{(1)}) }
    \nonumber\\
& \qquad \times 
   \frac{H_{\boldsymbol{\sigma}_{\alpha_+}}(\{\lambda _{\alpha_+}\})}{H_{\boldsymbol 1}(\{\xi _{\gamma_+}^{(1)}\})}\,
   \prod_{i\in \alpha_-}\prod_{\eps =\pm }\left\{ 
   \prod_{j\in \alpha_+}\frac{\sinh (\lambda_j^\sigma +\epsilon \lambda_i+\eta )}{\sinh (\lambda_j^\sigma +\epsilon\lambda_i)}
   \prod_{j\in \gamma_+}\frac{\sinh (\xi_j^{(1)}+\epsilon \lambda_i)}{\sinh (\xi_j^{(0)}+\epsilon \lambda_i)}\right\}  
   \nonumber \\
& \qquad \times \prod_{i\in \alpha _{+}}\left\{ \prod_{j\in \gamma _{+}}%
\frac{\sinh (\xi _{j}^{(1)}-\lambda _{i}^{\sigma })}{\sinh (\xi _{j}^{(0)}-\lambda
_{i}^{\sigma })}\ \frac{\prod_{j\in \alpha _{+}}\sinh (\lambda _{j}^{\sigma
}-\lambda _{i}^{\sigma }-\eta )}{\prod_{j\in \alpha _{+}\setminus \{i\}}\sinh
(\lambda _{j}^{\sigma }-\lambda _{i}^{\sigma })}\right\} \prod_{1\leq i<j\leq
s+s^{\prime }}\frac{\sinh (\lambda _{\text{\textsc{b}}_{i}}^{\sigma }-\lambda _{%
\text{\textsc{b}}_{j}}^{\sigma })}{\sinh (\lambda _{\text{\textsc{b}}%
_{i}}^{\sigma }-\lambda _{\text{\textsc{b}}_{j}}^{\sigma }-\eta )}  \notag \\
& \qquad \times \prod_{p=1}^{s}\left[ \sinh (\xi _{i_{p}}^{(1)}-\lambda _{\text{%
\textsc{b}}_{p}}^{\sigma }+\eta (1+b_{i_{p}}))\,\frac{\prod_{k=i_{p}+1}^{m}%
\sinh (\lambda _{\text{\textsc{b}}_{p}}^{\sigma }-\xi _{k}^{(1)}-\eta )}{%
\prod_{k=i_{p}}^{m}\sinh (\lambda _{\text{\textsc{b}}_{p}}^{\sigma }-\xi
_{k}^{(1)})}\right]  \notag \\
& \qquad \times \prod_{p=s+1}^{s+s^{\prime }}\left[ \sinh (\xi
_{i_{p}}^{(1)}-\lambda _{\text{\textsc{b}}_{p}}^{\sigma }-\eta (1-\bar{b}%
_{i_{p}}))\,\frac{\prod_{k=i_{p}+1}^{m}\sinh (\xi _{k}^{(1)}-\lambda _{\text{%
\textsc{b}}_{p}}^{\sigma }-\eta )}{\prod_{\substack{ k=i_{p}  \\ k\not=M+m+1-%
{\text{\textsc{b}}_{p}}}}^{m}\sinh (\xi _{k}^{(1)}-\lambda _{\text{\textsc{b}}%
_{p}}^{\sigma })}\right] .
\end{align}
The sum is here performed over all $\sigma _{j}\in \{+,-\}$ for $j\in \alpha _{+}$, we have defined $\lambda _{i}^{\sigma }=\sigma _{i}\lambda _{i}$ for $i\in \mathsf{B}_{\boldsymbol{\epsilon,\epsilon'}}$, with $\sigma _{i}=1$ if $i>M$, and 
\begin{alignat}{2}
& \alpha _{+}=\mathsf{B}_{\boldsymbol{\epsilon,\epsilon'}}\cap \{1,\ldots ,M\}, & & \alpha _{-}=\{1,\ldots ,M\}\setminus \alpha _{+}, \\
& \gamma _{-}=\{M+m+1-j\}_{j\in \mathsf{B}_{\boldsymbol{\epsilon,\epsilon'}}\cap \{N+1,\ldots ,N+m\}},\quad & & \gamma _{+}=\{1,\ldots ,m\}\setminus \gamma _{-}.
\end{alignat}
Finally, for a given set of spectral parameters $\{\mu\}=\{\mu_i\}_{i=1}^n$ and a given $n$-tuple of signs $\boldsymbol{\sigma}\equiv(\sigma_1,\ldots,\sigma_n)$, the function $H_{\boldsymbol \sigma }(\{\mu \})\equiv  H_\sigma(\{\mu \}|\varepsilon\varphi_-,\varepsilon\psi_-)$ is given by
\begin{multline}\label{def-H}
   H_{\boldsymbol{\sigma}}(\{\mu\})
   =\prod_{j=1}^n\left[ \sigma_j a(-\mu_j^\sigma)\,\frac{\sinh (2\mu_j-\eta )}{\sinh(2\mu_j)}
      \frac{\sinh (-\mu_j^\sigma-\frac\eta 2 +\varepsilon\varphi_-)\, \cosh (-\mu_j^\sigma-\frac\eta 2 +\varepsilon \psi_-)}{\sinh (\varepsilon \varphi_-)\,\cosh (\psi_-)}\right]   \\
 \times \prod_{1\leq i<j\leq n}\frac{\sinh (\mu_i^\sigma+\mu_j^\sigma+\eta )}{\sinh (\mu_i^\sigma+\mu_j^\sigma)}.
\end{multline}
\end{theorem}

\begin{rem}
  The function \eqref{def-H} appears as a coefficient in the boundary-bulk decomposition of boundary Bethe states into bulk Bethe states, see \cite{NicT23}.
\end{rem}

\begin{rem}
   We have here a  different normalisation factor $\prod_{j=1}^{\tilde{m}_{\boldsymbol{\epsilon,\epsilon'}}}\left[-e^{\eta(a-b+1-2j)}/2\right]$, in the expression \eqref{act-bound-coeff}, with respect to the result stated in \cite{NicT22,NicT23}. This is due to the fact that we have chosen here a slightly different normalisation for the operator $\widehat{\mathcal B}_-$, see Remark~\ref{rem-norm-hatB}.
\end{rem}

Let us again underline the fact that, although we use in \eqref{act-boundary} the Bethe like expression for the gauge boundary states, these states can be reformulated as separate states. However, as already mentioned in the introduction of this section, this reformulation may possibly involve {\em different SoV bases}. More precisely, it follows from \eqref{sep-ABA-R-bis} that the initial state in \eqref{act-boundary}  \eqref{Ref-SoV-R} is given as a separate state on the SoV basis \eqref{ket-SoV}  with gauge parameter $a-b-1$,
\begin{equation}\label{separate-in}
   \underline{\widehat{\mathcal{B}}}_{-,M}(\{\lambda_i\}_{i=1}^{M}|a-b+1)\, \ket{\eta ,y_R }
   =\frac{1}{\mathsf{c}_{Q,a-b-1}^{(R)} \, \mathsf{d}_{a-b-1+2M}^{(R,\varepsilon)}}\
   \ket{Q,a-b-1},
\end{equation}
in which $Q$ is the element of $\Sigma_Q^M$ with roots given by $\{\lambda_1,\ldots,\lambda_M\}$.
Similarly, the resulting state  in \eqref{act-boundary} is given as a separate state on the SoV basis \eqref{ket-SoV}  with gauge parameter $a-b-1-2\tilde{m}_{\boldsymbol{\epsilon,\epsilon'}} $,
\begin{multline}\label{separate-out}
   \underline{\widehat{\mathcal{B}}}_{-,M+\tilde{m}_{\boldsymbol{\epsilon,\epsilon'}}}(\{\mu_i\}_{i=1}^{M+\tilde{m}_{\boldsymbol{\epsilon,\epsilon'}}}|a-b+1-2\tilde{m}_{\boldsymbol{\epsilon,\epsilon'}})\,
 \ket{\eta ,y_R }
 \\
   =\frac{1}{\mathsf{c}_{P,a-b-1-2\tilde{m}_{\boldsymbol{\epsilon,\epsilon'}} }^{(R)} \, \mathsf{d}_{a-b-1-2\tilde{m}_{\boldsymbol{\epsilon,\epsilon'}}+2M}^{(R,\varepsilon)}}\
   \ket{P,a-b-1-2\tilde{m}_{\boldsymbol{\epsilon,\epsilon'}} },
\end{multline}
in which $P$ is the element of $\Sigma_Q^M$ with roots  $\{\mu_1,\ldots,\mu_{M+\tilde{m}_{\boldsymbol{\epsilon,\epsilon'}}}\}=\{\lambda _{i}\, |\, 1\le i\le M+m, i\notin \mathsf{B}_{\boldsymbol{\epsilon,\epsilon'}}\}$.
Proceeding further with the computation of correlation functions as in \cite{NicT22,NicT23} would  involve the computation of scalar products between the initial state and the resulting state, i.e. scalar products of the form 
\begin{equation}
     \moy{a-b+1, Q\, |\, P,a-b-1-2\tilde{m}_{\boldsymbol{\epsilon,\epsilon'}} }, 
\end{equation}
which remains so far an open question when $\tilde{m}_{\boldsymbol{\epsilon,\epsilon'}} \not= 0$. This is why in \cite{NicT22,NicT23} we restricted our study to matrix elements of local operators for which $\tilde{m}_{\boldsymbol{\epsilon,\epsilon'}} = 0$, i.e. for which the results of \cite{KitMNT18} for scalar products can directly be applied.

We no longer have this problem when working at the ungauged level, i.e. in the limit $\eta(a-b)\to +\infty$ in which the gauged states become ungauged ones, see \eqref{lim-rstate}-\eqref{lim-lstate} and \eqref{lim-rstate-ABA}-\eqref{lim-lstate-ABA}. Let us therefore investigate this limiting case more thoroughly.

Let us consider the limit $\eta(a-b)\to +\infty$ of the operators \eqref{operator-m}:
\begin{equation}\label{operator-m-lim}
   \underline{\mathbf E}_m^{\boldsymbol{\epsilon'},\boldsymbol{\epsilon}}(x)
   = \lim_{\substack{ \eta (a-b)\rightarrow +\infty  \\ a+b=x }}\underline{\mathbf E}_m^{\boldsymbol{\epsilon'},\boldsymbol{\epsilon}}(a,b).
\end{equation}
It is given as the following tensor product of local operators:
\begin{equation}
   \underline{\mathbf E}_m^{\boldsymbol{\epsilon'},\boldsymbol{\epsilon}}(x)
   =\prod_{n=1}^m E_n^{\epsilon'_n,\epsilon_n}(\xi_n |x_n,\bar x_n)
   \in \End(\otimes_{n=1}^m\mathcal{H}_n)
\end{equation}
in which $x_{n}$ and $\bar{x}_{n}$ are given in terms of the $m$-tuples $\boldsymbol{\epsilon}\equiv (\epsilon _{1},\ldots ,\epsilon _{m})$ and $\boldsymbol{\epsilon'}\equiv (\epsilon _{1}^{\prime },\ldots ,\epsilon
_{m}^{\prime })$ as
\begin{equation}\label{def-x_n}
x_{n}=x+1-\sum_{r=1}^{n}(-1)^{\epsilon _{r}},\qquad \bar{x}_{n}=x_{n}+2(\tilde{m}_{n+1}-1),
\end{equation}
and in which the operators $E_{n}^{i,j}(u|x,\bar{x})$ are given as a limit of the operators \eqref{Local-op}:
\begin{equation}\label{lim-ungop}
E_{n}^{i,j}(u|x,\bar{x})=\lim_{\substack{ \eta (a-b)\rightarrow +\infty  \\ \eta (\bar{a}-\bar{b})\rightarrow +\infty }}\left. E_{n}^{i,j}(u|(a,b),(\bar{a},\bar{b}))\right\vert _{\substack{ a+b=x \\ \bar{a}+\bar{b}=\bar{x}}}.
\end{equation}
Explicitly, \eqref{lim-ungop} correspond to the following local operators at site $n$:
\begin{align}
  &E_{n}^{1,1}(u|x,\bar{x}) = e^{\eta (x-\bar{x})}\,E_{n}^{1,1}+e^{u+\eta x}\,E_{n}^{2,1}, \\
  &E_{n}^{1,2}(u|x,\bar{x}) = -e^{\eta (x-\bar{x})}\,E_{n}^{1,1}+e^{-u-\eta \bar{x}}\,E_{n}^{1,2}-e^{u+\eta x}\,E_{n}^{2,1}+\,E_{n}^{2,2}, \\
  &E_{n}^{2,1}(u|x,\bar{x}) = e^{u+\eta x}\,E_{n}^{2,1}, \\
  &E_{n}^{2,2}(u|x,\bar{x}) =E_{n}^{2,2}-e^{u+\eta x}\,E_{n}^{2,1}.
\end{align}
Except for a finite numbers of values of $x\!\!\mod2\pi /\eta $, the set of operators
\begin{equation}
\mathbb{E}_{m}(x)=\left\{  \underline{\mathbf E}_m^{\boldsymbol{\epsilon'},\boldsymbol{\epsilon}}(x)\ \mid \ \boldsymbol{\epsilon},\boldsymbol{\epsilon'}\in \{1,2\}^{m}\right\}   \label{Local-Basis-lim}
\end{equation}
%
%\mathbb{E}_{m}(x)$ 
defines a basis of $\text{End}(\otimes _{n=1}^{m}\mathcal{H}_{n})$. The action of this basis of local operators on the ungauged separate/Bethe states is then a direct consequence of Theorem~\ref{th-2208}:

\begin{cor}
\label{cor-act-ungauged}
For arbitrary parameters $\{\lambda_1,\ldots,\lambda_M\}$, and $x$ defined by
\begin{equation}\label{cond-x-y}
  x=y_R +2M+1-N,
\end{equation}
in terms of $y_R$ \eqref{Ref-SoV-R}, the action on the boundary Bethe state
\begin{equation}
   \prod\limits_{i=1}^M \mathcal{B}_-(\lambda_i)\,\ket{\eta ,y_R } 
\end{equation}
of a generic element $ \underline{\mathbf E}_m^{\boldsymbol{\epsilon'},\boldsymbol{\epsilon}}(x)$ of the basis \eqref{Local-Basis-lim} of local operators on the first $m$ sites of the chain is given by
\begin{align}\label{act-boundary-lim}
   \underline{\mathbf E}_m^{\boldsymbol{\epsilon'},\boldsymbol{\epsilon}}(x)\  \prod\limits_{i=1}^M \mathcal{B}_-(\lambda_i)\,\ket{\eta ,y_R } 
    =\sum_{\mathsf{B}_{\boldsymbol{\epsilon,\epsilon'}}}\mathcal{\bar{F}}_{\mathsf{B}_{\boldsymbol{\epsilon,\epsilon'}}}(\{\lambda_{j}\}_{j=1}^{M},\{\xi _{j}^{(1)}\}_{j=1}^{m}|x)\ 
    \prod_{\substack{ i=1  \\ i\notin \mathsf{B}_{\boldsymbol{\epsilon,\epsilon'}}}}^{M+m} \mathcal{B}_- (\lambda_i)\, \ket{\eta ,y_R } ,
\end{align}
in which we have used the same notations as in Theorem~\ref{th-2208} and defined
\begin{align}
& \mathcal{\bar{F}}_{\mathsf{B}_{\boldsymbol{\epsilon,\epsilon'}}}(\{\lambda_j\}_{j=1}^{M},\{\xi_j^{(1)}\}_{j=1}^{m} |x )
   =(-1)^{(N+1)\tilde{m}_{\boldsymbol{\epsilon,\epsilon'}}} \prod_{n=1}^m e^{\eta x_n}
   \sum_{\sigma_{\alpha_{+}}=\pm }
   \frac{\prod_{j=1}^{s+s^{\prime }}d(\lambda _{\text{\textsc{b}}_{j}}^{\sigma })}{\prod_{j=1}^{m}d(\xi _{j}^{(1)})}\  
   \nonumber \\
& \qquad \times 
   \frac{H_{\sigma _{\alpha _{+}}}(\{\lambda _{\alpha _{+}}\})}{H_{1}(\{\xi _{\gamma _{+}}^{(1)}\})}
   \prod_{i\in \alpha_{-}}\prod_{\epsilon =\pm }\left\{ 
   \prod_{j\in \alpha _{+}}\frac{\sinh (\lambda_{j}^{\sigma }+\epsilon \lambda _{i}+\eta )}{\sinh (\lambda _{j}^{\sigma }+\epsilon\lambda _{i})}
   \prod_{j\in \gamma _{+}}\frac{\sinh (\xi _{j}^{(1)}+\epsilon \lambda_{i})}{\sinh (\xi _{j}^{(0)}+\epsilon \lambda _{i})}\right\}  
   \nonumber \\
& \qquad \times 
  \prod_{i\in \alpha _{+}}\left\{ 
  \prod_{j\in \gamma _{+}}\frac{\sinh (\xi _{j}^{(1)}-\lambda _{i}^{\sigma })}{\sinh (\xi _{j}^{(0)}-\lambda_{i}^{\sigma })}\ 
  \frac{\prod_{j\in \alpha _{+}}\sinh (\lambda _{j}^{\sigma}-\lambda _{i}^{\sigma }-\eta )}{\prod_{j\in \alpha _{+}\setminus \{i\}}\sinh(\lambda _{j}^{\sigma }-\lambda _{i}^{\sigma })}
  \right\} 
  \prod_{1\leq i<j\leq s+s^{\prime }}
  \frac{\sinh (\lambda_{\mathsc{b}_i}^{\sigma }-\lambda_{\mathsc{b}_j}^{\sigma })}{\sinh (\lambda_{\mathsc{b}_i}^{\sigma }-\lambda_{\mathsc{b}_j}^{\sigma }-\eta )}  
  \nonumber \\
& \qquad \times 
  \prod_{p=1}^{s}\left[e^{-\eta x_{i_p}-\xi_{i_p}^{(1)}+\mu_{\mathsc{b}_p^\sigma} }\,
  \frac{\prod_{k=i_{p}+1}^{m}\sinh(\lambda _{\mathsc{b}_p}^{\sigma }-\xi_{k}^{(1)}-\eta )}
        {\prod_{k=i_{p}}^{m}\sinh (\lambda _{\mathsc{b}_{p}}^{\sigma }-\xi_{k}^{(1)})}\right] 
  \nonumber\\
 & \qquad \times
  \prod_{p=s+1}^{s+s'}\left[ e^{-\eta \bar x_{i_p}-\xi_{i_p}^{(1)}+\mu_{\mathsc{b}_p^\sigma} }\,
  \frac{\prod_{k=i_{p}+1}^{m}\sinh (\xi _{k}^{(1)}-\lambda _{\mathsc{b}_{p}}^{\sigma }-\eta )}
         {\prod_{\substack{ k=i_{p}  \\ k\not=M+m+1-{\mathsc{b}_{p}}}}^{m}\sinh (\xi _{k}^{(1)}-\lambda _{\mathsc{b}_{p}}^{\sigma })}\right] .
\end{align}
\end{cor}

%%%%%%%%
\subsection{Scalar products of separate states in the ungauged basis}

As a next step of our computation of the correlation functions of $\mathcal{T}^\text{(Inv)}(\lambda )$, let us now discuss the computation of the scalar products of separate states in the ungauged basis \eqref{ungauge-ket}-\eqref{ungauge-bra}. 

Let $Q\in\Sigma_Q^M$ with roots $\lambda_1,\ldots,\lambda_M$ be a solution of the homogeneous functional $TQ$-equation \eqref{TQ-hom} with some eigenvalue $\tau(\lambda )$ of $\mathcal{T}^\mathrm{(Inv)}(\lambda)$. We consider here the renormalised scalar product between the eigenstate $\bra{Q }$ and an arbitrary Bethe state:
\begin{equation}\label{SP}
    \mathrm{SP}_Q(\{\mu_1,\ldots,\mu_{M'}\})
    =\frac{\bra{Q}\,\prod_{i=1}^{M'}\mathcal{B}_- (\mu_i)\, \ket{\eta ,y_R }}
           {\bra{Q}\,\prod_{i=1}^M \mathcal{B}_- (\lambda_i)\, \ket{\eta,y_R}}
\end{equation}
for some arbitrary set of complex numbers $\{\mu_1,\ldots,\mu_{M'}\}$. Note that, from \eqref{sep-ABA-R-ungauge},  $\mathrm{SP}_Q(\{\mu_1,\ldots,\mu_{M'}\})$ \eqref{SP} corresponds to the renormalised scalar product of the two separate states $\bra{Q }$ and $\ket{P}$, where $P\in\Sigma_Q^{M'}$ has roots $\mu_1,\ldots,\mu_{M'}$:
\begin{equation}
  \mathrm{SP}_Q(\{\mu_1,\ldots,\mu_{M'}\})
  =\frac{\mathsf{c}_{Q} }{\mathsf{c}_P}\,  \frac{ \moy{Q\,|\, P} }{\moy{ Q\,|\, Q } },
\end{equation}
with $\mathsf{c}_P$ and $\mathsf{c}_{Q}$ given by \eqref{c-ABA-ungauge}, i.e.
\begin{equation}\label{c_Q/c_P}
    \frac{\mathsf{c}_{Q} }{\mathsf{c}_P}=(-1)^{N(M'-M)}\frac{\prod_{j=1}^{M'} b_-(\mu_j)}{\prod_{j=1}^M b_-(\lambda_j)}.
\end{equation}
\subsubsection{Finite-size determinant representations}

From the SoV form of the separate states and from the orthogonality properties and normalisation of the SoV basis, this scalar product can be rewritten in a Slavnov type form, by using the general results of \cite{KitMNT18}, and by specifying them to the current characterisation of the spectrum of the transfer matrices $\mathcal{T}^\text{(Inv)}(\lambda )$:

\begin{prop}\label{prop-SP-finite} 
Let $Q\in\Sigma_Q^M$ with roots $\{\lambda_1,\ldots,\lambda_M\}$ be a solution of the homogeneous functional $TQ$-equation \eqref{TQ-hom} with some eigenvalue $\tau(\lambda )$ of $\mathcal{T}^\mathrm{(Inv)}(\lambda)$. Let  $\mu_1,\ldots,\mu_{M'}$ be arbitrary complex numbers. 

Then $\mathrm{SP}_Q(\{\mu_1,\ldots,\mu_{M'}\})=0$ if $M'<M$, whereas %{\bf (check the normalisation)}
\begin{multline}\label{SP-finite}
   \mathrm{SP}_{Q} (\{\mu_1,\ldots,\mu_{M'}\})
   =  \left(\frac{\kappa_-\, e^{\tau_-}}{\sinh\varsigma_-}\right)^{\! M'-M}\, \frac{\Gamma_-^{(M+M')}}{\Gamma_-^{(2M)}}
   %\mathsf{c}_-^{(M,M')}\, %(-1)^{M+M'}\, \frac{\Gamma^{(M+M')}_{(\{\varphi_-,\psi_-\}}}{\Gamma^{(2M)}}\ 
   \\
   \times
   \frac{\prod_{i=1}^M\sinh (2\lambda_i+\eta )}{\prod_{i=1}^{M'}\sinh (2\mu_i+\eta )} \,
   \frac{\widehat{V}( \lambda_M,\ldots, \lambda_1)}{\widehat{V}(\mu_{M'},\ldots,\mu_1)}\,
   \frac{\det_{M'}\mathcal{S}_Q (\boldsymbol{\mu},\boldsymbol{\lambda}) }
   {\det_M\mathcal{S}_Q(\boldsymbol{\lambda},\boldsymbol{\lambda}) }
  % \times \frac{\widehat{V}(\lambda_1,\ldots,\lambda_M)}{\widehat{V}(\mu_1,\ldots,\mu_{M'})}\,\frac{\det_M\mathcal{S}_{\tau ,\{p_{i}\}_{i=1}^{n_{p}}}}{\det_{n_{q_{(a)}^{(\epsilon )}}}\mathcal{S}_{\tau ,\{q_{{(a)},_{i}}^{(\epsilon )}\}_{i=1}^{n_{q_{(a)}^{(\epsilon)}}}}}
\end{multline}
if $M'\ge M$, where we have defined
\begin{equation}
%   \{a_{1},a_{2}\}=\{\eps\varphi_-,\eps(\psi_-+i\tfrac\pi 2)\},
%   \quad\text{and}\quad
   \Gamma_-^{(n)}=
\begin{cases}
{\prod\limits_{j=1}^{n-N}\frac{\sinh (\varepsilon\varphi_-)\sinh (\varepsilon(\psi_-+i\tfrac\pi 2))}{\sinh (j\eta-\varepsilon(\varphi_-+\psi_-+i\tfrac\pi 2))}} 
           & \text{if }\ n\geq N,\vspace{2mm} \\ 
{\prod\limits_{j=0}^{N-n-1}\frac{\sinh (-j\eta -\varepsilon(\varphi_-+\psi_-+i\tfrac\pi 2))}{\sinh (\varepsilon\varphi_-)\sinh (\varepsilon(\psi_-+i\tfrac\pi 2))}} 
           & \text{if }\ n<N.
\end{cases}
\end{equation}
For $\boldsymbol{\nu}=(\nu_1,\ldots,\nu_{N_\nu})$ and $\boldsymbol{\lambda}=(\lambda_1,\ldots,\lambda_M)$ with $N_\nu\ge M$, the $M_\nu\times M_\nu$ matrix $\mathcal{S}_Q (\boldsymbol{\nu},\boldsymbol{\lambda}) $ is a generalised Slavnov matrix with elements:
\begin{alignat}{2}
   &\left[ \mathcal{S}_{Q} (\boldsymbol{\nu},\boldsymbol{\lambda}) \right]_{j,k}
   =Q(\nu_j)\,\frac{\partial \tau(\nu_j)}{\partial \lambda_k }
   \qquad\ & &\text{if}\quad k\leq M, \label{Slavnov1}\\
%   \nonumber\\
%   &= \sum_{\bar{\epsilon}=\pm }\bar{\epsilon}\,\mathbf{A}^{(\epsilon )}(\bar{\epsilon}\nu_j)\, Q(\nu_j-\bar\eps\eta)\big[ t(\nu_j+\lambda_k-\bar\eps\tfrac\eta 2)-t(\nu_j-\lambda_k-\bar\eps\tfrac\eta 2)\big],
   &\left[ \mathcal{S}_{Q} (\boldsymbol{\nu},\boldsymbol{\lambda}) \right]_{j,k}
   =\sum_{\sigma=\pm }\sigma\,\mathbf{A}^\mathrm{(Inv)}_\varepsilon(-\sigma\nu_j)\,\sinh (2\nu_j+\sigma\eta )\, Q(\nu_{j}+\sigma\eta )\,& &\left( \frac{\cosh (2\nu_{j}+\sigma\eta)}{2}\right) ^{k-M-1}  \notag \\   
   &\hspace{4cm} & &\text{if}\quad k >M. \label{Slavnov2}
\end{alignat}
%
%if $k>M$.
\end{prop}

\begin{proof}
Let us here specify explicitly the $\varepsilon$ dependence. For $Q\equiv Q^{(\varepsilon)}\in\Sigma_Q^{M}$  a solution of the $TQ$-equation \eqref{TQ-hom} and $P$ an arbitrary element of $\Sigma_Q^{M'}$, we therefore want to compute the following ratio of scalar products of separate states:
\begin{equation}\label{SP-eps}
      \frac{ {}_{\varepsilon}\moy{Q^{(\varepsilon)}\, |\, P}_\eps}{ {}_{\varepsilon}\moy{Q^{(\varepsilon)}\, |\, \widetilde Q}_\varepsilon}
      =\frac{ {}_{-\varepsilon}\moy{Q^{(-\varepsilon)}\, |\, P}_\varepsilon}{ {}_{-\varepsilon}\moy{Q^{(-\varepsilon)}\, |\, \widetilde Q}_\varepsilon},
\end{equation}
in which $ \widetilde Q\in\Sigma_Q^{M}$ has arbitrary roots $\widetilde \lambda_1,\ldots, \widetilde \lambda_M$ (we shall ultimately take the limit $\widetilde \lambda_i \to \lambda_i$, $1\le i \le M$).
In \eqref{SP-eps} we have used the fact that the states ${}_{\varepsilon}\bra{Q^{(\varepsilon)}}$ and ${}_{-\varepsilon}\bra{Q^{(-\varepsilon)}}$ are proportional, see Remark~\ref{rem-eps}, so that we can use either the right hand side or the left hand side of \eqref{SP-eps} to compute the ratio. Thanks to the orthogonality property of the SoV basis,  the ratio \eqref{SP-eps} for $\varepsilon'=\varepsilon$ or $\varepsilon'=-\varepsilon$ can be then rewritten as the following ratio of determinants:
\begin{multline}\label{ratio-det}
    \frac{ {}_{\varepsilon'}\moy{Q^{(\varepsilon')}\, |\, P}_\varepsilon}{ {}_{\varepsilon'}\moy{Q^{(\varepsilon')}\, |\, Q^{(\varepsilon)}}_\varepsilon}
    \\
    =\frac{\det_{1\le i,j\le N} \left[\sum_{h=0}^1 
     \left(-\frac{\mathsf{g}_-^{(\varepsilon )}(\frac\eta 2+\xi _{n})}{\mathsf{g}_-^{(\varepsilon' )}(\frac \eta 2-\xi _{n})} \right)^{\! h}\,
    Q^{(\varepsilon')}(\xi_i^{(h)})\,P(\xi_i^{(h)})\,\left(\frac{\cosh(2\xi_i^{(1-h)})}2\right)^{\! j-1} \right]}
    {\det_{1\le i,j\le N} \left[\sum_{h=0}^1 
     \left(-\frac{\mathsf{g}_-^{(\varepsilon )}(\frac\eta 2+\xi _{n})}{\mathsf{g}_-^{(\varepsilon' )}(\frac \eta 2-\xi _{n})} \right)^{\! h}\,
    Q^{(\varepsilon')}(\xi_i^{(h)})\, Q^{(\varepsilon)}((\xi_i^{(h)})\,\left(\frac{\cosh(2\xi_i^{(1-h)})}2\right)^{\! j-1} \right]},
\end{multline} 
in which we have used the representation \eqref{eigencovect-SoV-bis} and the relation \eqref{bra-eps-eps'}. Explicitly,
\begin{equation}
    \frac{\mathsf{g}_-^{(\varepsilon)}(\frac\eta 2+\xi _{n})}{\mathsf{g}_-^{(\varepsilon' )}(\frac \eta 2-\xi _{n})}
    =\frac{\sinh(\varepsilon\xi_n+\varphi_-)\,\cosh(\varepsilon\xi_n+\psi_-)}{\sinh(-\varepsilon'\xi_n+\varphi_-)\,\cosh(-\varepsilon'\xi_n+\psi_-)}
    =\frac{\sinh(\xi_n+\varepsilon\varphi_-)\,\sinh(\xi_n+\varepsilon\widetilde \psi_-)}{\sinh(\xi_n-\varepsilon'\varphi_-)\,\sinh(\xi_n-\varepsilon'\widetilde\psi_-)},
\end{equation}
in which we have set $\widetilde\psi_-=\psi_-+i\tfrac\pi 2$. This representation corresponds to the one obtained in Proposition~5.1 of \cite{KitMNT18}, so that we can directly use the results of \cite{KitMNT18} which follows from the transformations of this representation. Note that we are here in the simple case in which $n_{\boldsymbol{\varepsilon,\varepsilon'}}=2$ for $\varepsilon'=\varepsilon$, or $n_{\boldsymbol{\varepsilon,\varepsilon'}}=0$ if $\varepsilon'=-\varepsilon$, in formula (5.3) of \cite{KitMNT18}.

Let us first suppose that $M'<M=M^{(\varepsilon)}$. Then,  since $M^{(\varepsilon)}+M^{(-\varepsilon)}=N$ (see Proposition~\ref{prop-Wr}), we have $M'+M^{(-\varepsilon)}<N$. It then follows from the representation obtained in Theorem~5.1 of \cite{KitMNT18} for the choice $\varepsilon'=-\varepsilon$ that the scalar product ${}_{-\varepsilon}\moy{Q^{(-\varepsilon)}\, |\, P}_\varepsilon$ in the numerator of the right hand side of \eqref{SP-eps} vanishes (see Remark~2 of \cite{KitMNT18} just after Theorem~5.1). This proves the first assertion of Proposition~\ref{prop-SP-finite}.

Let us now consider the case $M'\geq M=M^{(\varepsilon)}$. Then taking the choice $\varepsilon^{\prime}=\varepsilon$, we are in the case $n_{\boldsymbol{\varepsilon,\varepsilon'}}=2$ of the general analysis of scalar products developed in \cite{KitMNT18}, see above. This allows us to use the representation obtained in Theorem 5.3 of \cite{KitMNT18}, with   the rank one matrix $\mathcal{P}$ being here equal to zero. % since we consider the simple case in which $n_{\boldsymbol{\varepsilon,\varepsilon'}}=2$.
Explicitly, it gives
\begin{multline}
    \frac{ {}_{\varepsilon}\moy{Q^{(\varepsilon)}\, |\, P}_\varepsilon}{ {}_{\varepsilon}\moy{Q^{(\varepsilon)}\, |\, \widetilde Q}_\varepsilon}
   =  (-1)^{N(M'-M)}\, \frac{\Gamma^{(M+M')}_- }{\Gamma^{(2M)}_- }\ 
   \frac{\prod_{i=1}^M\sinh (2\widetilde \lambda_i+\eta )\, \sinh (2\widetilde \lambda_i-\eta )}{\prod_{i=1}^{M'}\sinh (2\mu_i+\eta )\, \sinh (2\mu_i+\eta )} 
   \\ 
   \times \frac{\widehat{V}(\widetilde \lambda_M,\ldots,\widetilde \lambda_1)}{\widehat{V}(\mu_{M'},\ldots,\mu_1)}\,
   \frac{\det_{M'}\mathcal{S}_{Q^{(\varepsilon)}} (\boldsymbol{\mu},\boldsymbol{\lambda}) }
   {\det_M\mathcal{S}_{Q^{(\varepsilon)}}(\boldsymbol{\widetilde\lambda},\boldsymbol{\lambda}) }
\end{multline}
if $M'\ge M$.
% where we have defined
%
%\begin{equation}
%   \{a_{1},a_{2}\}=\{\eps\varphi_-,\eps(\psi_-+i\tfrac\pi 2)\},
%   \quad\text{and}\quad
 %  \Gamma_{\{a_1,a_2\}}^{(x)}=
%\begin{cases}
%{\prod\limits_{j=1}^{x-N}\frac{\sinh (a_{1})\sinh (a_{2})}{\sinh (j\eta-a_{1}-a_{2})}} 
%           & \text{if }\ x\geq N,\vspace{2mm} \\ 
%{\prod\limits_{j=0}^{N-x-1}\frac{\sinh (-j\eta -a_{1}-a_{2})}{\sinh(a_{1})\sinh (a_{2})}} 
%           & \text{if }\ x<N,
%\end{cases}
%\end{equation}
%
It remains to take the limit $\widetilde \lambda_i \to \lambda_i$, $1\le i\le M$, which is smooth in the matrix $\mathcal{S}_{Q^{(\varepsilon)}}(\boldsymbol{\widetilde\lambda},\boldsymbol{\lambda})$, and to take into account the normalisation \eqref{c_Q/c_P}.
\end{proof}

We are more particularly interested here in the situation in which the state $\prod_{i=1}^{M'}\mathcal{B}_- (\mu_i)\, \ket{\eta ,y_R }$ is obtained from the state $\prod_{i=1}^M \mathcal{B}_- (\lambda_i)\, \ket{\eta,y_R}$ by the action of some operator $ \underline{\mathbf E}_m^{\boldsymbol{\epsilon'},\boldsymbol{\epsilon}}(x)$ as in \eqref{act-boundary-lim}. In this situation, $\{\mu_1,\ldots,\mu_{M'}\}\subset\{\lambda_1,\ldots,\lambda_M\}\cup\{\xi_1^{(1)},\ldots,\xi_m^{(1)} \}$. Hence, let us set
\begin{equation}\label{mu-lambda}
   \{\mu_1,\ldots,\mu_{M'}\}=\{\lambda_{\pi_1},\ldots,\lambda_{\pi_n}\}\cup \{\xi_{\gamma_1}^{(1)},\ldots,\xi_{\gamma_{m'}}^{(1)}\}
\end{equation}
with $n+m'=M'$, %and $\{\lambda_1,\ldots,\lambda_M\}=\{\lambda_{\pi_1},\ldots,\lambda_{\pi_n}\}\cup\{\lambda_{\pi_{n+1}},\ldots,\lambda_{\pi_M}\}$, 
in which $\pi$ is a permutation of $\{1,\ldots,M\}$ whereas $\gamma$ is a permutation of $\{1,\ldots,m\}$. The expression \eqref{SP-finite} in the case $M'\ge M$ can then be rewritten as
\begin{align}\label{SP-finite-bis}
  & \mathrm{SP}_{Q} (\{\mu_1,\ldots,\mu_{M'}\})
   =  \left(\frac{\kappa_-\, e^{\tau_-}}{\sinh\varsigma_-}\right)^{\! M'-M}\, \frac{\Gamma_-^{(M+M')}}{\Gamma_-^{(2M)}}
   \nonumber\\
  &\quad \times
    \frac{\widehat{V}( \lambda_{\pi_M},\ldots, \lambda_{\pi_{n+1}})}{\widehat{V}(\xi_{\gamma_{m'}}^{(1)},\ldots,\xi_{\gamma_1}^{(1)})}\,
    \prod_{j=1}^n\frac{\prod_{k=n+1}^M(\sinh^2\lambda_{\pi_j}-\sinh^2\lambda_{\pi_k})}{\prod_{k=1}^{m'}(\sinh^2\lambda_{\pi_j}-\sinh^2\xi_{\gamma_k}^{(1)})}
   \nonumber \\
  &\quad \times
  \prod_{j=1}^{m'} \left[\frac{\mathbf{A}^\mathrm{(Inv)}_\varepsilon(-\xi_{\gamma_j}^{(1)})\, Q(\xi_{\gamma_j}^{(1)}-\eta) }{\sinh(2 \xi_{\gamma_j}^{(1)}+\eta)}\right]\prod_{j=n+1}^M\left[\frac{\sinh(2\lambda_{\pi_j}+\eta)}{\mathbf{A}^\mathrm{(Inv)}_\varepsilon(-\lambda_{\pi_j})\, Q(\lambda_{\pi_j}-\eta)}\right]
   \frac{\det_{M'}\mathcal{M} (\boldsymbol{\mu},\boldsymbol{\lambda_\pi}) }
   {\det_M\mathcal{N}(\boldsymbol{\lambda_\pi}) },
\end{align}
in which $\mathcal{N}(\boldsymbol{\lambda_\pi})$ is the $M\times M$ matrix with elements
\begin{equation}
   \big[ \mathcal{N}(\boldsymbol{\lambda_\pi}) \big]_{j,k}
   =\delta_{j,k}\, \frac\partial{\partial\mu}\left(i\log\frac{\mathbf{A}^\mathrm{(Inv)}_\varepsilon(-\mu)\, Q(\mu+\eta)}{\mathbf{A}^\mathrm{(Inv)}_\varepsilon(\mu)\, Q(\mu-\eta)}\right)\Big|_{\mu=\lambda_{\pi_j}}
     +2\pi \big[ K(\lambda_{\pi_j}-\lambda_{\pi_k})-K(\lambda_{\pi_j}+\lambda_{\pi_k})\big],
\end{equation}
whereas $\mathcal{M} (\boldsymbol{\mu},\boldsymbol{\lambda_\pi})$ is a $M'\times M'$ matrix taking  the following block form:
\begin{equation}
   \mathcal{M} (\boldsymbol{\mu},\boldsymbol{\lambda_\pi})
   = \begin{pmatrix} \mathcal{M}^{(1,1)} & \mathcal{M}^{(1,2)}  \\ 
       \mathcal{M}^{(2,1)} & \mathcal{M}^{(2,2)} \end{pmatrix}.
\end{equation}
The block matrices $\mathcal{M}^{(1,1)}$, $\mathcal{M}^{(1,2)}$, $\mathcal{M}^{(2,1)}$ and $\mathcal{M}^{(2,2)}$ have respective size $M\times n$, $M\times m'$, $(M'-M)\times M$ and $(M'-M)\times m'$, and their elements are given by
\begin{xalignat}{2}
    &\big[\mathcal{M}^{(1,1)}\big]_{j,k} = \big[ \mathcal{N}(\boldsymbol{\lambda_\pi}) \big]_{j,k},
    \qquad & & j\leq M,  \quad k\leq n,  
    \\
    &\big[\mathcal{M}^{(1,2)}\big]_{j,k} =i[t(\xi _{\gamma_k}-\lambda_{\pi_j})-t(\xi _{\gamma_k}+\lambda_{\pi_j})],
    \qquad  & &j\leq M,  \quad k\leq m',  
     \\
    &\big[\mathcal{M}^{(2,1)}\big]_{j,k} =\sum_{\sigma =\pm }\sinh (2\lambda_{\pi_k}+\sigma\eta)\left( \frac{\cosh (2\lambda_{\pi_k}+\sigma\eta )}{2}\right)^{\! j-1},
    \quad & & j\leq M'-M,  \ \, k\leq n, 
    \\
    &\big[\mathcal{M}^{(2,2)}\big]_{j,k} = \sinh (2\xi _{\gamma_k})\left( \frac{\cosh (2\xi _{\gamma_k})}{2}\right) ^{\! j-1},
    \quad & &j\leq M'-M,  \ \, k\leq m'.
\end{xalignat}
Here we have defined
\begin{align}
    &t(\lambda )=\frac{\sinh \eta }{\sinh (\lambda -\tfrac\eta 2)\,\sinh (\lambda+\tfrac\eta 2)}
                        =\coth (\lambda -\tfrac\eta 2)-\coth (\lambda +\tfrac\eta 2), \label{def-t}\\
    &K(\lambda)= \frac{i\sinh(2\eta)}{2\pi\sinh(\lambda+\eta)\,\sinh(\lambda-\eta)}
    =\frac{i}{2\pi}\big[t(\lambda+\tfrac\eta 2)+t(\lambda-\tfrac\eta 2)\big]. \label{def-K}
\end{align}
The ratio of determinants in \eqref{SP-finite-bis} can be rewritten in terms of a single determinant as
\begin{equation}\label{ratio-det-SP}
   \frac{\det_{M'}\mathcal{M} (\boldsymbol{\mu},\boldsymbol{\lambda_\pi}) }
   {\det_M\mathcal{N}(\boldsymbol{\lambda_\pi}) }
   =\det_{M'} \begin{pmatrix} \mathbb{I}_{n\times n} & \mathcal{W}^{(1,2)} \\ 
   \mathcal{W}^{(2,1)} & \mathcal{W}^{(2,2)} \end{pmatrix}
   %\mathcal{W} (\boldsymbol{\mu},\boldsymbol{\lambda_\pi})
   =\det_{m'}\left[ \mathcal{W}^{(2,2)}- \mathcal{W}^{(2,1)} \, \mathcal{W}^{(1,2)}\right],
\end{equation}
%
%in which 
%
%\begin{equation}
%   \mathcal{W} (\boldsymbol{\mu},\boldsymbol{\lambda_\pi})=
%   \begin{pmatrix} \mathbb{I}_{n\times n} & \mathcal{W}^{(1,2)} \\ 
%   \mathcal{W}^{(2,1)} & \mathcal{W}^{(2,2)} \end{pmatrix}
%\end{equation}
%
where $\mathbb{I}_{n\times n} $ is the $n\times n$ identity matrix and $\mathcal{W}^{(1,2)}$, $\mathcal{W}^{(2,1)}$ and $\mathcal{W}^{(2,2)}$ are matrices of respective size $n \times m'$, $m'\times n$ and $m' \times m'$, defined by
\begin{align}
  &\big[ \mathcal{W}^{(1,2)}\big]_{j,k} =\left[ \mathcal{N}(\boldsymbol{\lambda_\pi})^{-1}\mathcal{M}^{( 1,2) }\right] _{j,k},
  \qquad j\leq n,\ k\leq m',  
  \\
  &\big[ \mathcal{W}^{(2,1)}\big]_{j,k} =
  \begin{cases}
  0 &\text{if }\ j\leq M-n,\ k\leq n,  \vspace{1mm} \\
  \big[\mathcal{M}^{(2,1)}\big]_{j-(M-n),k} &\text{if }\ M-n < j\leq m',\  k\leq n,  
  \end{cases} 
  \\ 
  &\big[ \mathcal{W}^{(2,2)}\big]_{j,k} =
  \begin{cases}
  \left[ \mathcal{N}(\boldsymbol{\lambda_\pi})^{-1}\mathcal{M}^{( 1,2) }\right] _{j+n,k}
  &\text{if }\ j\leq M-n, \ k\leq m',\vspace{1mm}\\ 
   \sinh (2\xi _{\gamma_k})\left( \frac{\cosh (2\xi _{\gamma_k})}{2}\right) ^{\! j-M+n-1}
   &\text{if }\ M-n<j\leq m',\ k\leq m'.
   \end{cases}
\end{align}
Hence, the elements of the $m'\times m'$ matrix $\mathcal{S}'=\mathcal{W}^{(2,2)}- \mathcal{W}^{(2,1)} \, \mathcal{W}^{(1,2)}$ in \eqref{ratio-det-SP} are given by
\begin{align}
   &\left[ \mathcal{S}' \right]_{j,k}= \left[ \mathcal{N}(\boldsymbol{\lambda_\pi})^{-1}\mathcal{M}^{( 1,2) }\right] _{j+n,k}
   \qquad \text{if }\ j\le M-n, \label{mat-S'1}\\
   &\left[ \mathcal{S}' \right]_{j+N-n,k}= \sinh (2\xi _{\gamma_k})\left( \frac{\cosh (2\xi _{\gamma_k})}{2}\right) ^{\! j-1}
   \nonumber\\
   &\hspace{2cm}
   -\sum_{\ell=1}^n  \sum_{\sigma =\pm }\sinh (2\lambda_{\pi_\ell}+\sigma\eta)\left( \frac{\cosh (2\lambda_{\pi_\ell}+\sigma\eta )}{2}\right)^{\! j-1} \left[ \mathcal{N}(\boldsymbol{\lambda_\pi})^{-1}\mathcal{M}^{( 1,2) }\right] _{\ell,k}.
   \label{mat-S'2}
\end{align}

%%%
\subsubsection{Thermodynamic behaviour of scalar products}

Let us now study the thermodynamic behaviour, when $N\to +\infty$, of the scalar products \eqref{SP-finite-bis} in the specific case in which $\{\lambda_1,\ldots,\lambda_M\}$ stand for the Bethe roots of the ground state of the chain.
%The case in which $M'=M$ has already been studied in \cite{NicT23}, and we recall that such  scalar product vanishes if $M'<M$ from Proposition~\ref{prop-SP-finite}. Hence, here, we shall only consider the remaining case in which $M' > M$. For technical reasons, we shall moreover focus here on the antiferromagnetic regime $\Delta>1$.
The spin chain under consideration being isospectral to an open spin chain with diagonal boundary conditions, see Theorem~\ref{Theo-Sp-Inv}, the description of its ground state is well known, see \cite{SkoS95,KapS96,GriDT19}. 

We recall that nearly all Bethe roots $\lambda_i$ for the ground state are real in the regime $|\Delta|<1$, or purely imaginary in the regime $\Delta>1$. In the thermodynamic limit $N\to\infty$, these real/purely imaginary Bethe roots for the ground state condensate on an interval of the real/imaginary axis with some density function $\rho(\lambda)$ solution of the following integral equation:
\begin{equation}\label{int-rho-ext}
    \rho(\lambda)+\int_{-\Lambda}^\Lambda K(\Lambda-\mu)\,\rho(\mu)\,d\mu=\frac{it(\lambda)}\pi,
\end{equation}
where $\Lambda=+\infty$ in the regime $|\Delta|<1$, whereas $\Lambda=-i\tfrac\pi 2$ in the regime $\Delta>1$.
%$t$ and $K$ are given by \eqref{def-t} and \eqref{def-K}. Explicitly we have
Explicitly, we have
\begin{equation}\label{rho}
   \rho(\lambda)
   =\left\lbrace\,
   \begin{array}{@{}l@{\quad}l@{\quad}l@{}}
   \displaystyle
   \frac{1}{\zeta\,\cosh(\frac{\pi\lambda}\zeta)} &\text{with} \ \ \zeta=i\eta>0  \  &\text{if}\ \ |\Delta|<1,
   \vspace{2mm}\\
   \displaystyle 
   \frac i\pi\frac{\vartheta'_1(0,q)}{\vartheta_2(0,q)}\frac{\vartheta_3(i\lambda,q)}{\vartheta_4(i\lambda,q)}  
                 &\text{with} \ \ q=e^{-\zeta},\ \zeta=-\eta>0 \  &\text{if}\ \ \Delta>1,
   \end{array}\right.
\end{equation}
where $\vartheta_i(u,q)$, $i\in\{1,2,3,4\}$, are the theta-functions of name $q$ defined as in \cite{GraR07L}. 
In addition to the aforementioned real/purely imaginary Bethe roots, the set of Bethe roots for the ground state may also contain extra isolated complex roots, called boundary roots. There are two possible boundary roots $\check \lambda_\sigma=\tfrac\eta 2- \varsigma_\sigma^{(D)}+\check\epsilon_\sigma$ for $\sigma=\pm$, with $\check\epsilon_\sigma$ being an exponentially small correction in $N$   in the large $N$ limit.  As shown in \cite{SkoS95,KapS96,GriDT19}, the presence of such a boundary root within the set of roots for the ground state depends on both boundary parameters $\varsigma_+^{(D)}$ and $\varsigma_-^{(D)}$ and on the parity of the number of sites $N$ of the chain. 

Let us now consider the thermodynamic behaviour of the elements of the matrix $\mathcal{S}'=\mathcal{W}^{(2,2)}- \mathcal{W}^{(2,1)} \, \mathcal{W}^{(1,2)}$ in \eqref{ratio-det-SP}. It was shown in Section 4.4 of \cite{KitKMNST07} that\footnote{We suppose here for simplicity that $\varsigma_+^{(D)}\not=\varsigma_-^{(D)}\mod i\pi$.}
\begin{align}\label{mat-R-thermo}
    %\left[ \mathcal{N}^{-1}\mathcal{M}^{(1,2)}\right] _{i,k} 
    \left[ \mathcal{N}(\boldsymbol{\lambda_\pi})^{-1}\mathcal{M}^{( 1,2) }\right] _{j,k}
    \underset{N\to\infty}{\sim} \begin{cases}
     \displaystyle\frac{\rho(\lambda_{\pi_j}-\xi_{\gamma_k})-\rho(\lambda_{\pi_j}+\xi_{\gamma_k})}{2N\rho(\lambda_{\pi_j})} \quad
          &\text{if }\lambda_{\pi_j}\in(0,\Lambda),
          \vspace{1mm}\\
    \displaystyle  i\pi\,\check\epsilon_\sigma\big[\rho(\lambda_{\pi_j}-\xi_{\gamma_k})-\rho(\lambda_{\pi_j}+\xi_{\gamma_k})\big] \quad 
          &\text{if }\lambda_{\pi_j}=\check\lambda_\sigma,
    \end{cases}
\end{align}
which gives in particular the form of the first $M-n$ lines \eqref{mat-S'1} of the matrix $\mathcal{S}'$.
These are the only elements of $\mathcal{S}'$ in the case in which $M'=M$, i.e. $m'=M-n$. However, in the case $M'>M$, i.e. $m'>M-n$, the determinant contains some extra lines \eqref{mat-S'2}.

Let us first investigate the behaviour of these extra lines in the antiferromagnetic regime $\Delta>1$. Let us observe that, $n$ being of the same order as $M$ in the thermodynamic limit (we recall that $m$, and hence $m'$, remain finite in this limit), we have the following sum-to-integral rule for large $N$:
% the sum in \eqref{mat-S'2} tends to an integral in the thermodynamic limit according to the rule:
%
\begin{align}
    \frac{1}N\sum_{\ell=1}^n f(\lambda_{\pi_\ell})
    =\frac{1}{N}\!\!\!\! \sum_{\substack{j=1 \\ \lambda_j\in(0,\Lambda)}}^{M}  \!\!\!\!  f(\lambda_j) +O(\tfrac 1 N)
    \underset{N\rightarrow \infty }{\longrightarrow }
    \int_{0}^{\Lambda }d\lambda\, f(\lambda)\, \rho (\lambda) %+O(\tfrac 1 N)
    =\frac 12 \int_{-\Lambda }^{\Lambda }d\lambda\,f(\lambda)\,\rho(\lambda), % +O(\tfrac 1 N),
\end{align}
for $f$ being a smooth even $i\pi$-periodic function on $i\mathbb{R}$, the corrections to this limit being again of order $O(\frac 1N)$, see \cite{GriDT19}. It then follows that the extra lines \eqref{mat-S'2} become, in the thermodynamic limit:
\begin{align}\label{S-extra-gen}
    \left[ \mathcal{S}' \right]_{j+N-n,k}
    &= \sinh (2\xi _{\gamma_k})\left( \frac{\cosh (2\xi _{\gamma_k})}{2}\right) ^{\! j-1}
   \nonumber\\
  & - \int_{-\Lambda }^{\Lambda }  \sum_{\sigma =\pm }\sinh (2\lambda+\sigma\eta)\left( \frac{\cosh (2\lambda+\sigma\eta )}{2}\right)^{\! j-1} 
    \frac{\rho(\lambda-\xi_{\gamma_k})-\rho(\lambda+\xi_{\gamma_k})}4\, d\lambda +O(\tfrac 1N)
   %\left[ \mathcal{N}(\boldsymbol{\lambda_\pi})^{-1}\mathcal{M}^{( 1,2) }\right] _{\ell,k}
   \nonumber\\
   &= \sinh (2\xi _{\gamma_k})\left( \frac{\cosh (2\xi _{\gamma_k})}{2}\right) ^{\! j-1}
   \nonumber\\
   & -\frac 12 \int_{-\Lambda }^{\Lambda }  \sum_{\sigma =\pm }\sinh (2\lambda+\sigma\eta)\left( \frac{\cosh (2\lambda+\sigma\eta )}{2}\right)^{\! j-1} 
   \rho(\lambda-\xi_{\gamma_k})\, d\lambda +O(\tfrac 1N),
\end{align}
in which we have used the parity of the density function $\rho$.
It is convenient, in the antiferromagnetic regime $\Delta>1$, to perform the change of variables 
\begin{equation}\label{change-var}
    u_j=i\lambda_j, \qquad \tilde \xi_k=i\xi_k, \qquad \zeta=-\eta>0, \qquad \rho_{\Delta>1}(u)=-i\rho(-iu)= \frac 1\pi\frac{\vartheta'_1(0,q)}{\vartheta_2(0,q)}\frac{\vartheta_3(u,q)}{\vartheta_4(u,q)}  ,
\end{equation}
or equivalently, to perform the change of variables $u=i\lambda$ in the above integral. This gives
\begin{align}\label{S-extra-anti}
  i\left[ \mathcal{S}' \right]_{j+N-n,k}
   &= \sin(2\tilde\xi_{\gamma_k})\left( \frac{\cos (2\tilde \xi _{\gamma_k})}{2}\right) ^{\! j-1} 
   \nonumber\\
%   &\hspace{1cm}
  & -\frac 12 \int_{-\frac\pi 2 }^{\frac\pi 2 }  \sum_{\sigma =\pm }\sin (2u-\sigma\zeta)\left( \frac{\cos (2u-\sigma\zeta )}{2}\right)^{\! j-1}  \rho_{\Delta>1}(u-\tilde\xi_{\gamma_k})\, du +O(\tfrac 1N)
  \nonumber\\
  &= \sin(2\tilde\xi_{\gamma_k})\left( \frac{\cos (2\tilde \xi _{\gamma_k})}{2}\right) ^{\! j-1} 
  -\int_{-\frac\pi 2-i\frac\zeta 2}^{\frac\pi 2-i\frac\zeta 2} f_j(u)\, du 
  + \int_{-\frac\pi 2+i\frac\zeta 2}^{\frac\pi 2+i\frac\zeta 2}   f_j(u)\, du +O(\tfrac 1N)
  \nonumber\\
  &= 2i\pi\Res f_j(z)\Big|_{z=\tilde\xi_{\gamma_k}}-\int_{\Gamma_\zeta} f_j(z)\,dz +O(\tfrac 1N)
   \nonumber\\
  &= O(\tfrac 1N).
\end{align}
Here, we have defined the meromorphic functions $f_j$, $1\le j \le m'$,  as
\begin{align}
   f_j(z)=\frac {\sin(2z)}2 \, \left(\frac{\cos(2z)}2\right)^{\! j-1}\rho_{\Delta>1}(u-\tilde\xi_{\gamma_k}+i\tfrac\zeta 2),
\end{align}
and $\Gamma_\zeta$ is a rectangular counter-clockwise close contour with vertices at $-\frac\pi 2-i\frac\zeta 2$, $\frac\pi 2-i\frac\zeta 2$, $\frac\pi 2+i\frac\zeta 2$ and $-\frac\pi 2+i\frac\zeta 2$. In the second equality of \eqref{S-extra-anti}, we have used the quasi-periodicity property
\begin{equation}
   \rho_{\Delta>1}(v+i\zeta)=-\rho_{\Delta>1}(v),
\end{equation}
of the function $\rho_{\Delta>1}$, whereas in the third equality of \eqref{S-extra-anti}, we have used the periodicity property $f_j(v+\pi)=f_j(v)$ of the function $f_j$, which ensures that the integral of this function over the two vertical segments of $\Gamma_\zeta$ cancels out.

Hence, we have proven that, in the antiferromagnetic regime $\Delta>1$, each extra line \eqref{mat-S'2} of the determinant \eqref{ratio-det-SP} is at most of order $O(\frac 1N)$. In other words, in the situation \eqref{mu-lambda}, where $\{\lambda_1,\ldots,\lambda_M\}$ stand for the Bethe roots of the ground state of the chain, the renormalised scalar product \eqref{SP-finite-bis} is at most of order $O( \frac 1{N^{m'}})$. Actually, what will be important for the computation of the correlation functions is that\footnote{The estimation \eqref{SP-thermo} is valid notably when all roots $\{\lambda_{\pi_1},\ldots,\lambda_{\pi_n}\}$ are purely imaginary, i.e. belong to $(0,\Lambda)$. If one of the root is a boundary root $\check\lambda_\sigma$ (see \eqref{mat-R-thermo}), then one can use instead the following more accurate estimation:
\begin{equation}\label{SP-thermo-BR}
  \mathrm{SP}_{Q} (\{\lambda_{\pi_1},\ldots,\lambda_{\pi_n}\}\cup \{\xi_{\gamma_1}^{(1)},\ldots,\xi_{\gamma_{m'}}^{(1)}\})=o(\check\epsilon_\sigma/N^{M-n-1}) \qquad \text{if } n+m'>M.
\end{equation}
}
\begin{equation}\label{SP-thermo}
  \mathrm{SP}_{Q} (\{\lambda_{\pi_1},\ldots,\lambda_{\pi_n}\}\cup \{\xi_{\gamma_1}^{(1)},\ldots,\xi_{\gamma_{m'}}^{(1)}\})=o(1/N^{M-n}) \qquad \text{if } n+m'>M,
\end{equation}
which we have proven here in the antiferromagnetic regime $\Delta>1$ for $\{\lambda_1,\ldots,\lambda_M\}$ being the Bethe roots of the ground state.

The previous reasoning is not directly generalisable to the critical regime $|\Delta|<1$, due to the fact that, in this regime, the domain of integration $(-\Lambda,\Lambda)$ becomes unbounded, so that carelessly replacing sums by integrals as in \eqref{S-extra-gen} may lead to diverging integrals. Nevertheless, a similar reasoning can be made to show the cancellation of the first extra line in the case in which $0<\Delta<1$, i.e. $0<\zeta=i\eta<\tfrac\pi 2$. More precisely, we have in that case
\begin{align}\label{S-extra1-critical}
    \left[ \mathcal{S}' \right]_{1+N-n,k}
      &= \sinh (2\xi _{\gamma_k}) -\frac 12 \int_\mathbb{R}  \sum_{\sigma =\pm }\sinh (2\lambda+\sigma\eta)\,
   \rho(\lambda-\xi_{\gamma_k})\, d\lambda +O(\tfrac 1N)
   \nonumber\\
   &= \sinh (2\xi _{\gamma_k}) -\lim_{R\to +\infty} \int_{\Gamma^{(R)}_{\zeta}}  \frac{\sinh (2z)}2\,
   \rho(z-\xi_{\gamma_k}+i\tfrac\zeta 2)\, dz +O(\tfrac 1N),
%   \nonumber\\
%   &= O(\tfrac 1N),
\end{align}
in which $\Gamma_\zeta^{(R)}$ is a rectangular counterclockwise close contour with vertices at $-R-i\frac\zeta 2$, $R-i\frac\zeta 2$, $R+i\frac\zeta 2$ and $-R+i\frac\zeta 2$. Here we have used again that $\rho(z+i\zeta)=-\rho(z)$, and the cancellation of the integral on the vertical segments of $\Gamma_\zeta^{(R)}$ vers $R\to +\infty$ comes from the fact that the integrand tends uniformly to zero on these segments for $0<\zeta<\tfrac\pi 2$. Moreover, since %the residue of the function under the integral at the pole $z=\xi _{\gamma_k})$ being $
\begin{equation}
   2i\pi \Res \left[\sinh (2z)\, \rho(z-\xi_{\gamma_k}+i\tfrac\zeta 2)\right]\Big|_{z=\xi_{\gamma_k}}=2\sinh(2\xi_{\gamma_k}),
\end{equation}
each element of the line \eqref{S-extra1-critical} goes to zero in the thermodynamic limit for $\zeta <\pi /2$\footnote{Note that this result is also valid in the XXX case $\Delta=1$, see \cite{Nic21}, and we can argue that it also holds for $\zeta =\pi /2$, i.e. the XX case. Indeed, we have just shown that the difference of integrals in the first line of \eqref{S-extra1-critical} is $\sinh 2\xi _{\gamma _{k}}$ independently from $\zeta $, for any $\zeta <\pi /2$, which should remain constant also in the limit $\zeta =\pi /2$.}. Hence, provided the other lines, when combined with the pre-factor, do not introduce further divergences (which would remain to be rigorously proven), this argument tends to indicate that the estimation \eqref{SP-thermo} also holds in the regime $0<\Delta<1$. A natural conjecture is that it still holds in the whole critical regime $|\Delta|<1$.

%\end{proof}

%%%%%%%%%%%%
\subsection{Complete set of elementary blocks for the correlation functions}

From the previous results on the scalar products of separate states and transfer matrix eigenstates, and from the action of local operators on separate states, we can now compute the matrix elements of the complete set of quasi-local operators $\underline{\mathbf E}_m^{\boldsymbol{\epsilon'},\boldsymbol{\epsilon}}(x)$ \eqref{operator-m-lim}, for $x$ given in terms of $y_R$ \eqref{Ref-SoV-R} as in \eqref{cond-x-y}, in an eigenstate $\ket{Q}$ of the transfer matrix $\mathcal{T}^\mathrm{(Inv)}(\lambda )$:
\begin{equation}\label{matrix-el}
    \moy{\underline{\mathbf E}_m^{\boldsymbol{\epsilon'},\boldsymbol{\epsilon}}(x)}
    \equiv \frac{\bra{ Q }\,\underline{\mathbf E}_m^{\boldsymbol{\epsilon'},\boldsymbol{\epsilon}}(x)\, \ket{Q } }{\moy{Q\, |\,Q} },
    \qquad \text{for} \quad \boldsymbol{\epsilon},\boldsymbol{\epsilon'}\in \{1,2\}^{m}.
\end{equation}
%
%The obtained result depends on the value of $\tilde{m}_{\boldsymbol{\epsilon,\epsilon'}}$ \eqref{def-tildem}.
This result generalises, in this particular simple case, the result obtained in \cite{NicT23} in the case $\tilde{m}_{\boldsymbol{\epsilon,\epsilon'}}=0$.

%\subsubsection{The finite chain}

In finite volume, these matrix elements are given by the following expressions: % in terms of multiple sums:

\begin{theorem}\label{th-FiniteMatrixE}
Let $Q\in\Sigma_Q^M$ with roots $\{\lambda_1,\ldots,\lambda_M\}$ be a solution of the homogeneous $TQ$-equation \eqref{TQ-hom} associated to a transfer matrix eigenvalue $\tau(\lambda)$ of $\mathcal{T}^\mathrm{(Inv)}(\lambda )$, and let $x$ be defined by
\begin{equation}\label{cond-x-y-bis}
  x=y_R +2M+1-N,
\end{equation}
in terms of $y_R$ \eqref{Ref-SoV-R}.
Then, the matrix elements \eqref{matrix-el} of a quasi-local operators $\underline{\mathbf E}_m^{\boldsymbol{\epsilon'},\boldsymbol{\epsilon}}(x)$ \eqref{operator-m-lim} in the eigenstate $\ket{Q}$ of the transfer matrix $\mathcal{T}^\mathrm{(Inv)}(\lambda )$ vanishes identically when $\tilde{m}_{\boldsymbol{\epsilon,\epsilon'}}<0$, in which $\tilde{m}_{\boldsymbol{\epsilon,\epsilon'}}$ is given in terms of $\boldsymbol{\epsilon},\boldsymbol{\epsilon'}$ as in \eqref{def-tildem}:
\begin{equation}\label{matrix-el-0}
    \moy{\underline{\mathbf E}_m^{\boldsymbol{\epsilon'},\boldsymbol{\epsilon}}(x)}=0
    \qquad \text{if}\quad  \tilde{m}_{\boldsymbol{\epsilon,\epsilon'}}<0.
\end{equation}
If instead $\tilde{m}_{\boldsymbol{\epsilon,\epsilon'}} \ge 0$, it admits the following multiple sum representation:
\begin{equation}
     \moy{\underline{\mathbf E}_m^{\boldsymbol{\epsilon'},\boldsymbol{\epsilon}}(x)}
%    \frac{\bra{ Q }\,\underline{\mathbf E}_m^{\boldsymbol{\epsilon'},\boldsymbol{\epsilon}}(x)\, \ket{Q } }{\moy{Q\, |\,Q} }
    =\sum_{\mathsc{b}_1=1}^M\ldots \sum_{\mathsc{b}_{s}=1}^M\sum_{\mathsc{b}_{s+1}=1}^{M+m}\ldots \sum_{\mathsc{b}_{s+s^{\prime }}=1}^{M+m}\frac{H_{\{\mathsc{b}_{j}\}}(\{\lambda_{i}\}_{i=1}^M)}
            {\prod\limits_{1\leq l<k\leq m}\!\!\!\!\sinh (\xi _{k}-\xi _{l})\prod\limits_{1\leq p\leq q\leq m}\!\!\!\!\sinh (\xi _{p}+\xi _{q})} ,
\end{equation}
where $\lambda_{M+j}=\xi _{m+1-j}^{(1)}$ for $j\in \{1,\ldots ,m\}$ and
\begin{align}
&H_{\{\mathsc{b}_{j}\}}(\{\lambda \}_{i=1}^M)
    = (-1)^{(N+1)\tilde{m}_{\boldsymbol{\epsilon,\epsilon'}}} 
   \left(\frac{\kappa_-\, e^{\tau_-}}{\sinh\varsigma_-}\right)^{\! \tilde{m}_{\boldsymbol{\epsilon,\epsilon'}}}\, \frac{\Gamma_-^{(2M+\tilde{m}_{\boldsymbol{\epsilon,\epsilon'}})}}{\Gamma_-^{(2M)}}
   \prod_{n=1}^m e^{\eta x_n}
%   \left( \frac{\prod\limits_{j=1}^{m}\sinh (\xi _{j}-\epsilon \varphi ^{(2)})}{\prod\limits_{i=1}^{s+s^{\prime }}\sinh (\lambda _{\text{\textsc{b}}_{i}}^{\sigma }-\epsilon \varphi ^{(2)}+\eta /2)}\right) ^{\delta _{a,2}} 
\notag \\
&\hspace{1cm} \times 
\sum\limits_{\sigma _{\mathsc{b}_{j}}}\frac{(-1)^{s+m-n}
\prod\limits_{i=1}^{s+s^{\prime }}\sigma _{\mathsc{b}_{i}}
\prod\limits_{i=1}^{s+s^{\prime }}\prod\limits_{j=1}^{m}\sinh (\lambda_{\mathsc{b}_{i}}^{\sigma }+\xi _{j}+\eta /2)}{\prod\limits_{1\leq i<j\leq s+s^{\prime }}\sinh (\lambda _{\mathsc{b}_{i}}^{\sigma}-\lambda _{\mathsc{b}_{j}}^{\sigma }-\eta )\sinh (\lambda _{\mathsc{b}_{i}}^{\sigma }+\lambda _{\mathsc{b}_{j}}^{\sigma }+\eta )}
\nonumber\\
&\hspace{1cm} \times 
\prod\limits_{p=1}^{s}\left\{ e^{-\eta x_{i_p}-\xi_{i_p}^{(1)}+\lambda_{\mathsc{b}_p^\sigma} }\,
\prod\limits_{k=1}^{i_{p}-1}\sinh(\lambda _{\mathsc{b}_{p}}^{\sigma }-\xi_{k}^{(1)})
\prod\limits_{k=i_{p}+1}^{m}\!\!\sinh (\lambda _{\mathsc{b}_{p}}^{\sigma }-\xi _{k}^{(0)})\right\}   
\notag \\
&\hspace{1cm} \times 
\prod\limits_{p=s+1}^{s+s^{\prime }}\left\{ e^{-\eta \bar x_{i_p}-\xi_{i_p}^{(1)}+\lambda_{\mathsc{b}_p^\sigma} }\,
\prod\limits_{k=1}^{i_{p}-1}\sinh (\lambda _{\mathsc{b}_{p}}^{\sigma}-\xi _{k}^{(1)})
\prod\limits_{k=i_{p}+1}^{m}\!\!\sinh (\lambda _{\mathsc{b}_{p}}^{\sigma }-\xi _{k}^{(1)}+\eta )\right\}  \
%\notag \\
%& \hspace{9cm} \times 
\det_{m'}\Omega .  \label{H-finite}
\end{align}
Here, the sum is performed over all $\sigma _{\mathsc{b}_j}\in\{+,-\}$ for $\mathsc{b}_j \leq M$, and $\sigma _{\mathsc{b}_j}=1$ for $\mathsc{b}_j> M$, and the $m'\times m'$ matrix $\Omega $ reads for\footnote{Note that all the values of $l\in \{1,\ldots,M-n\}$ are exactly the values of $l$ such that $\mathsc{b}_l\leq M$, by the definition of the sets $\mathsf{B}_{\boldsymbol{\epsilon,\epsilon'}}$ and $\alpha_{-}$.}  $l\leq M-n,\ k\leq m'$: 
\begin{equation}
\Omega_{l,k}=i\sum_{j=1}^M \mathcal{N}_{\mathsc{b}_l,j}^{-1}[t(\xi _{\gamma_{k}}-\lambda_j-t(\xi _{\gamma
_{k}}+\lambda_j)]
\end{equation}
while for $1\leq l\leq n+m'-M,\ k\leq m',$

\begin{align}
\Omega_{M-n+l,k}& =i\sum_{r=1}^n\left[ \sum_{\sigma =\pm }\sinh (2\lambda_{\pi_{r}}+\sigma \eta )\left( \frac{\cosh (2\lambda_{\pi_{r}}+\sigma \eta )}{2}\right) ^{l-1}\right]  
\notag \\
& \times \sum_{j=1}^M\mathcal{N}_{\pi_{r},j}^{-1}[t(\xi _{\gamma _{k}}+\lambda_j)-t(\xi _{\gamma _{k}}-\lambda_j] +\sinh(2\xi _{\gamma _{k}})\left( \frac{\cosh( 2\xi _{\gamma _{k}})}{2}\right)^{l-1}.
\end{align}

\end{theorem}

%\subsubsection{Correlation functions in semi-infinite chain}

Let us now assume that $\ket{Q}$ is the ground state of the model, so that the considered matrix elements \eqref{matrix-el}, for all $\boldsymbol{\epsilon},\boldsymbol{\epsilon'}\in \{1,2\}^{m}$, constitute the complete set of elementary building blocks for the zero-temperature correlation functions of local operators $\mathcal{O}_m\in\otimes_{n=1}^m\mathcal{H}_n$.
Then, considering the thermodynamic limit $N\to +\infty$ of the previous result, we obtain that, in the antiferromagnetic regime $\Delta>1$:
\begin{equation}\label{matrix-el-therm-0}
    \lim_{N\to +\infty}\moy{\underline{\mathbf E}_m^{\boldsymbol{\epsilon'},\boldsymbol{\epsilon}}(x)} =0
    \qquad \text{if}\quad \tilde{m}_{\boldsymbol{\epsilon,\epsilon'}}\not= 0.
\end{equation}
The vanishing of \eqref{matrix-el-therm-0} for $\tilde{m}_{\boldsymbol{\epsilon,\epsilon'}}> 0$ directly follows from \eqref{SP-thermo}, which comes from the vanishing of all extra lines of the determinant of the scalar product in the regime $\Delta>1$. When $0<\Delta<1$, the first extra line of the determinant still vanishes, although the thermodynamic limit is less controllable in that case, so that we believe that \eqref{SP-thermo}, and therefore \eqref{matrix-el-therm-0}, still hold for these values of $\Delta$. A natural conjecture, although less well-founded, is that they still hold in the whole critical regime $|\Delta|<1$.

In the case $\tilde{m}_{\boldsymbol{\epsilon,\epsilon'}}=0$, one obtains
\begin{multline}\label{el-block-thermo}
\lim_{N\to +\infty}\moy{\underline{\mathbf E}_m^{\boldsymbol{\epsilon'},\boldsymbol{\epsilon}}(x)} 
   = \frac{(-1)^{s}  \prod\limits_{n=1}^m e^{\eta x_n}}{\prod\limits_{j<i}\sinh(\xi _{i}-\xi _{j})\prod\limits_{i\leq j}\sinh (\xi _{i}+\xi _{j})}
 \\
 \times 
  \int_{\mathcal C}\prod_{j=1}^{s}d\lambda _{j}\ \int_{\mathcal{C}_{\boldsymbol \xi}}\prod_{j=s+1}^{m}\!\!d\lambda _{j}\ 
  H_{m}(\{\lambda _{j}\}_{j=1}^{M};\{\xi _{k}\}_{k=1}^{m})\ 
  \det_{1\leq j,k\leq m}\big[\Phi (\lambda _{j},\xi _{k})\big],
\end{multline}
in which 
\begin{equation}
\Phi (\lambda _{j},\xi _{k})=\frac{\rho (\lambda _{j}-\xi _{k})-\rho(\lambda _{j}+\xi _{k})}{2},  \label{mat-Phi}
\end{equation}
and 
\begin{align}\label{Hm}
&H_{m}(\{\lambda _{j}\}_{j=1}^{M};\{\xi _{k}\}_{k=1}^{m})
   =\frac{\prod\limits_{j=1}^{m}\prod\limits_{k=1}^{m}\sinh (\lambda _{j}+\xi
_{k}+\eta /2)}{\!\!\!\!\prod\limits_{1\leq i<j\leq m}\!\!\!\!\sinh (\lambda
_{i}-\lambda _{j}-\eta )\,\sinh (\lambda _{i}+\lambda _{j}+\eta )}  \notag \\
& \hspace{2cm}\times 
   \prod\limits_{p=1}^{s}\left\{ e^{-\eta x_{i_p}-\xi_{i_p}^{(1)}+\lambda_p }\,
   \prod\limits_{k=1}^{i_{p}-1}\sin(\lambda _{p}-\xi _{k}^{(1)})
   \prod\limits_{k=i_{p}+1}^{m}\!\!\sinh (\lambda_{p}-\xi _{k}^{(1)}-\eta )\right\}  
   \notag \\
& \hspace{2cm}\times 
   \!\!\prod\limits_{p=s+1}^{m}\left\{ e^{-\eta \bar x_{i_p}-\xi_{i_p}^{(1)}+\lambda_p }\,
   \prod\limits_{k=1}^{i_{p}-1}\sinh (\lambda _{p}-\xi_{k}^{(1)})
   \prod\limits_{k=i_{p}+1}^{m}\!\!\sinh (\lambda _{p}-\xi_{k}^{(1)}+\eta )\right\}  .
%   \notag \\
%& \times \left( \prod\limits_{i=1}^{m}\frac{\sinh (\xi _{j}-\epsilon \varphi^{(2)})}{\sinh (\lambda _{\text{\textsc{b}}_{i}}^{\sigma }-\epsilon \varphi^{(2)}+\eta /2)}\right) ^{\delta _{a,2}}.
\end{align}
The contours $\mathcal{C}$ is here defined as 
\begin{equation}
\mathcal{C}=[-\Lambda ,\Lambda ],  \label{C-nude}
\end{equation}
and the contour $\mathcal{C}_{\boldsymbol{\xi}}$ as 
\begin{equation}
\mathcal{C}_{\boldsymbol{\xi}}=\mathcal{C}\cup \Gamma^{-} (\{\xi_{k}^{(1)}\}_{k=1}^{m}),  \label{C-xi}
\end{equation}
where $\Gamma^{-} (\{\xi_{k}^{(1)}\}_{k=1}^{m})$) surrounds the  points $\xi _{1}^{(1)},\ldots ,\xi _{m}^{(1)}$ with index $-1$, all other poles being outside.
We finally recall that $x_n$ and $\bar x_n$ are given in terms of $x$ and of the $m$-tuples $\boldsymbol{\epsilon,\epsilon'}\in\{1,2\}^m$ as in \eqref{def-x_n}, whereas $x$ is itself given in terms of $y_R$ \eqref{Ref-SoV-R} by \eqref{cond-x-y-bis}.
%\end{theorem}

\begin{rem}
Note that for $\tilde{m}_{\boldsymbol{\epsilon,\epsilon'}}=0$, these correlation functions are special cases of those computed in \cite{NicT23}. 
\end{rem}

\begin{rem}\label{rem-config-bis}
A similar result can be obtained in the same way for another configuration of unparallel boundary fields which shares the same simple properties than the special case considered here, namely the fact that correlations functions can be computed in an ungauged basis and that the spectrum is isospectral to the diagonal case. This other configuration is given by boundary fields of the form
\begin{equation}
\frac{\sinh \eta }{\sinh \varsigma_+}\cosh \varsigma_+\, \sigma_{1}^{z}
+\frac{\sinh \eta }{\sinh \varsigma_-}\big[\cosh \varsigma_-\, \sigma_{N}^{z}+e^{rt_-}\,(\sigma _{N}^{x}+ir\sigma _{N}^{y})\big],
\end{equation}
see footnote~\ref{foot-config-bis}. In that case, the vanishing of the scalar products for $\tilde{m}_{\boldsymbol{\epsilon,\epsilon'}}\not=0$ can be shown in the same way as here (with the same degree of conjecture).
The only difference in the above result is that, for $\tilde{m}_{\boldsymbol{\epsilon,\epsilon'}}=0$, the coefficient %$H_{\{\mathsc{b}_{j}\}}$ \eqref{H-finite} and 
$H_m$ \eqref{Hm} contains an extra factor: 
\begin{equation}\label{extra-factor}
  \prod\limits_{j=1}^{m}\frac{\sinh (\xi_{j}-\varepsilon \varsigma_+)}{\sinh (\lambda_j-\varepsilon \varsigma_++\eta /2)},
\end{equation}
and that, if the set of Bethe roots for the ground state contains the boundary root $\varepsilon\varsigma_+-\eta /2$,  the contour $\mathcal{C}$ in \eqref{el-block-thermo} is modified as
\begin{equation}
\mathcal{C}=[-\Lambda ,\Lambda ]\cup \Gamma^{-} (\varepsilon\varsigma_+-\eta /2)  \label{C-BR}
\end{equation}
to take into account the contribution of this boundary root in the final result.
\end{rem}

%%%%%%%%%%%%%%%%%%%%%%%%%%%
\section{Boundary overlaps}
\label{sec-overlap}

In this last section, we explain how to compute overlaps between eigenstates before and after a boundary quantum quench consisting in changing the boundary magnetic field at site 1 from the special invariant configuration considered previously to a more general one, the boundary magnetic field at site $N$ remaining  unchanged\footnote{Note that we do not necessarily need Nepomechie's constraint for the Hamiltonian $H_{h_+,h_-}$ to be able to compute the overlaps from an algebraic viewpoint. However, in the completely general case, the description of the spectrum and eigenstates involves an inhomogeneous $TQ$-equation which makes the thermodynamic limit much more complicated to study, see section~\ref{sec-spectrum-gen}.}:
\begin{equation}\label{quench}
 H_{h_+^\mathrm{(Inv)},h_-} \longrightarrow H_{h_+,h_-},
 \qquad\text{or conversely}\quad 
 H_{h_+,h_-}  \longrightarrow H_{h_+^\mathrm{(Inv)},h_-},
\end{equation}
in which $H_{h_+,h_-}$ denotes the Hamiltonian \eqref{Ham} with general boundary fields, and $h_+^\mathrm{(Inv)}$ is given by \eqref{h+inv}. In other words, we explain how to generalise, to this special case of unparallel boundary fields, the results obtained in \cite{AbeKT25} in the simplest case in which both fields remain parallel to the $z$-anisotropy direction.

The computation of overlaps is a crucial step for the study of the dynamic of a system after a quantum quench \cite{BroDWC14,IllNWCEP15,PirPV17}. In general, this is a very difficult problem, even for quantum integrable systems  in which one has access to an  exact and manageable description of the eigenstates. %many physical quantities are {\em a priori} exactly computable. 
The reason is that, usually, the construction of the eigenstates before and after the quench involves different algebras. %In fact, overlaps were so far computed only in very specific situations, see e.g. .

As shown in \cite{AbeKT25}, the XXZ spin chain with open boundary conditions and longitudinal boundary fields (i.e. diagonal K-matrices) provides an interesting example in which the computation of such overlaps after a local boundary quench is possible: in the case of diagonal boundary conditions, indeed,  eigenstates of the model are expressed as Bethe states of the form\footnote{An alternative description of the eigenstates in terms of the elements of the boundary monodromy matrix $\mathcal{U}_+$ constructed from $K_+$ is also possible. In the latter case, the algebraic form of the Bethe states with diagonal boundary conditions remains invariant under a change of $h_-$.}
\begin{equation}\label{Bethe-diag}
    \prod_{j=1}^M\mathcal{B}_-(\lambda_j)\,\ket{0},\qquad \bra{0}\prod_{j=1}^M\mathcal{C}_-(\lambda_j),
\end{equation}
in which  $\mathcal{B}_-$ and $\mathcal{C}_-$ are the off-diagonal elements of the monodromy matrix $\mathcal{U}_-$ \eqref{def-U-}, the construction of which does not involve the boundary parameters of $K_+$, i.e. the field $h_+$, but only  those of $K_-$, i.e. the field $h_-$.  Hence, a change of $h_+$ does not modify the algebraic form \eqref{Bethe-diag} of the eigenstates of the model, it only influences the set of Bethe roots $\{\lambda_1,\ldots,\lambda_M\}$ which solves different systems of Bethe equations before and after the quench. The computation of the overlaps between the eigenstates before and after the boundary quench is then possible by means of the generalisation \cite{Wan02,KitKMNST07} of Slavnov's determinant representation \cite{Sla89} for the scalar products of Bethe states, see \cite{AbeKT25} for more details.

The situation is however more complicated in the case of more general non-diagonal boundary conditions, when at least one of the two considered configurations involves a general matrix $K_+$.
%, or at least as soon as one of the two considered configurations involves a non-longitudinal field $h_+$ corresponding to a non-diagonal matrix $K_+$. 
In the latter case, the construction of the eigenstates of the model relies on the Vertex-IRF transformation \eqref{mat-S} to diagonalise the boundary matrix $K_+$, see section~\ref{sec-spectrum-gen}. The eigenstates of the model are then constructed from the transformed gauge  boundary monodromy matrix \eqref{gauged-U}, either as separate states or as generalised Bethe states. The crucial point is that both constructions involve the gauge parameters $\alpha$ and $\beta$ which, in the case of a general $K_+$ matrix, depend on the boundary parameters of $K_+$, and hence on $h_+$, through the relations \eqref{cond-diff-gauge}-\eqref{cond-sum-gauge}. Hence, the computation of overlaps after such a more general boundary quench would involve the computation of scalar products of separate states expressed on left and right SoV bases $\bra{\alpha-\beta+1}$ and $\ket{\alpha'-\beta'-1}$ for different values of $\alpha-\beta$ and $\alpha'-\beta'$. As already mentioned, this is still an open question.

The situation is nevertheless much simpler if one of the two configuration involves the invariant special boundary condition $h_+^\mathrm{(Inv)}$ \eqref{h+inv} at site 1, i.e. the gauge invariant boundary matrix \eqref{gauge-inv-K}. We know from section~\ref{sec-SoVbases} that the transfer matrix $ \mathcal{T}^\text{(Inv)}(\lambda)$ \eqref{Tinv}, i.e. the Hamiltonian $H_{h_+^\mathrm{(Inv)},h_-} $ for any boundary $h_-$, can be diagonalised in the whole family of SoV bases \eqref{ket-SoV}-\eqref{ket-Tinv} and \eqref{bra-SoV}-\eqref{bra-Tinv} for arbitrary values of $\alpha-\beta$, including in the ungauged limit $\eta(\alpha-\beta)\to +\infty$ \eqref{ungauge-ket}-\eqref{ungauge-bra}. It means that we can choose, to diagonalise $ \mathcal{T}^\text{(Inv)}(\lambda)$, the SoV basis which also diagonalises the transfer matrix $\mathcal{T}(\lambda)$ corresponding to the Hamiltonian $H_{h_+,h_-} $ with the new general boundary field $h_+$ after the quench. The computation of the corresponding overlaps then reduces to the computation of the scalar products of separate states expressed on left and right SoV bases with the same value of $\alpha-\beta$, the latter being fixed by \eqref{cond-diff-gauge} in terms of the parameters of the new general (non-diagonal) $K_+$ matrix. We explain in this section how to compute these boundary overlaps.

%Note that it is also {\em a priori} possible to compute more general overlaps involving the matrix elements of the complete set of local operators \eqref{Local-op}. We shall conclude this section by commenting about this.

%%%
\subsection{SoV spectrum and eigenstates of the generic open transfer matrix}
\label{sec-spectrum-gen}

Let us first briefly recall the SoV description of the spectrum and eigenstates of the transfer matrix $\mathcal{T}(\lambda )$ associated to two general boundary K-matrices $K_{\mp }(\lambda |\tau_{\mp },\varphi _{\mp },\psi _{\mp })$ \eqref{mat-K}-\eqref{def-Kpm}, see \cite{Nic12,KitMN14,KitMNT18,NicT23} for more details.

%%%%%%%%%
%\subsubsection{Expression of the generic transfer matrix in terms of the gauge operators}
\subsubsection{SoV basis}

In the case of an XXZ spin chain with generic boundary conditions, and in particular with a generic boundary field $h_+$,  the Vertex-IRF  transformation of section~\ref{sec-gauge} can be used to pseudo-diagonalise the non-diagonal matrix $K_+$. The idea is, as previously, to write the transfer matrix $\mathcal{T}(\lambda)$ in terms of $\mathcal{A}_{-}(\pm\lambda |\alpha ,\beta-1 )$ only, or alternatively in terms of $\mathcal{D}_{-}(\pm\lambda |\alpha ,\beta+1 )$ only. This can be done in the generic case by adequately fixing the gauge parameters $\alpha$ and $\beta$, see \cite{KitMNT18,NicT22}.

More precisely, rewriting $\mathcal{T}(\lambda)$ as
\begin{align}
    \mathcal{T}(\lambda) 
    = \tr\left[ \widehat{\mathcal{K}}_{+}(\lambda |\alpha ,\beta )\ \widehat{\mathcal{U}}_{-}(\lambda
|\alpha ,\beta )\right]
\end{align}
with
\begin{align}\label{hatK+}
    \widehat{K}_+(\lambda |\alpha,\beta)
    &= S^{-1}(\lambda -\eta /2|\alpha -1,\beta )\ K_{+}(\lambda )\ S(\eta /2-\lambda |\alpha +1,\beta ) 
    \nonumber\\
    &= S^{-1}(\lambda +\eta /2|\alpha ,\beta )\ K_{+}(\lambda )\ S(-\lambda-\eta /2 |\alpha ,\beta ) ,
\end{align}
and $\widehat{\mathcal{U}}_{-}(\lambda |\alpha ,\beta ) $ as in \eqref{hat-U-}, it was shown in \cite{KitMNT18,NicT22} that, under the following choice of the gauge parameters $\alpha$ and $\beta$ diagonalising \eqref{hatK+}, %{\bf (check the signs)}
\begin{align}
   &\eta(\alpha-\beta)=-\tau_+-\epsilon_+(\varphi_+-\psi_+)-\frac{1+\epsilon_+}2i\pi \mod 2i\pi 
   \qquad \text{for } \epsilon_+\in\{1,-1\}, \label{cond-diff-gauge}\\
%   &\sinh(\eta(\alpha-\beta)+\tau_+)=\sinh(\varphi_+-\psi_+), \label{cond-diff-gauge}\\
   &\eta(\alpha+\beta)=-\tau_++\epsilon'_+(\varphi_+-\psi_+)+\frac{1-\epsilon'_+}2i\pi \mod 2i\pi 
   \qquad \text{for } \epsilon'_+\in\{1,-1\}, \label{cond-sum-gauge}
%   &\sinh(\eta(\alpha+\beta)+\tau_+)=\sinh(\varphi_+-\psi_+), \label{cond-sum-gauge}
\end{align}
%
%i.e.
%
%\begin{align}
%& \eta \alpha =-\tau_++\frac{\epsilon_+^{\prime }-\epsilon_+}{2}(\varphi_+-\psi_+)-\frac{\epsilon_++\epsilon_+^{\prime }}{4}i\pi +ik\pi\mod 2i\pi ,  \label{Gauge-cond-A} \\
%& \eta \beta =\frac{\epsilon_++\epsilon_+^{\prime }}{2}(\varphi_+-\psi_+)+\frac{2+\epsilon_+-\epsilon_+^{\prime }}{4}i\pi +ik\pi \mod 2i\pi ,  \label{Gauge-cond-B}
%\end{align}
%
%for $\epsilon_+,\epsilon_+^{\prime }\in \{1,-1\}$, $k\in\{0,1\}$,
the transfer matrix can be expressed as %\note{corrected coeff}
\begin{align}
\mathcal{T}(\lambda )
& =\mathsf{\bar a}_{+}(\lambda | \epsilon_+)\frac{\sinh (2\lambda +\eta )}{\sinh 2\lambda }\,
      \mathcal{A}_{-}(\lambda |\alpha ,\beta -1)
     +\mathsf{\bar a}_{+}(-\lambda |\epsilon_+ )\frac{\sinh (2\lambda -\eta )}{\sinh 2\lambda }\,
     \mathcal{A}_{-}(-\lambda |\alpha ,\beta -1),  \label{Gauge-T-A} \\
& =\mathsf{\bar d}_{+}(\lambda|\epsilon_+)\frac{\sinh (2\lambda +\eta )}{\sinh 2\lambda }\,
     \mathcal{D}_{-}(\lambda |\alpha ,\beta +1)
   +\mathsf{\bar d}_{+}(-\lambda |\epsilon_+)\frac{\sinh (2\lambda -\eta )}{\sinh 2\lambda }\,
   \mathcal{D}_{-}(-\lambda |\alpha,\beta +1),  \label{Gauge-T-D}
\end{align}
where 
\begin{align}
\mathsf{\bar a}_{+}(\lambda |\epsilon_+)& =\epsilon_+\,e^{-\lambda +\frac{\eta }{2}}\,\frac{\sinh (\lambda -\frac{\eta }{2}+\epsilon_+\varphi_+)\,\cosh(\lambda -\frac{\eta }{2}-\epsilon_+\psi_+)}{\sinh \varphi_+\,\cosh\psi_+}, \\
\mathsf{\bar d}_{+}(\lambda|\epsilon_+ )& =-\epsilon_+\,e^{-\lambda +\frac{\eta }{2}}\,\frac{\sinh (\lambda -\frac{\eta }{2}-\epsilon_+\varphi_+)\,\cosh(\lambda -\frac{\eta }{2}+\epsilon_+\psi_+)}{\sinh \varphi_+\,\cosh\psi_+}.
\end{align}
In the following, we shall fix $\epsilon_+=\epsilon'_+$.

%%%%%%%%%%%

As shown in \cite{KitMNT18,NicT22}, it follows from the expression \eqref{Gauge-T-A}-\eqref{Gauge-T-D} in terms of the gauge operators $\mathcal{A}_{-}(\pm\lambda |\alpha ,\beta-1 )$ or $\mathcal{D}_{-}(\pm\lambda |\alpha ,\beta+1 )$ {\em for $\alpha$ and $\beta$ fixed by \eqref{cond-diff-gauge} } that, under the condition \eqref{cond-inh} and provided that the normalisation coefficient in \eqref{orth-SoV-Sk} does not vanish, the states 
\begin{equation}
  \ket{\mathbf{h},\alpha,\beta+1}_\text{Sk}\qquad \text{and} \qquad {}_\text{Sk}\bra{\alpha,\beta-1,\mathbf{h}} ,
\end{equation}
given by the expressions \eqref{ket-SoV} and \eqref{bra-SoV}, form an SoV basis in $\mathcal{H}$, respectively in $\mathcal{H}^*$, for the generic transfer matrix $\mathcal{T}(\lambda )$.
We recall that, from the study of Section~\ref{sec-diff-bases}, these SoV states depend on the specific value of $\alpha-\beta$ (here fixed by \eqref{cond-diff-gauge}), but not on the value of $\alpha+\beta$, as shown in Proposition~\ref{prop-id-gauge} by the identification \eqref{gauged-SoV-identity}. Therefore, we shall rather denote them in the following as
\begin{equation}
  \ket{\mathbf{h},\alpha-\beta-1} \qquad \text{and} \qquad \bra{\alpha-\beta+1,\mathbf{h}} .
\end{equation}

It is important to underline again here that, when we vary the boundary conditions, we also as a consequence vary the value of $\alpha-\beta$ since it is fixed by  \eqref{cond-diff-gauge}. Therefore, two spin chains with different boundary fields at site 1 are in general solved in different SoV basis (or, in other words, the corresponding generalised Bethe states involve different gauge algebras, contrary to what happens in the diagonal case, see \cite{AbeKT25}), which makes the computation of overlaps between their eigenstates {\em a priori} not so simple.

However, to any given spin chain with generic non-longitudinal boundary fields \eqref{Ham} which can be solved in generalised Sklyanin's SoV framework, one can associated the spin chain with the same boundary field at site $N$ and the special boundary condition \eqref{gauge-inv-K}-\eqref{h+inv} at site 1. Then, among the many possible SoV bases for the latter model, see Section~\ref{sec-SoVbases}, one can choose the one satisfying the condition \eqref{cond-diff-gauge}, so that in that case eigenstates of the two models are expressed in the same SoV bases, and overlaps can easily be computed.

%%%%%%%%%%%%
\subsubsection{Spectrum and eigenstates}

Under the assumption that the two boundary matrices are not both proportional to the identity matrix and that the inhomogeneity parameters are generic (or at least satisfy \eqref{cond-inh}), the transfer matrix $\mathcal{T}(\lambda )$ is diagonalisable with simple spectrum. 

To characterise its spectrum and eigenstates, let us define the coefficients, for given values of $\varepsilon,\epsilon_+\in\{+,-\}$, 
\begin{equation}\label{full-A}
\underline{\mathbf{A}}_{\varepsilon,\epsilon_+ }(\lambda )=(-1)^{N}\,\frac{\sinh (2\lambda+\eta )}{\sinh (2\lambda )}\,\underline{\mathbf{a}}_{\varepsilon,\epsilon_+ }(\lambda)\,a(\lambda )\,d(-\lambda ),  
\end{equation}
in which
\begin{align}\label{full-a}
\underline{\mathbf{a}}_{\varepsilon,\epsilon_+ }(\lambda )
&=\frac{\sinh (\lambda -\frac{\eta }{2}+\epsilon_+ \varphi _{+})\,\cosh (\lambda -\frac{\eta }{2}-\epsilon_+\psi _{+})\,\sinh (\lambda -\frac{\eta }{2}+\varepsilon \varphi _{-})\,\cosh(\lambda -\frac{\eta }{2}+\varepsilon\psi _{-})}{\sinh (\epsilon_+ \varphi _{+})\,\cosh(\psi _{+})\sinh (\varepsilon \varphi _{-})\,\cosh (\psi _{-})}.
\end{align}
We also introduce the following separate states:
\begin{align}
  \ket{Q,\alpha-\beta-1} 
  &\equiv  \ket{Q,\alpha-\beta-1} _\varepsilon \nonumber\\
  &= \sum_{\mathbf{h}\in\{0,1\}^{N}}\prod_{n=1}^{N} Q(\xi _{n}^{(h_{n})})\
e^{-\sum_{j}h_{j}\xi _{j}}\,\widehat{V}(\xi _{1}^{(h_{1})},\ldots ,\xi_{N}^{(h_{N})})\ \ket{\mathbf{h},\alpha-\beta-1}_\varepsilon ,  
      \label{r-sep-gen} 
\end{align}
defined on the SoV basis \eqref{ket-SoV}-\eqref{ket-Tinv} with $\mathsf{g}_-=\mathsf{g}_-^{(\varepsilon)}$ as in \eqref{geps},  
and
\begin{align}
\bra{\underline{\alpha-\beta+1,  Q}} &\equiv   {}_\varepsilon\bra{\underline{\alpha-\beta+1,  Q}} 
\equiv  {}_\varepsilon\bra{\alpha-\beta+1, \underline{\mathbf{A}}_{\varepsilon,\epsilon_+ }, Q} \nonumber\\
 &= \sum_{\mathbf{h}\in\{0,1\}^{N}}\prod_{n=1}^{N}\left[ \left(\frac{\sinh (2\xi _{n}-2\eta )}{\sinh(2\xi _{n}+2\eta )}\,\frac{\underline{\mathbf{A}}_{\varepsilon,\epsilon_+ }(\xi _{n}+\frac{\eta }{2})}{\underline{\mathbf{A}}_{\varepsilon,\epsilon_+ }(-\xi _{n}+\frac{\eta }{2})} \right)^{\! h_{n}}\ Q(\xi _{n}^{(h_{n})})\right]  \notag \\
& \hspace{3cm} \times e^{-\sum_{j}h_{j}\xi _{j}}\,\widehat{V}(\xi _{1}^{(h_{1})},\ldots,\xi _{N}^{(h_{N})})\ {}_\varepsilon\bra{\alpha-\beta+1, \mathbf{h}},
\nonumber\\
&= \frac{\widehat{V}(\xi _{1}^{(0)},\ldots,\xi _{N}^{(0)})}{\widehat{V}(\xi _{1}^{(1)},\ldots,\xi _{N}^{(1)})} \sum_{\mathbf{h}\in\{0,1\}^{N}}\prod_{n=1}^{N}\left[ \left(-\frac{\underline{\mathbf{a}}_{\varepsilon,\epsilon_+ }(\xi _{n}+\frac{\eta }{2})}{\underline{\mathbf{a}}_{\varepsilon,\epsilon_+ }(-\xi _{n}+\frac{\eta }{2})} \right)^{\! h_{n}}\ Q(\xi _{n}^{(h_{n})})\right]  \notag \\
& \hspace{3cm} \times e^{-\sum_{j}h_{j}\xi _{j}}\,\widehat{V}(\xi _{1}^{(1-h_{1})},\ldots,\xi _{N}^{(1-h_{N})})\ {}_\varepsilon\bra{\alpha-\beta+1, \mathbf{h}},
\label{l-sep-gen}
\end{align}
defined on the SoV basis \eqref{bra-SoV}-\eqref{bra-Tinv} with $\mathsf{g}_-=\mathsf{g}_-^{(\varepsilon)}$ as in \eqref{geps}.
Here the value of $\alpha-\beta$ is fixed in terms of the boundary parameters of $K_+$ and of $\epsilon_+$ by \eqref{cond-diff-gauge}.
Note that \eqref{r-sep-gen} coincides, for this fixed value of $\alpha-\beta$, with the definition \eqref{r-separate}, whereas \eqref{l-sep-gen} differs from \eqref{l-separate} in that it involves the coefficient $\underline{\mathbf{A}}_{\varepsilon,\epsilon_+ }$ \eqref{full-A} defined in terms of \eqref{full-a} instead of $\mathbf{A}^\mathrm{\! (Inv)}_\varepsilon$ \eqref{Aeps}.
Hence, if $Q\in\Sigma_Q^M$ with roots $\lambda_1,\ldots,\lambda_M$, \eqref{r-sep-gen} can be rewritten as
\begin{align}
   \ket{Q,\alpha-\beta-1}_\varepsilon 
   &=\mathsf{c}_{Q,\alpha-\beta-1}^{(R)}\ \underline{\widehat{\mathcal{B}}}_{-,M}(\{\lambda_{i}\}_{i=1}^{M}|\alpha -\beta +1)\,\ket{\Omega _{\alpha -\beta -1+2M} }_\varepsilon \nonumber\\
   &=\mathsf{c}_{Q,\alpha-\beta-1}^{(R)} \, \mathsf{d}_{\alpha-\beta-1+2M}^{(R,\varepsilon)}\  \underline{\widehat{\mathcal{B}}}_{-,M}(\{\lambda_{i}\}_{i=1}^{M}|\alpha-\beta+1)\,\ket{\eta, y_R^{(\varepsilon)} } ,
\end{align}
in which $\ket{\Omega _{\alpha -\beta -1+2M} }=\ket{\Omega _{\alpha -\beta -1+2M} }_\varepsilon$ is given by \eqref{R-ref-state}, and $y_R=y_R^{(\varepsilon)}$ is defined in terms of $\varepsilon$ and the $K_{-}$ boundary parameters as in \eqref{Ref-SoV-R}. Instead, \eqref{l-sep-gen} can be rewritten as
\begin{align}
    \bra{\underline{\alpha-\beta+1,  Q}} 
    = \mathsf{c}_{Q,\alpha-\beta+1}^{(L)} \ \bra{\underline{ \Omega _{\alpha -\beta +1-2M} } }\,\underline{\widehat{\mathcal{B}}}_{-,M}(\{\lambda _{i}\}_{i=1}^{M}|\alpha -\beta +1-2M),
\end{align}
in which 
\begin{multline}
   %\bra{\underline{\Omega}_{x}} 
   \bra{\underline{ \Omega _{\alpha -\beta +1-2M} } }
   =\frac{1}{\mathsf{N}(\boldsymbol{\xi}| \alpha -\beta-2M)}
   \sum_{\mathbf{h}\in\{0,1\}^N}\prod_{n=1}^{N}
   \left(\frac{\sinh (2\xi _{n}-2\eta )}{\sinh(2\xi _{n}+2\eta )}\,\frac{\underline{\mathbf{A}}_{\varepsilon,\epsilon_+ }(\xi _{n}+\frac{\eta }{2})}{\underline{\mathbf{A}}_{\varepsilon,\epsilon_+ }(-\xi _{n}+\frac{\eta }{2})} \right)^{\! h_{n}}
   %\left[\frac{\sinh (2\xi _{n}-\eta )}{\sinh (2\xi_{n}+\eta )}\,\frac{\mathsf{A}_{-}^{\left( \varepsilon \right) }(\xi _{n}+\frac{\eta }{2}) }{\mathsf{A}_{-}^{\left( \varepsilon \right) }(-\xi _{n}+\frac{\eta }{2})}\right]^{h_{n}}\,
   \\
   \times
   e^{-\sum_j h_j \xi_j }\, \widehat{V}(\xi_{1}^{(h_{1})},\ldots ,\xi _{N}^{(h_{N})})\ \bra{\alpha -\beta +1-2M,\mathbf{h}}.
    \label{L-ref-state-gen}
\end{multline}
%
%{\bf I don't understand how this state can still be equal to $\bra{ y_L^{(\varepsilon)} ,\eta }$, and hence to \eqref{L-ref-state}, as the proof of \cite{NicT22} seems to indicate ???!!! Something to clarify here.}

Let us also define the following combinations of the boundary parameters:
\begin{equation}\label{funct-f-inh}
\mathfrak{f}_{\varepsilon,\epsilon_+}^{(n)}=\frac{2\kappa _{+}\kappa_{-}\big[ \cosh (\tau _{+}-\tau _{-})-\cosh (\epsilon_+(\varphi _{+}-\psi_{+})+\varepsilon(\varphi _{-}+\psi _{-})+(N-1-2n)\eta )\big] }{\sinh \varsigma _{+}\,\sinh\varsigma _{-}}.
\end{equation}
Then, the set $\Sigma _{\mathcal{T}}$ of the transfer matrix eigenvalues is given by the set of entire functions $\tau(\lambda )$ such that there exists  $Q\in \Sigma _{Q}$ satisfying with $\tau(\lambda )$ the following functional $TQ$-equation: 
\begin{multline}
   \tau(\lambda )\, Q(\lambda ) 
   =\underline{\mathbf{A}}_{\varepsilon,\epsilon_+ }(\lambda )\,Q(\lambda -\eta )+\underline{\mathbf{A}}_{\varepsilon,\epsilon_+ }(-\lambda )\,Q(\lambda +\eta ) 
    \\
 +\mathfrak{f}_{\varepsilon,\epsilon_+ }^{(N)}\,a(\lambda )\,a(-\lambda )\,d(\lambda)\,d(-\lambda )\,[\cosh ^{2}(2\lambda )-\cosh ^{2}\eta ],
\label{TQ-inh}
\end{multline}
and the unique (up to an overall normalisation factor) $\mathcal{T}(\lambda )$ eigenvector and eigencovector with eigenvalue $\tau(\lambda) $ are given by \eqref{r-sep-gen} and \eqref{l-sep-gen} respectively. 

To properly fix the degree of the polynomial $Q(\lambda )$ we should distinguish the following cases:

\begin{enumerate}
\item\label{case1} If 
\begin{align}
&\text{either }\mathfrak{f}_{\varepsilon,\epsilon_+ }^{(N)} \not=0  \label{FN-nonzero} \\
&\text{or }\mathfrak{f}_{\varepsilon ,\epsilon_+}^{(N)} =0\text{ with }\kappa _{+}\kappa_{-}\not=0,  \label{FN-Zero-1}
\end{align}
then the complete transfer matrix spectrum is constructed by polynomials $Q(\lambda )\in \Sigma _{Q}^{N},$ satisfying with $\tau (\lambda )$ the inhomogeneous or homogeneous functional $TQ$-functional equation \eqref{TQ-inh}, respectively.

\item\label{case2} If 
\begin{equation}
\kappa _{+}\kappa _{-}=0,
\end{equation}
with the other parameters remaining finite, then the complete transfer matrix spectrum is constructed by polynomial $Q(\lambda )\in\Sigma_Q=\cup _{n=0}^{N}\Sigma _{Q}^{n}$ satisfying with $\tau (\lambda )$ the homogeneous functional $TQ$-equation 
\begin{equation}
 \tau(\lambda )\, Q(\lambda )=\underline{\mathbf{A}}_{\varepsilon,\epsilon_+ }(\lambda )\,Q(\lambda -\eta )+\underline{\mathbf{A}}_{\varepsilon,\epsilon_+ }(-\lambda )\,Q(\lambda +\eta ).
\label{TQ-hom-gen}
\end{equation}
\end{enumerate}

An interesting particular case to study is when the boundary parameters satisfy Nepomechie's constraint \cite{Nep04,NepR03,NepR04}: under the condition \eqref{FN-nonzero}, one can still have
\begin{equation}
\mathfrak{f}_{\varepsilon,\epsilon_+ }^{(M)}=0,  \label{FM-zero}
\end{equation}
i.e.
\begin{equation}\label{constraint}
   \cosh (\tau _{+}-\tau _{-})-\cosh (\epsilon_+(\varphi _{+}-\psi_{+})+\varepsilon(\varphi _{-}+\psi _{-})+ (N-1-2M)\eta)=0
\end{equation}
for some given $M\in \{0,\ldots ,N-1\}$ and a given choice of $\varepsilon,\epsilon_+ \in\{+,-\}$. Then, whereas the complete spectrum can still be described by solutions $Q\in\Sigma_Q^N$ of the inhomogeneous $TQ$-equation \eqref{TQ-inh}, part of the spectrum  can alternatively be described by solutions $Q\in \Sigma _{Q}^{M}$ of the homogeneous $TQ$-equation \eqref{TQ-hom-gen}. Moreover, if \eqref{constraint} is satisfied for  given values of $\varepsilon,\epsilon_+$ and a given value of $M\in \{0,\ldots ,N-1\}$, it is also satisfied for $-\varepsilon,-\epsilon_+$ and $M'=N-1-M$, and it was conjectured in \cite{NepR04} that the union of these two sectors provides the complete spectrum.

%%%%%%%%%%%%
\subsection{Determinant representations for overlaps after a boundary quench}

Let us now consider the boundary quantum quench \eqref{quench}, in which the system is prepared in an eigenstate of $\mathcal{T}^\mathrm{(Inv)}$ (typically the ground state of $H_{h_+^\mathrm{(Inv)},h_-}$) and evolves at $t>0$ with the Hamiltonian $H_{h_+,h_-}$ (or conversely). We recall that the two Hamiltonians share the same boundary field $h_-$ at site $N$, but that the boundary field at site 1 is changed at $t=0$ from $h_+^\mathrm{(Inv)}$ \eqref{h+inv} to $h_+$, in which $h_+$ is {\em a priori} generic (or conversely). We therefore want to calculate the following overlaps:
\begin{equation}\label{overlaps}
   \moy{\alpha-\beta+1,Q\, |\, P,\alpha-\beta-1},\qquad \text{or}\qquad
   \moy{\underline{\alpha-\beta+1,P}\, |\, Q,\alpha-\beta-1},
\end{equation}
in which %(for given signs  $\varepsilon$ and $\epsilon_+$) 
$Q\in\Sigma_Q$ with roost $\lambda_1,\ldots,\lambda_{M_Q}$ is a solution of the $TQ$-equation \eqref{TQ-hom} with an eigenvalue $\tau_Q$ of $\mathcal{T}^\mathrm{(Inv)}$, whereas $P\in\Sigma_Q$ with roots $\mu_1,\ldots,\mu_{M_P}$ is a solution of the $TQ$-equation \eqref{TQ-inh} or \eqref{TQ-hom-gen} with an eigenvalue $\tau_P$ of $\mathcal{T}$, according to whether we are in case~\ref{case1} or in case~\ref{case2}. In other words, $\bra{\alpha-\beta+1,Q}$ and $\ket{Q,\alpha-\beta-1}$ are left and right eigenstates of $\mathcal{T}^\mathrm{(Inv)}$, whereas $\bra{\underline{\alpha-\beta+1,P}}$ and $\ket{P,\alpha-\beta-1}$ are left and right eigenstates of $\mathcal{T}$. We recall that the value of $\alpha-\beta$ is here fixed in terms of the boundary parameters $\tau_+$, $\varphi_+$ and $\psi_+$ of $\mathcal{T}$ by the condition \eqref{cond-diff-gauge}.
%Moreover, the value of the signs $\varepsilon$ and $\epsilon_+$ is also fixed.

The finite-size expression of these overlaps is then a direct corollary of the results of \cite{KitMNT18}. More precisely, the expression for the overlap $\moy{\alpha-\beta+1,Q\, |\, P,\alpha-\beta-1}$ is a direct corollary of Theorem~5.2 and 5.3 of \cite{KitMNT18} in which $\ket{P,\alpha-\beta-1}$ is considered as an arbitrary separate state. The result can therefore be formulated completely similarly as Proposition~\ref{prop-SP-finite}:
\begin{equation}\label{overlap1}
  \frac{\moy{\alpha-\beta+1,Q\, |\, P,\alpha-\beta-1}}{\moy{\alpha-\beta+1,Q\, |\, Q,\alpha-\beta-1}}=0 \qquad \text{if}\quad M_P<M_Q,
\end{equation}
and if $M_P\ge M_Q$ it is given by the following determinant representation in terms of the generalised Slavnov determinant \eqref{Slavnov1}-\eqref{Slavnov2} for $M\equiv M_Q$, $M'\equiv M_P$ and $\tau\equiv\tau_Q$:
\begin{multline}\label{overlap2}
    \frac{\moy{\alpha-\beta+1,Q\, |\, P,\alpha-\beta-1}}{\moy{\alpha-\beta+1,Q\, |\, Q,\alpha-\beta-1}}
   =  (-1)^{N(M_P-M_Q)}\, \frac{\Gamma^{(M_P+M_Q)}_- }{\Gamma^{(2M_Q)}_- }\ 
   \frac{\prod_{i=1}^{M_Q}\sinh (2 \lambda_i+\eta )\, \sinh (2\lambda_i-\eta )}{\prod_{i=1}^{M_P}\sinh (2\mu_i+\eta )\, \sinh (2\mu_i+\eta )} 
   \\ 
   \times \frac{\widehat{V}(\lambda_{M_Q},\ldots, \lambda_1)}{\widehat{V}(\mu_{M_P},\ldots,\mu_1)}\,
   \frac{\det_{M_P} \left[\mathcal{S}_Q(\boldsymbol{\mu},\boldsymbol{\lambda}) \right]}
   {\det_{M_Q} \left[ \mathcal{S}_Q(\boldsymbol{\lambda},\boldsymbol{\lambda}) \right]},
\end{multline}
in which we have used the same notations as in Proposition~\ref{prop-SP-finite}. Note that we have, for simplicity, written here the renormalised version of these overlaps so as to avoid dealing with complicated normalisation factors, but that all these normalisation factors have explicitly been computed and can be found in \cite{KitMNT18}.

The expression for the overlap $\moy{\underline{\alpha-\beta+1,P}\, |\, Q,\alpha-\beta-1}$, in which we consider $\ket{Q,\alpha-\beta-1}$ as an arbitrary state, can also be obtained from the results of \cite{KitMNT18}. In the completely generic case in which $P$ is solution of the inhomogeneous $TQ$-equation \eqref{TQ-inh} (case~\ref{case1} with \eqref{FN-nonzero}), the overlap can still be written in terms of a generalised Slavnov determinant but with a more complicated expression, see Appendix~E.2 of \cite{KitMNT18}. This expression simplifies in the situations in which $P$ is solution of an homogeneous $TQ$-equation, i.e. under Nepomechie's constraint $\mathfrak{f}_{\varepsilon,\epsilon_+ }^{(M_P)} =0$, with $M_P\le N$. In that case, if we suppose moreover that $M_Q\ge M_P$, we have, from Theorem~5.2 and Theorem~5.3 of \cite{KitMNT18}, %{\bf (to check)}
\begin{multline}\label{overlap3}
    \frac{\moy{\underline{\alpha-\beta+1,P}\, |\, Q,\alpha-\beta-1}}{\moy{\underline{\alpha-\beta+1,P}\, |\, P,\alpha-\beta-1}}
   =  (-1)^{N(M_P-M_Q)}\, \frac{\Gamma^{(M_P+M_Q)}_{\{a\} }}{\Gamma^{(2M_P)}_{\{a\} }}\ 
  \frac{\prod_{i=1}^{M_P}\sinh (2 \mu_i+\eta )\, \sinh (2\mu_i-\eta )}{\prod_{i=1}^{M_Q}\sinh (2\lambda_i+\eta )\, \sinh (2\lambda_i+\eta )} 
   \\ 
   \times \frac{\widehat{V}(\mu_{M_P},\ldots, \mu_1)}{\widehat{V}(\lambda_{M_Q},\ldots,\lambda_1)}\,
   \frac{\det_{M_Q}\left[\mathcal{S}_P(\boldsymbol{\lambda},\boldsymbol{\mu}) +\mathcal{P}_P(\boldsymbol{\lambda},\boldsymbol{\mu})\right]}
   {\det_{M_P}\left[ \mathcal{S}_P(\boldsymbol{\mu},\boldsymbol{\mu}) \right]},
\end{multline}
where  the matrix $\mathcal{S}_{P} (\boldsymbol{\nu},\boldsymbol{\mu})$ is defined similarly as $\mathcal{S}_{Q} (\boldsymbol{\nu},\boldsymbol{\lambda})$ but with the role of $P$ and $Q$ interchanged, i.e.
\begin{alignat}{2}
   &\left[ \mathcal{S}_{P} (\boldsymbol{\nu},\boldsymbol{\mu}) \right]_{j,k}
   =P(\nu_j)\,\frac{\partial \tau_P(\nu_j)}{\partial \mu_k }
   \qquad\ & &\text{if}\quad k\leq M_P, \\
%   \nonumber\\
%   &= \sum_{\bar{\epsilon}=\pm }\bar{\epsilon}\,\mathbf{A}^{(\epsilon )}(\bar{\epsilon}\nu_j)\, Q(\nu_j-\bar\eps\eta)\big[ t(\nu_j+\lambda_k-\bar\eps\tfrac\eta 2)-t(\nu_j-\lambda_k-\bar\eps\tfrac\eta 2)\big],
   &\left[ \mathcal{S}_{P} (\boldsymbol{\nu},\boldsymbol{\mu}) \right]_{j,k}
   =\sum_{\sigma=\pm }\sigma\,\underline{\mathbf{A}}_{\varepsilon,\epsilon_+}(-\sigma\nu_j)\,\sinh (2\nu_j+\sigma\eta )\, P(\nu_{j}+\sigma\eta )\,& &\left( \frac{\cosh (2\nu_{j}+\sigma\eta)}{2}\right) ^{k-M_P-1}  \notag \\   
   &\hspace{4cm} & &\text{if}\quad k >M_P,
\end{alignat}
and where $\mathcal{P}_P(\boldsymbol{\lambda},\boldsymbol{\mu})$ is a rank one matrix which has only one non-zero column:

\begin{align}\label{rank1-P}
   \left[\mathcal{P}_P(\boldsymbol{\lambda},\boldsymbol{\mu})\right]_{j,k}
   &=\delta_{k,M_Q}\big[ g_{\{a\}}^{(M_P+M_Q)}(\lambda_j)\,
     \frac{\sinh(2\lambda_j+\eta)\sinh(2\lambda_j-\eta)}{P(\lambda_j)}
        \nonumber\\
  &\qquad
  -\sum_{\sigma=\pm 1}\sigma\,\underline{\mathbf{A}}_{\varepsilon,\epsilon_+}(-\sigma \lambda_j)\,
    \sinh(2\lambda_j+\sigma\eta)\, P(\lambda_j+\sigma\eta)
    \nonumber\\
   &\qquad
   \times
   \sum_{l=1}^{M_P} \frac{2\,g_{\{a\}}^{(M_P+M_Q)}(\mu_l)\, \sinh(2\mu_l-\eta)\, P(\lambda_j)}{  f_{\{a\}}(-\mu_l)\, P'(\mu_l)\, P(\mu_l-\eta)\,\big[\cosh(2 \lambda_j+\sigma\eta)-\cosh(2\mu_l-\eta)\big]}\big].
\end{align}
In this expression, we have set
\begin{equation}\label{def-set-a}
   \{a\}\equiv\{a_{1},a_{2},a_3,a_4\}=\{\eps\varphi_-,\eps(\psi_-+i\tfrac\pi 2),\epsilon_+\varphi_+,\epsilon_+(-\psi_++i\tfrac\pi 2)\},
\end{equation}
defined the function $f_{\{a\}}$ as
\begin{equation}
  f_{\{a\}}(\lambda)
    = (-1)^N\,
     \prod_{\ell=1}^{4}\frac{\sinh(\lambda-a_\ell+\frac{\eta}{2})}{\sinh(a_\ell)}\
     \frac{a(-\lambda)\, d(\lambda)}{\sinh(2\lambda)},
    \label{feps}
\end{equation}
and the functions $g^{(L)}_{\{a\}}$ as follows:
\begin{align}
  &g_{\{a\}}^{(N)}(\lambda)
  =  \frac{\sinh(\sum_\ell a_\ell-\eta)}{\prod_{\ell}\sinh(a_\ell)}\,
     a(\lambda)\, d(\lambda)\, a(-\lambda)\, d(-\lambda),
     \label{geps-N}\\
   &g_{\{a\}}^{(L)}(\lambda)
   = (-1)^{L-N}\, g_{\{a\}}^{(N)}(\lambda)
   - \bar{f}_{\{a\}}^{(L)}(\lambda)
   \qquad \text{if } L>N,
   \label{geps-L>N}
\end{align}
whereas, if $L<N$, the function $g_{\{a\}}^{(L)}(\lambda)$ is defined by induction as
\begin{multline}
    g^{(L)}_{\{a\}}(z)
    =
    \frac{\prod_{\ell}\sinh(a_\ell)\ \bar{f}^{(L)}_{\{a\}}(z)}{\sinh((L+1-N)\eta-\sum_\ell a_\ell)}
    \cdot 
  \lim_{z'\to\infty} 2^{N+L}
  \frac{\bar{f}^{(L+1)}_{\{a\}}(z')+g^{(L+1)}_{\{a\}}(z')}{\cosh(2z')^{N+L}}
       \\
    - \bar{f}^{(L)}_{\{a\}}(z)  
   - \bar{f}^{(L+1)}_{\{a\}}(z)
   -g^{(L+1)}_{\{a\}}(z),
   \label{rec-geps}
\end{multline}
in which we have used the shortcut notation:
\begin{equation}\label{feps-L}
\bar{f}^{(l)}_{\{a\}}(\lambda) 
     =
\sum_{\sigma\in\{+,-\}} 
     f_{\{a\}}(\sigma \lambda)
     \left(\frac{\cosh(2\lambda+\sigma\eta)}{2}\right)^{l-1}.
\end{equation}
The normalisation coefficients $\Gamma_{\{a\}}^{(x)}$ appearing in \eqref{overlap3} are also defined in terms of \eqref{def-set-a} as
\begin{equation}
   \Gamma_{\{a\}}^{(x)}=
\begin{cases}
{\prod\limits_{j=1}^{x-N}\frac{\prod_\ell\sinh (a_\ell)}{\sinh (j\eta-\sum_\ell a_\ell)}} 
           & \text{if }\ x\geq N,\vspace{2mm} \\ 
{\prod\limits_{j=0}^{N-x-1}\frac{\sinh (-j\eta -\sum_\ell a_\ell)}{\prod_\ell \sinh(a_\ell)}} 
           & \text{if }\ x<N.
\end{cases}
\end{equation}
%
%{\bf To write the expression here at least in the simplest case of the constraint.}

As a next step, it would be interesting to study the thermodynamic limit of these expressions, as in \cite{AbeKT25}. For the overlaps between the  ground-states of the two models, this would require to study the precise description of the ground state of the spin chain with general boundary condition. An additional technical difficulty, compared to what has been done in \cite{AbeKT25}, would be to understand how to deal with determinants involving extra lines, which appear here when considering states with different numbers of roots, and also how to deal with the rank one matrix \eqref{rank1-P} appearing in the expression \eqref{overlap3}.

%%%%%%%%%%%%%%%%
\section{Conclusion}

%We considered here the XXZ spin chain with arbitrary boundary field at site $N$ and a special boundary field at site 
In the framework of the XXZ spin chain with unparallel boundary fields, we focussed here on the special boundary condition \eqref{h+inv}  which has the following interesting property: the Vertex-IRF transform \eqref{hatK+inv} of the corresponding K-matrix \eqref{gauge-inv-K} is diagonal and invariant under a change of the gauge parameters $\alpha$ and $\beta$. As a result, the transfer matrix $\mathcal{T}^\text{(Inv)}(\lambda)$ \eqref{Tinv} of the spin chain with this gauge invariant boundary condition at site 1 and a generic boundary field at site $N$ is solvable, within Sklyanin's generalised SoV approach, in a whole family of SoV bases \eqref{ket-SoV}-\eqref{bra-SoV} 
%which pseudo-diagonalise the transformed gauge operator $\widehat{\mathcal{B}}_{-}(\lambda |\alpha -\beta ) $ \eqref{Bhat} 
for almost any values of the gauge parameters $\alpha$ and $\beta$. We explicitly showed that these SoV bases actually  depend only on the difference $\alpha-\beta$ of these gauge parameters (and not on their sum), see Proposition~\ref{prop-id-gauge}, the limiting case $\eta(\alpha-\beta)\to +\infty$ corresponding to an ungauged basis \eqref{ungauge-ket}-\eqref{ungauge-bra}. 

This particular model provides an example of a spin chain with unparallel boundary fields in which all elementary building blocks for correlation functions can be computed, see Section~\ref{sec-corr}. It also provides an interesting example of a starting or final configuration of a boundary quench protocol --- involving a fixed generic boundary field $h_-$ at site $N$ and a change of boundary field at site 1 from $h_+^\mathrm{(Inv)}$ to a non-longitudinal field $h_+$ (or conversely) ---  for which overlaps can be computed as determinants, see Section~\ref{sec-overlap}: this can be done by choosing, among the whole family of SoV bases enabling to diagonalise the Hamiltonian $H_{h_+^\mathrm{(Inv)},h_-}$, the one which also enables to diagonalise $H_{h_+,h_-}$, i.e. for which the value of $\alpha-\beta$ is fixed in terms of $h_+$ as in \eqref{cond-diff-gauge}.

As a next step, as mentioned at the end of Section~\ref{sec-overlap}, it would be interesting to study the thermodynamic limit of the obtained determinant representations \eqref{overlap2} and \eqref{overlap3} for these overlaps, following what has been done in the diagonal case in \cite{AbeKT25}. 
Another interesting perspective would be to further clarify the difference between all these SoV bases for different values of $\alpha-\beta$, and in particular to understand how one can compute scalar products of separate states built on two  such different SoV bases, i.e. with  different values of $\alpha-\beta$. This would open the way to the computation of all the elementary blocks for correlation functions that could not be computed in \cite{NicT22,NicT23} for the spin chain with general boundary fields, and also give access to overlaps of boundary quenches involving more general configurations of boundary fields.

%%%%%%%%%%%%%%%%%%%%%%%%%%%%%%%%%%%%%%%%%%%%%%%%%%%%%%%%%%%%%%%%%%%%%%%%%%%%%%%%%%%%%%%%%%%%%%%%%%%%%%%%%%
\section*{Acknowledgments}

G. N. is supported by CNRS and Laboratoire de Physique, ENS-Lyon. 
V. T. is supported by CNRS and LPTMS, Universit\'e Paris-Saclay.

\bibliographystyle{unsrt}

%\bibliographystyle{SciPost_bibstyle}
%\bibliography{/Users/vterras/Documents/Dropbox/Bib_files/biblio.bib}

%\bibliographystyle{unsrt}
%\bibliography{/Users/giuli/OneDrive/All-Matrix-Elements/biblio.bib}

\end{document}